\pdfoutput=1
\documentclass[11pt,a4paper,final]{article}

\usepackage[utf8]{inputenc}
\usepackage{amsmath}
\usepackage{amsfonts}
\usepackage{amssymb}
\usepackage{graphicx}
\usepackage{array}
\usepackage{multirow}
\usepackage{multicol}
\usepackage{siunitx}

\usepackage{float}
\usepackage{natbib}
\floatstyle{plaintop}
\restylefloat{table}

\usepackage[left=1in,right=1in,top=1in,bottom=1in]{geometry}

\usepackage{booktabs} % For formal tables
\usepackage[ruled]{algorithm2e} % For algorithms
\usepackage{tikz}
\usetikzlibrary{arrows}
\usetikzlibrary{shapes.multipart}
\usepackage{caption}

\SetAlFnt{\small}
\SetAlCapFnt{\small}
\SetAlCapNameFnt{\small}
\SetAlCapHSkip{0pt}
\IncMargin{-\parindent}
\usepackage[T1]{fontenc}    % use 8-bit T1 fonts
\usepackage{hyperref}       % hyperlinks hyperfootnotes=false

\usepackage{xcolor}

\newcommand{\declarecolor}[2]{\definecolor{#1}{RGB}{#2}\expandafter\newcommand\csname #1\endcsname[1]{\textcolor{#1}{##1}}}

\declarecolor{White}{255, 255, 255}
\declarecolor{Black}{0, 0, 0}
\declarecolor{Maroon}{128, 0, 0}
\declarecolor{Coral}{255, 127, 80}
\declarecolor{Red}{182, 21, 21}
\declarecolor{LimeGreen}{50, 205, 50}
\declarecolor{DarkGreen}{0, 90, 0}
\declarecolor{Navy}{0, 0, 128}

\definecolor{plotblue}{HTML}{377eb8}
\definecolor{plotorange}{HTML}{ff7f00}
\definecolor{plotgreen}{HTML}{4daf4a}

\hypersetup{
    colorlinks=true,
    citecolor=DarkGreen,
    linkcolor=Navy,
    filecolor=magenta,      
    urlcolor=black,
    pdfpagemode=FullScreen,
}
\usepackage{url}            % simple URL typesetting

\usepackage{complexity}

\let\A\relax

\let\E\relax

\let\R\relax

\usepackage{booktabs}       % professional-quality tables
\usepackage{nicefrac}       % compact symbols for 1/2, etc.
\usepackage{microtype}      % microtypography
\usepackage{mathtools}
\usepackage{amsthm}
\usepackage{bbm}
\usepackage{tikz} 
\usetikzlibrary{calc}
\usepackage{wrapfig}
\usepackage{enumitem}
\usepackage{authblk}
\usepackage{pifont}

\usepackage[suppress]{color-edits}
\addauthor{kh}{red}

\usepackage{caption}
\usepackage{subcaption}

\usepackage[capitalize,noabbrev]{cleveref}

\newtheorem{theorem}{Theorem}[section]

\newtheorem{fact}[theorem]{Fact}
\newtheorem{condition}[theorem]{Condition}
\newtheorem{property}[theorem]{Property}
\newtheorem{claim}[theorem]{Claim}
\newtheorem{remark}[theorem]{Remark}
\newtheorem{corollary}[theorem]{Corollary}
\newtheorem{proposition}[theorem]{Proposition}
\newtheorem{assumption}[theorem]{Assumption}

\newtheorem{definition}[theorem]{Definition}

\crefname{assumption}{Assumption}{assumptions}

\usepackage{bm}

\newcommand{\SW}{\textsc{SW}}
\newcommand{\opt}{\textsc{OPT}}

\renewcommand{\vec}[1]{\bm{#1}}

\newcommand{\va}{\vec{a}}

\newcommand{\vu}{\vec{u}}
\newcommand{\vx}{\vec{x}}

\newcommand{\hvu}{\hat{\vec{u}}}
\newcommand{\vm}{\vec{m}}
\newcommand{\ub}{h}
\newcommand{\diff}{\Delta}
\newcommand{\phimax}{\Phi_{\textrm{max}}}
\newcommand{\lb}{\ell}
\newcommand{\smax}{S_{\textrm{max}}}
\newcommand{\spnorm}{\sigma}
\newcommand{\vy}{\vec{y}}
\newcommand{\rvx}{\mathring{\vx}}
\newcommand{\rvy}{\mathring{\vy}}
\newcommand{\bvx}{\Bar{\vx}}
\newcommand{\bvy}{\Bar{\vy}}
\newcommand{\hvx}{\hat{\vec{x}}}
\newcommand{\hvy}{\hat{\vec{y}}}
\newcommand{\tvx}{\Tilde{\vec{x}}}
\newcommand{\rtvx}{\Tilde{\mathring{\vec{x}}}}

\newcommand{\rtvy}{\Tilde{\mathring{\vec{y}}}}

\newcommand{\cR}{\mathcal{R}}

\newcommand{\game}{\mathcal{G}}

\DeclarePairedDelimiterX{\infdivx}[2]{(}{)}{%
  #1\;\delimsize\|\;#2%
}
\newcommand{\brg}[1]{\mathcal{B}_{\cR_{#1}}\infdivx}
\newcommand{\range}[1]{[\![#1]\!]}

\newcommand{\kl}{D_{\textrm{KL}}\infdivx}

\DeclareMathOperator{\nereg}{NE-Reg}
\DeclareMathOperator{\reg}{Reg}
\DeclareMathOperator{\swreg}{SwapReg}
\DeclareMathOperator{\hreg}{\hat{R}eg}
\newcommand{\lreg}[1]{\textrm{Reg}^{(#1)}_{\mathcal{L}}}
\DeclareMathOperator{\areg}{\vec{\alpha}-Reg}
\DeclareMathOperator{\sreg}{StackReg}
\DeclareMathOperator{\proj}{Proj}

\newcommand{\preco}{\mat{Q}}

\newcommand{\adagrad}{\texttt{AdaGrad}}
\newcommand{\optada}{\mathtt{OptAdaGrad}}

\usepackage[utf8]{inputenc}

\usepackage{amsmath,amssymb,amsthm}

\DeclareMathOperator*{\argmin}{arg\,min}
\DeclareMathOperator*{\argmax}{arg\,max}

\newcommand{\diam}{\Omega}

\newcommand{\OMD}{\texttt{OMD}}

\newcommand{\EWOO}{\texttt{EWOO}}
\newcommand{\OFTRL}{\texttt{OFTRL}}
\newcommand{\FTRL}{\texttt{FTRL}}
\newcommand{\OGD}{\texttt{OGD}}
\newcommand{\GD}{\texttt{GD}}
\newcommand{\opthedge}{\texttt{OptHedge}}
\newcommand{\EG}{\texttt{EG}}
\newcommand{\FTL}{\texttt{FTL}}

\NewDocumentCommand{\pv}{m e{_} m}{%
  #1\IfValueT{#2}{_{#2}}^{(#3)}%
}

\newcommand{\defeq}{\coloneqq}

\newcommand{\A}{\mathcal A}

\newcommand{\E}{\mathbb E}

\newcommand{\R}{\mathbb R}

\newcommand{\X}{\mathcal X}
\newcommand{\Y}{\mathcal Y}

\newcommand{\N}{\mathbb N}
\newcommand{\cE}{\mathcal{E}}

\newcommand{\calR}{\mathcal{R}}
\newcommand{\calZ}{\mathcal{Z}}
\newcommand{\calA}{\mathcal{A}}

\newcommand{\mat}{\mathbf}
\newcommand{\dx}{d_x}
\newcommand{\dy}{d_y}

\newcommand{\vz}{\vec{z}}
\newcommand{\hvz}{\hat{\vec{z}}}

\newcommand{\hinsim}{V}
\newcommand{\diffsim}{V_{\diff}}
\newcommand{\NEsim}{V_{\textrm{NE}}}
\newcommand{\MVIsim}{V_{\textrm{MVI}}}

\newcommand{\xstar}{\vx^{\star}}
\newcommand{\zstar}{\vz^{\star}}

\newcommand{\astar}{\va^{\star}}
\newcommand{\ystar}{\vy^{\star}}

\newcommand{\dualgap}{\textsc{DualGap}}

\title{Meta-Learning in Games\thanks{A version of this paper was published in ICLR 2023 (11th International Conf. on Learning Representations).}}

\newcommand{\CMU}{Carnegie Mellon University}

\author[1]{Keegan Harris\footnote{Equal contribution.}}
\author[1]{Ioannis Anagnostides\textsuperscript{\textdagger}}
\author[2]{Gabriele Farina}
\author[1]{Mikhail Khodak}
\author[1]{Zhiwei Steven Wu}
\author[1,3,4,5]{Tuomas Sandholm}

\affil[1]{Carnegie Mellon University}
\affil[2]{FAIR, Meta AI}
\affil[3]{Strategy Robot, Inc.}
\affil[4]{Optimized Markets, Inc.}
\affil[5]{Strategic Machine, Inc.}
\affil[ ]{\texttt {\{keeganh,ianagnos,mkhodak,zhiweiw,sandholm\}@cs.cmu.edu}}
\affil[ ]{\texttt{gfarina@meta.com}}

\begin{document}
\date{}
\maketitle
%\vspace*{-0.4in}

\pagenumbering{gobble}
\begin{abstract}
In the literature on game-theoretic equilibrium finding, focus has mainly been on solving a single game in isolation. In practice, however, strategic interactions---ranging from routing problems to online advertising auctions---evolve dynamically, thereby leading to many similar games to be solved. %One use-case is the setting where one wants to find equilibria for multiple game variations. Another use-case is strategic interactions that evolve dynamically, ranging from routing problems to online advertising auctions. 
To address this gap, we introduce \textit{meta-learning} for  equilibrium finding and learning to play games. We establish the first meta-learning guarantees for a variety of fundamental and well-studied classes of games, including two-player zero-sum games, general-sum games, and Stackelberg games. In particular, we obtain rates of convergence to different game-theoretic equilibria that depend on natural notions of similarity between the sequence of games encountered, while at the same time recovering the known single-game guarantees when the sequence of games is arbitrary. Along the way, we prove a number of new results in the single-game regime through a simple and unified framework, which may be of independent interest. Finally, we evaluate our meta-learning algorithms on endgames faced by the poker agent~\emph{Libratus} against top human professionals. The experiments show that games with varying stack sizes can be solved significantly faster using our meta-learning techniques than by solving them separately, often by an order of magnitude.\looseness-1
\end{abstract}

% The \author macro works with any number of authors. There are two commands
% used to separate the names and addresses of multiple authors: \And and \AND.
%
% Using \And between authors leaves it to \LaTeX{} to determine where to break
% the lines. Using \AND forces a linebreak at that point. So, if \LaTeX{}
% puts 3 of 4 authors names on the first line, and the last on the second
% line, try using \AND instead of \And before the third author name.
\newpage
\tableofcontents
\newpage
\pagenumbering{arabic}
\setlength{\belowcaptionskip}{-10pt}
\section{Introduction}
Research on game-theoretic equilibrium computation has primarily focused on solving a single game in isolation. In practice, however, there are often many similar games which need to be solved. One use-case is the setting where one wants to find an equilibrium for each of  multiple game variations---for example poker games where the players have various sizes of chip stacks. Another use-case is strategic interactions that evolve dynamically: in online advertising auctions, the advertiser's value for different keywords adapts based on current marketing trends~\citep{Nekipelov15:Econometrics}; routing games---be it Internet routing or physical transportation---reshape depending on the topology and the cost functions of the underlying network~\citep{Hoefer11:Competitive}; and resource allocation problems~\citep{Ramesh04:Efficiency} vary based on the values of the goods/services.  Successful agents in such complex decentralized environments must effectively learn how to incorporate past experience from previous strategic interactions in order to adapt their behavior to the current and future tasks.\looseness-1

\emph{Meta-learning}, or \emph{learning-to-learn}~\citep{Thrun12:Learning}, is a common formalization for machine learning in dynamic single-agent environments. 
In the meta-learning framework, a learning agent faces a sequence of tasks, and the goal is to use knowledge gained from previous tasks in order to improve performance on the current task at hand. Despite rapid progress in this line of work, prior results have not been tailored to tackle multiagent settings. This begs the question: \emph{Can players obtain provable performance improvements when meta-learning across a sequence of games?} We answer this question in the affirmative by introducing meta-learning for equilibrium finding and learning to play games, and providing the first performance guarantees in a number of fundamental multiagent settings.\looseness-1%, where the performance of our algorithms depends on a natural notion of similarity between the games encountered.

%and significant gains under natural notions of similarity for the sequence of games in a number of fundamental multiagent settings when the . .

%\emph{Meta-learning}, or \emph{learning-to-learn}~\citep{Thrun12:Learning}, is a rapidly growing and general framework in machine learning, designed to formulate and address such questions. In the meta-learning framework, the learning agent is assumed to face a sequence of tasks, and the fundamental question is how to leverage prior knowledge to improve the overall performance of the agent. However, despite the rapid progress on that line of work in recent years, prior results have not been tailored to tackle decentralized multiagent settings; this begs the question: \emph{Can we obtain provable improvements when decentralized meta-learners are competing in a sequence of games?} In this paper, we answer this question in the positive, providing the first provable and significant gains under natural notions of similarity for the sequence of games in a number of fundamental multiagent settings.
\subsection{Overview of Our Results}

Our main contribution is to develop a general framework for establishing the first provable guarantees for meta-learning in games, leading to a comprehensive set of results in a variety of well-studied multiagent settings. In particular, our results encompass environments ranging from two-player zero-sum games with general constraint sets (and multiple extensions thereof), to general-sum games and Stackelberg games. See \Cref{tab:results} for a summary of our results. Our refined guarantees are parameterized based on natural similarity metrics between the sequence of games. For example, in zero-sum games we obtain last-iterate rates that depend on the variance of the Nash equilibria (\Cref{theorem:informal-last}); in potential games based on the deviation of the potential functions (\Cref{theorem:informal-pot}); and in Stackelberg games our regret bounds depend on the similarity of the leader's optimal commitment in hindsight (\Cref{theorem:informal-Stack}). All of these measures are algorithm-independent, and tie naturally to the underlying game-theoretic solution concepts.\looseness-1

Importantly, our algorithms are agnostic to how similar the games are, but are nonetheless specifically designed to adapt to the similarity. Our guarantees apply under a broad class of no-regret learning algorithms, such as \emph{optimistic mirror descent ($\OMD$)}~\citep{Chiang12:Online,Rakhlin13:Optimization}, with the important twist that each player employs an additional regret minimizer for meta-learning the parameterization of the base-learner; the latter component builds on the meta-learning framework of \citet{Khodak19:Adaptive}. For example, in zero-sum games we leverage an initialization-dependent \emph{RVU bound}~\citep{Syrgkanis15:Fast} in order to meta-learn the initialization of $\OMD$ across the sequences of games, leading to per-game convergence rates to Nash equilibria that closely match our refined lower bound (\Cref{theorem:informal-lb}). 
%
%In particular, we obtain many of our results by first deriving an initialization-dependent \emph{Regret bounded
%by Variation in Utilities (RVU)} bound (or RVU-like bound) for our problem of interest \citep{Syrgkanis15:Fast}, followed by meta-learning the initialization (and sometimes other parameters) of an optimistic no-regret learning algorithm using 
%
More broadly, in the worst-case---\emph{i.e.}, when the sequence of games is arbitrary---we recover the near-optimal guarantees known for static games, but as the similarity metrics become more favorable we establish significant gains in terms of convergence to different notions of equilibria.\looseness-1 

Along the way, we also obtain new insights and results even from a single-game perspective, including convergence rates of $\OMD$ and the \emph{extra-gradient method} in H\"older continuous variational inequalities~\citep{Rakhlin13:Online}, and certain nonconvex-nonconcave problems such as those considered by~\citep{Diakonikolas21:Efficient} and stochastic games. Further, our analysis is considerably simpler than prior techniques
and unifies several prior results. Finally, in~\Cref{sec:exp} we evaluate our techniques on a series of poker endgames faced by the poker agent \emph{Libratus}~\citep{Brown18:Superhuman} against top human professionals. The experiments show that our meta-learning algorithms offer significant gains compared to solving each game in isolation, often by an order of magnitude.\looseness-1

\begin{table}[h]
    \caption{A summary of our key theoretical results on meta-learning in games.}
    \vspace{-2pt}
    \small
    \centering
    \scalebox{.95}{\begin{tabular}{llll}
        % \toprule
        \bf Class of games & \bf Specific problem & \bf Key results & \bf Location
         \\ \toprule
        & Bilinear saddle-point problems & \Cref{theorem:informal-dualgap,theorem:informal-last} & \Cref{sec:zero-sum} \\
         Zero-sum games & H\"older continuous VIs & \Cref{theorem:last-MVI,theorem:rate-Holder} & \Cref{sec:MVI,appendix:Holder} \\
         & Lower bound & \Cref{theorem:informal-lb} & \Cref{sec:zero-sum} \\ \midrule
        & Potential games & \Cref{theorem:informal-pot} & \Cref{sec:general-sum} \\ 
        General-sum games & (Coarse) Correlated equilibria & \Cref{theorem:cce,theorem:ce} & \Cref{appendix:cce,appendix:CE}  \\
        & Approximately optimal welfare & \Cref{theorem:sw} & \Cref{sec:general-sum} \\ \midrule
         Stackelberg games & Security games & \Cref{theorem:informal-Stack} & \Cref{sec:stack} \\ \bottomrule
    \end{tabular}}
    \label{tab:results}
\end{table}

\subsection{Related Work}
\label{sec:rel}

The study of online learning algorithms in games has engendered a prolific area of research, tracing back to the pioneering works of~\citet{Robinson22:An} and~\citet{Blackwell65:An}.  
While traditional analyses rely on black-box guarantees from the no-regret framework~\citep{Cesa-bianchi2006:Prediction}, recent works have established exponential improvements over those guarantees when specific learning dynamics are in place (\emph{e.g.},~\citep{Daskalakis15:Near,Syrgkanis15:Fast,Daskalakis21:Near}).\looseness-1

However, that line of work posits that the underlying game remains invariant. Yet, there is ample motivation for studying games that gradually change over time, such as online advertising~\citep{Nekipelov15:Econometrics,Lykouris16:Learning,Nisan17:An} or congestion games~\citep{Hoefer11:Competitive,Bertrand20:Dynamic,Meigs17:Learning}. Another related direction consists of \emph{warm starting} for solving zero-sum extensive-form games~\citep{Brown16:Strategy}, which is typically employed in conjunction with abstraction-based techniques~\citep{Brown14:Regret,Brown15:Regret,Kroer18:A,Brown15:Simultaneous}. Specifically, there one constructs a sequence of progressively finer abstractions for an underlying game, so that the equilibria of each game can assist the solution of the next one. Perhaps the cardinal difference with our setting is that in abstraction-based applications one is only interested in the performance in the ultimate game.

%While most prior work on learning in games posits that the underlying game remains invariant, there is ample motivation for studying games that gradually change over time, such as online advertising~\citep{Nekipelov15:Econometrics,Lykouris16:Learning,Nisan17:An} or congestion games~\citep{Hoefer11:Competitive,Bertrand20:Dynamic,Meigs17:Learning}. Indeed, a number of prior works study the performance of learning algorithms in time-varying zero-sum games~\citep{Zhang22:No,Fiez21:Online,Ducovelle22:multiagent,Cardoso19:Competing}; there, it is natural to espouse dynamic notions of regret~\citep{Yang16:Tracking,Zhao20:Dynamics}. A work closely related to ours is the recent paper by~\citet{Zhang22:No}, which provides regret bounds in time-varying bilinear saddle-point problems parameterized by the similarity of the payoff matrices and the equilibria of those games. In contrast to our meta-learning setup, they study a more general setting in which the game can change arbitrarily from round-to-round. While our problem can be viewed a special type of a time-varying game in which the boundaries between different games are fixed and known, algorithms designed for generic time-varying games will not perform as well in our setting, as they do not utilize this extra information. As a result, we view these results as complementary to ours. For a more detailed discussion, see~\Cref{appendix:comparison}.\looseness-1

An emerging paradigm for modeling such considerations is meta-learning, which has gained increasing popularity in the machine learning community in recent years; for a highly incomplete set of pointers, we refer to~\citep{Balcan15:Efficient,Al-Shedivat18:Continuous,Finn17:Model,Finn19:Online,Balcan19:Provable,Li17:Meta,Chen22:Memory}, and references therein. Our work constitutes the natural coalescence of meta-learning with the line of work on (decentralized) online learning in games. Although, as we pointed out earlier, learning in dynamic games has already received considerable attention, we are the first (to our knowledge) to formulate and address such questions within the meta-learning framework; \emph{c.f.}, see~\citep{Kayaalp20:Dif,Kayaalp21:Distributed,Li22:Learning}. 
Finally, our methods may be viewed within the \emph{algorithms with predictions} paradigm~\citep{Mitzenmacher20:Algorithms}:
we speed up equilibrium computation by learning to predict equilibria across multiple games, with the task-similarity being the measure of prediction quality.
For further related work, see~\Cref{appendix:comparison}.\looseness-1
\section{Our Setup: Meta-Learning in Games}
\label{sec:background}

%In this section, we introduce our meta-learning in games setup, along with our notation and some basic background on online optimization.
%We now introduce our meta-learning in games setup, along with our notation and some background.\looseness-1

\paragraph{Notation} We use boldface symbols to represent vectors and matrices. Subscripts are typically reserved to indicate the player, while superscripts usually correspond to the iteration or the index of the task. We let $\N \defeq \{1, 2, \dots, \}$ be the set of natural numbers. For $T \in \N$, we use the shorthand notation $\range{T} \defeq \{1, 2, \dots, T\}$. For a nonempty convex and compact set $\X$, we denote by $\diam_{\X}$ its $\ell_2$-diameter: $\diam_{\X} \defeq \max_{\vx, \vx' \in \X} \|\vx - \vx'\|_2$. Finally, to lighten the exposition we use the $O(\cdot)$ notation to suppress factors that depend polynomially on the natural parameters of the problem.\looseness-1

\paragraph{The general setup} We consider a setting wherein players interact in a sequence of $T$ repeated games (or \emph{tasks}), for some $\N \ni T \gg 1$. Each task itself consists of $m \in \N$ iterations. Any fixed task $t$ corresponds to a multiplayer game $\game^{(t)}$ between a set $\range{n}$ of players, with $n \geq 2$; it is assumed for simplicity in the exposition that $n$ remains invariant across the games, but some of our results apply more broadly. Each player $k \in \range{n}$ selects a strategy $\vx_k$ from a convex and compact set of strategies $\X_k \subseteq \R^{d_k}$ with nonempty relative interior. For a given joint strategy profile $\vx \defeq (\vx_1, \dots, \vx_n) \in \bigtimes_{k=1}^n \X_k$, there is a multilinear utility function $u_k : \vx \mapsto \langle \vx_k, \vu_k(\vx_{-k}) \rangle$ for each player $k$, where $\vx_{-k} \defeq (\vx_1, \dots, \vx_{k-1}, \vx_{k+1}, \dots, \vx_n)$. We will also let $L > 0$ be a Lipschitz parameter of each game, in the sense that for any player $k \in \range{n}$ and any two strategy profiles $\vx_{-k}, \vx_{-k}' \in \bigtimes_{k' \neq k} \X_{k'}$,\looseness-1
\begin{equation}
    \label{eq:Lip}
    \|\vu_k(\vx_{-k}) - \vu_k(\vx_{-k}') \|_2 \leq L \|\vx_{-k} - \vx_{-k}'\|_2.
\end{equation}
Here, we use the $\ell_2$-norm for convenience in the analysis; \eqref{eq:Lip} can be translated to any equivalent norm. Finally, for a joint strategy profile $\vx \in \bigtimes_{k = 1}^n \X_k$, the \emph{social welfare} is defined as $\SW(\vx) \defeq \sum_{k=1}^n u_k(\vx)$, so that $\opt \defeq \max_{\vx \in \bigtimes_{k = 1}^n \X_k} \SW(\vx)$ denotes the optimal social welfare.\looseness-1

A concrete example encompassed by our setup is that of \emph{extensive-form games}. More broadly, it captures general games with concave utilities~\citep{Rosen65:Existence,Hsieh21:Adaptive}. %Some of our results apply even beyond that broad setting, as it will be formalized in the sequel.\looseness-1

\paragraph{Online learning in games} Learning proceeds in an online fashion as follows. At every iteration $i \in \range{m}$ of some underlying game $t$, each player $k \in \range{n}$ has to select a strategy $\vx_k^{(t,i)} \in \X_k$. Then, in the full information setting, the player observes as feedback the utility corresponding to the other players' strategies at iteration $i$; namely, $\vu_k^{(t, i)} \defeq \vu_k(\vx_{-k}^{(t,i)})\in\R^{d_k}$. For convenience, we will assume that $\|\vu_k(\vx_{-k}^{(t,i)})\|_\infty \leq 1$.
The canonical measure of performance in online learning is that of \emph{external regret}, comparing the performance of the learner with that of the optimal fixed strategy in hindsight:\looseness-1 
%At time $i$ in task $t$, the player takes action $\pv{\vx}{t,i} \in \Delta^d$ and receives feedback $\pv{\vu}{t,i} \in [-1,1]^d$. We denote elements of $\pv{\vx}{t,i}$ ($\pv{\vu}{t,i}$) as $\pv{\vx}{t,i}[a]$ ($\pv{\vu}{t,i}[a]$).
%
\begin{definition}[Regret]\label{def:regret}
Fix a player $k \in \range{n}$ and some game $t \in \range{T}$. The (external) regret of player $k$ is defined as\looseness-1
\begin{equation*}
    \pv{\reg}{t,m}_k \defeq \max_{\pv{\rvx}{t}_k \in \X_k} \left\{ \sum_{i=1}^m \langle \pv{\rvx}{t}_k, \pv{\vu}{t,i}_k \rangle \right\} - \langle \pv{\vx}{t,i}_k, \pv{\vu}{t,i}_k \rangle.
\end{equation*}
%where $\pv{\rvx}{t}_k$ is the optimum-in-hindsight strategy for player $k$ in game $t$.
\end{definition}
We will let $\pv{\rvx}{t}_k$ be an optimum-in-hindsight strategy for player $k$ in game $t$; ties are broken arbitrarily, but according to a fixed rule (\emph{e.g.}, lexicographically). In the meta-learning setting, our goal will be to optimize the average performance---typically measured in terms of convergence to different game-theoretic equilibria---across the sequence of games.\looseness-1

\paragraph{Optimistic mirror descent} Suppose that $\cR_k : \X_k \to \R$ is a $1$-strongly convex regularizer with respect to a norm $\|\cdot\|$. We let $\brg{k}{\vx_k}{\vx'_k} \defeq \cR_k(\vx_k) - \cR_k(\vx'_k) - \langle \nabla \cR_k(\vx'_k), \vx_k - \vx'_k \rangle$ denote the \emph{Bregman divergence} induced by $\cR_k$, where  $\vx_k'$ is in the relative interior of $\X_k$. \emph{Optimistic mirror descent ($\OMD$)}~\citep{Chiang12:Online,Rakhlin13:Optimization} is parameterized by a prediction $\pv{\vm}_k{t,i} \in \R^{d_k}$ and a learning rate $\eta > 0$, and is defined at every iteration $i \in \N$ as follows.\looseness-1
\begin{equation*}
\begin{aligned}
    \pv{\vx}_k{t,i} &:= \arg \max_{\vx_k \in \mathcal{X}_k} \left\{ \langle \vx_k, \pv{\vm}{t,i}_k \rangle - \frac{1}{\eta} \brg{k}{\vx_k}{\pv{\hvx}_k{t,i-1}} \right\}, \\
    \pv{\hvx}_k{t,i} &:= \arg \max_{\hvx_k \in \mathcal{X}_k} \left\{ \langle \hvx_k, \pv{\vu_k}{t,i} \rangle - \frac{1}{\eta} \brg{k}{\hvx_k}{\pv{\hvx_k}{t,i-1}} \rangle \right\}.
\end{aligned}
\end{equation*}
Further, $\hvx_k^{(1,0)} \defeq \argmin_{\hvx_k \in \X_k} \cR_k(\hvx_k) \eqqcolon \vx_k^{(1,0)}$, and $\vm_k^{(t,1)} \defeq \vu_k(\vx_{-k}^{(t,0)})$. Under Euclidean regularization, $\cR_k(\vx_k) \defeq \frac{1}{2} \|\vx_k\|_2^2$, we will refer to $\OMD$ as \emph{optimistic gradient descent (\OGD)}.\looseness-1
\section{Meta-Learning How to Play Games}
\label{sec:main}

In this section, we present our main theoretical results: provable guarantees for online and decentralized meta-learning in games. We commence with zero-sum games in \Cref{sec:zero-sum}, and we then transition to general-sum games (\Cref{sec:general-sum}) and Stackelberg (security) games (\Cref{sec:stack}).

\subsection{Zero-Sum Games}
\label{sec:zero-sum}

We first highlight our results for bilinear saddle-point problems (BSPPs), which take the form $\min_{\vx \in \X} \max_{\vy \in \Y} \vx^\top \mat{A} \vy$, where $\mat{A} \in \R^{d_x \times d_y}$ is the payoff matrix of the game. A canonical application for this setting is on the solution of zero-sum imperfect-information extensive-form games~\citep{Romanovskii62:Reduction,Koller92:The}, as we explore in our experiments (\Cref{sec:exp}). Next we describe a number of extensions to gradually more general settings, and we conclude with our lower bound (\Cref{theorem:informal-lb}). The proofs from this subsection are included in \Cref{appendix:zero-sum}.\looseness-1

We first derive a refined meta-learning convergence guarantee for the average of the players' strategies. Below, we denote by $\hinsim^2_x \defeq \frac{1}{T} \min_{\vx \in \X} \sum_{t=1}^T \|\rvx^{(t)} - \vx\|_2^2$ the task similarity metric for player $x$, written in terms of the optimum-in-hindsight strategies; analogous notation is used for player $y$.\looseness-1

\begin{theorem}[Informal; Detailed Version in \Cref{cor:sum-zero-sum}]
    \label{theorem:informal-dualgap}
    Suppose that both players employ $\OGD$ with a suitable (fixed) learning rate and follow the leader over previous optimum-in-hindsight strategies for the initialization. %throughout the sequence of $T$ BSPPs. 
    Then, the game-average duality gap of the players' average strategies is bounded by
    \begin{equation}
        \label{eq:meta-dualgap}
        \frac{1}{T} \sum_{t=1}^T \frac{1}{m} \left( \reg_x^{(t,m)} + \reg_y^{(t,m)} \right) \leq \frac{2L}{m} \left( \hinsim_x^2 + \hinsim_y^2 \right) + \frac{8L(1 + \log T)}{m T} \left( \Omega^2_{\X} + \Omega^2_{\Y} \right).
    \end{equation}
\end{theorem}

Here, the second term in the right-hand side of~\eqref{eq:meta-dualgap} becomes negligible for a large number of games $T$, while the first term depends on the task similarity measures. For any sequence of games, \Cref{theorem:informal-dualgap} nearly matches the lower bound in the single-task setting~\citep{Daskalakis15:Near}, but our guarantee can be significantly better when $\hinsim_x^2, \hinsim_y^2 \ll 1$. To achieve this, the basic idea is to use---on top of $\OGD$---a ``meta'' regret minimization algorithm that, for each player, learns a sequence of initializations by taking the average of the past optima-in-hindsight, which is equivalent to \emph{follow the leader ($\FTL$)} over the regret upper-bounds of the within-task algorithm; see~\Cref{alg:meta-ogd} (in \Cref{appendix:sw}) for pseudocode of the meta-version of $\OGD$ we consider. Similar results can be obtained more broadly for $\OMD$ (\emph{c.f.}, see \Cref{appendix:cce,appendix:CE}). We also obtain analogous refined bounds for the \emph{individual} regret of each player (\Cref{cor:ind-reg}).\looseness-1

%We focus on $\OGD$ for ease of exposition and its desirable last-iterate properties. 
%Similar results can be obtained for other variants of $\OMD$ commonly used in games, e.g., in~\Cref{appendix:cce} we meta-learn the initialization of \emph{optimistic hedge} to obtain results for general-sum games.

One caveat of \Cref{theorem:informal-dualgap} is that the underlying task similarity measure could be algorithm-dependent, as the optimum-in-hindsight for each player could depend on the other player's behavior. To address this, we show that if the meta-learner can initialize using \emph{Nash equilibria (NE)} (recall \Cref{def:NE}) from previously seen games, the game-average last-iterate rates gracefully decrease with the similarity of the Nash equilibria of those games. More precisely, if $\vz \defeq (\vx, \vy) \in \X \times \Y \eqqcolon \calZ$, we let $\NEsim^2 \defeq \frac{1}{T} \max_{\vz^{(1, \star)}, \dots, \vz^{(T, \star)}} \min_{\vz \in \calZ} \sum_{t=1}^T \|\vz^{(t, \star)} - \vz \|_2^2$, where $\vz^{(t, \star)}$ is any Nash equilibrium of the $t$-th game. As we point out in the sequel, we also obtain results under a more favorable notion of task similarity that does not depend on the worst sequence of NE.\looseness-1

\begin{theorem}[Informal; Detailed Version in \Cref{theorem:last-worst}]
    \label{theorem:informal-last}
    When both players employ $\OGD$ with a suitable (fixed) learning rate and $\FTL$ over previous NE strategies for the initialization, then
    \begin{equation*}
        \Bar{m} \leq \frac{2 \NEsim^2}{\epsilon^2} + \frac{8(1 + \log T)}{T \epsilon^2} \left( \diam^2_{\X} + \diam^2_{\Y} \right)
    \end{equation*}
    iterations suffice to reach an $O(\epsilon)$-approximate Nash equilibrium on average across the $T$ games.
\end{theorem}
\Cref{theorem:informal-last} recovers the optimal $m^{-1/2}$ rates for $\OGD$~\citep{Golowich20:Tight,Golowich20:Last} under an arbitrary sequence of games, but offers substantial gains in terms of the average iteration complexity when the Nash equilibria of the games are close. For example, when they lie within a ball of $\ell_2$-diameter $\sqrt{\delta (\diam^2_{\X} + \diam^2_{\Y})}$, for some $\delta \in (0, 1]$, \Cref{theorem:informal-last} improves upon the rate of $\OGD$ by at least a multiplicative factor of $1 / \delta$ as $T \rightarrow \infty$. While \emph{generic}---roughly speaking, randomly perturbed---zero-sum (normal-form) games have a unique Nash equilibrium~\citep{VanDamme87:Stability}, the worst-case NE similarity metric used in \Cref{theorem:informal-last} can be loose under multiplicity of equilibria. For that reason, in \Cref{sec:improved-task} we further refine \Cref{theorem:informal-last} using the most favorable sequence of Nash equilibria; this requires that players know each game after its termination, which is arguably a well-motivated assumption in some applications. We further remark that \Cref{theorem:informal-last} can be cast in terms of the similarity $\hinsim_x^2 + \hinsim_y^2$, instead of $\NEsim^2$, using the parameterization of~\Cref{theorem:informal-dualgap}. 
Finally, since the base-learner can be viewed as an algorithm with predictions---the number of iterations to compute an approximate NE is smaller if the Euclidean error of a prediction of it (the initialization) is small---\Cref{theorem:informal-last} can also be viewed as {\em learning} these predictions~\citep{khodak2022awp} by targeting that error measure.\looseness-1

\paragraph{Extensions} Moving beyond bilinear saddle-point problems, we extend our results to gradually broader settings. %Along the way, we obtain new insights and connections even from a single-game perspective. 
First, in \Cref{sec:MVI} we apply our techniques to general variational inequality problems under a Lipschitz continuous operator for which the so-called \emph{MVI property}~\citep{Mertikopoulos19:Optimistic} holds. Thus, \Cref{theorem:informal-dualgap,theorem:informal-last} are extended to settings such as smooth convex-concave games and zero-sum polymatrix (multiplayer) games~\citep{Cai16:A}. Interestingly, extensions are possible even under the \emph{weak MVI property}~\citep{Diakonikolas21:Efficient}, which captures certain ``structured'' nonconvex-nonconcave games. In a similar vein, we also study the challenging setting of Shapley's stochastic games~\citep{Shapley53:Stochastic} (\Cref{sec:stochastic}). There, we show that there exists a time-varying---instead of constant---but non-vanishing learning rate schedule for which $\OGD$ reaches minimax equilibria, thereby leading to similar extensions in the meta-learning setting. Next, we relax the underlying Lipschitz continuity assumption underpinning the previous results by instead imposing only $\alpha$-H\"older continuity (recall \Cref{def:Holder}). We show that in such settings $\OGD$ enjoys a rate of $m^{- \alpha/2}$ (\Cref{theorem:rate-Holder}), which is to our knowledge a new result; in the special case where $\alpha = 1$, we recover the recently established $m^{-1/2}$ rates. Finally, while we have focused on the $\OGD$ algorithm, our techniques apply to other learning dynamics as well. For example, in \Cref{appendix:EG} we show that the extensively studied extra-gradient $(\EG)$ algorithm~\citep{Korpelevich76:Extragradient} can be analyzed in a unifying way with $\OGD$, thereby inheriting all of the aforementioned results under $\OGD$; this significantly broadens the implications of~\citep{Mokhtari20:A}, which only applied in certain unconstrained problems. Perhaps surprisingly, although $\EG$ is \emph{not} a no-regret algorithm, our analysis employs a regret-based framework using a suitable proxy for the regret (see \Cref{theorem:rvu-EG}).\looseness-1

\paragraph{Lower bound} We conclude this subsection with a lower bound, showing that our guarantee in \Cref{theorem:informal-dualgap} is essentially sharp under a broad range of our similarity measures. Our result significantly refines the single-game lower bound of~\citet{Daskalakis15:Near} by constructing an appropriate distribution over sequences of zero-sum games.

\begin{theorem}[Informal; Precise Version in \Cref{theorem:lower-bound}]
    \label{theorem:informal-lb}
    For any $\epsilon > 0$, there exists a distribution over sequences of $T$ zero-sum games, with a sufficiently large $T = T(\epsilon)$, such that 
    \begin{equation*}
        \frac{1}{T} \sum_{t=1}^T \E[\reg_x^{(t,m)} + \reg_y^{(t,m)}] \geq \frac{1}{2} \left( \hinsim_x^2 + \hinsim_y^2 \right) -  \epsilon = \frac{1}{2} \NEsim^2 -  \epsilon.
    \end{equation*}
\end{theorem}

\subsection{General-Sum Games}
\label{sec:general-sum}

In this subsection, we switch our attention to general-sum games. Here, unlike zero-sum games, no-regret learning algorithms are instead known to generally converge---in a time-average sense---to \emph{correlated equilibrium} concepts, which are more permissive than the Nash equilibrium. Nevertheless, there are structured classes of general-sum games for which suitable dynamics do reach Nash equilibria; perhaps the most notable example being that of \emph{potential games}. In this context, we first obtain meta-learning guarantees for potential games, parameterized by the similarity of the potential functions. Then, we derive meta-learning algorithms with improved guarantees for convergence to correlated and \emph{coarse} correlated equilibria. Finally, we conclude this subsection with improved guarantees of convergence to near-optimal---in terms of social welfare---equilibria. Proofs from this subsection are included in \Cref{appendix:sw,appendix:general-sum}.\looseness-1

\paragraph{Potential games} A potential game is endowed with the additional property of admitting a potential: a player-independent function that captures the player's benefit from unilaterally deviating from any given strategy profile (\Cref{def:pot}). In our meta-learning setting, we posit a sequence of potential games $(\Phi^{(t)})_{1 \leq t \leq t}$, each described by its potential function. Unlike our approach in \Cref{sec:zero-sum}, a central challenge here is that the potential function is in general nonconcave/nonconvex, precluding standard regret minimization approaches. Instead, we find that by initializing at the previous last-iterate the dynamics still manage to adapt based on the similarity $\diffsim \defeq \frac{1}{T} \sum_{t=1}^{T-1} \diff(\Phi^{(t)}, \Phi^{(t+1)})$, where $\diff(\Phi, \Phi') \defeq \max_{\vx} (\Phi(\vx) - \Phi'(\vx))$, which captures the deviation of the potential functions. This initialization has the additional benefit of being agnostic to the boundaries of different tasks. Unlike our results in~\Cref{sec:zero-sum}, the following guarantee applies even for vanilla (\emph{i.e.}, non-optimistic) projected gradient descent (\GD).\looseness-1

\begin{theorem}[Informal; Detailed Version in \Cref{cor:pot-last}]
    \label{theorem:informal-pot}
    $\GD$ with suitable parameterization requires
    $O \left( \frac{\diffsim}{\epsilon^2} + \frac{\phimax}{\epsilon^2 T} \right)$
    iterations to reach an $\epsilon$-approximate Nash equilibrium on average across the $T$ potential games, where $\max_{\vx, t}|\pv{\Phi}{t}(\vx)| \leq \phimax$.
\end{theorem}

\Cref{theorem:informal-pot} matches the known rate of $\GD$ for potential games in the worst case, but offers substantial gains in terms of the average iteration complexity when the games are similar. For example, if $|\pv{\Phi}{t}(\vx) - \pv{\Phi}{t-1}(\vx)| \leq \alpha$, for all $\vx \in \bigtimes_{k=1}^n \X_k$ and $t \geq 2$, then $O(\alpha/\epsilon^2)$ iterations suffice to reach an $\epsilon$-approximate NE on an average game, as $T \to +\infty$. Such a scenario may arise in, \emph{e.g.}, a sequence of routing games if the cost functions for each edge change only slightly between games.\looseness-1

%they can be much better if similarities exist in the potential games encountered. For example, suppose the sequence of potentials $\pv{\Phi}{1}, \ldots, \pv{\Phi}{T}$ varies in such a way that is $\alpha$-Lipschitz in $t$, i.e., $|\pv{\Phi}{t}(\vx) - \pv{\Phi}{t-1}(\vx)| \leq \alpha, \forall \vx, t$. In such a setting, $O(\alpha / \epsilon^2) + o_T(\poly(1/ \epsilon, \phimax))$ iterations are sufficient to reach an $\epsilon$-approximate NE on average, a substantial improvement over the single-game rate if $\alpha$ is small.

\paragraph{Convergence to correlated equilibria} In contrast, for general games the best one can hope for is to obtain improved rates for convergence to correlated or coarse correlated equilibria~\citep{Hart00:Adaptive,blum2007external}. It is important to stress that learning correlated equilibria is fundamentally different than learning Nash equilibria---which are product distributions. For example, for the former any initialization---which is inevitably a product distribution in the case of uncoupled dynamics---could fail to exploit the learning in the previous task (\Cref{prop:separation}): unlike Nash equilibria, correlated equilibria (in general) cannot be decomposed for each player, thereby making uncoupled methods unlikely to adapt to the similarity of CE. Instead, our task similarity metrics depend on the optima-in-hindsight for each player. Under this notion of task similarity, we obtain task-average guarantees for CCE by meta-learning the initialization (by running $\FTL$) and the learning rate (by running the $\EWOO$ method of \citet{hazan2007logarithmic} over a sequence of regret upper bounds) of \emph{optimistic hedge}~\citep{Daskalakis21:Near} (\Cref{theorem:cce})---$\OMD$ with entropic regularization. Similarly, to obtain guarantees for CE, we use the \emph{no-swap-regret} construction of~\citet{blum2007external} in conjuction with the logarithmic barrier~\citep{anagnostides2022uncoupled} (\Cref{theorem:ce}).\looseness-1% and meta-learn the initialization (via $\FTL$ over previous optima-in-hindsight).

\subsubsection{Social Welfare Guarantees}
\label{sec:sw}

We conclude this subsection with meta-learning guarantees for converging to near-optimal equilibria (\Cref{theorem:sw}). Let us first recall the following central definition. 

\begin{definition}[Smooth games~\citep{Roughgarden15:Intrinsic}]
    \label{def:smooth}
    A game $\game$ is $(\lambda, \mu)$-smooth, with $\lambda, \mu > 0$, if there exists a strategy profile $\xstar \in \bigtimes_{k = 1}^n \X_k $ such that for any $\vx \in \bigtimes_{k = 1}^n \X_k$,
    \begin{equation}
        \label{eq:smooth}
        \sum_{k=1}^n u_k(\xstar_k, \vx_{-k}) \geq \lambda \opt - \mu \SW(\vx),
    \end{equation}
    where $\opt$ is the optimal social welfare and $\SW(\vx)$ is the social welfare of joint strategy profile $\vx$.
\end{definition}

Smooth games capture a number of important applications, including network congestion games~\citep{Awerbuch13:The,Christodoulou05:The} and simultaneous auctions~\citep{Christodoulou16:Bayesian,Roughgarden17:The} (see \Cref{appendix:sw} for additional examples); both of those settings are oftentimes non-static in real-world applications, thereby motivating our meta-learning considerations. In this context, we assume that there is a sequence of smooth games $(\game^{(t)})_{1 \leq t \leq T}$, each of which is $(\lambda^{(t)}, \mu^{(t)})$-smooth (\Cref{def:smooth}).\looseness-1

\begin{theorem}[Informal; Detailed Version in \Cref{theorem:sw-formal}]
    \label{theorem:sw}
    If all players use $\OGD$ with suitable parameterization in a sequence of $T$ games $(\game^{(t)})_{1 \leq t \leq T}$, each of which is $(\lambda^{(t)}, \mu^{(t)})$-smooth, then
    \begin{equation}
        \label{eq:meta-sw}
        \frac{1}{m T} \sum_{t=1}^T \sum_{i=1}^m \SW(\vx^{(t,i)}) \geq \frac{1}{T} \sum_{t=1}^T \frac{\lambda^{(t)}}{1 + \mu^{(t)}} \opt^{(t)} - \frac{2L \sqrt{n-1}}{m} \sum_{k=1}^n \hinsim^2_k - \widetilde{O}\left(\frac{1}{mT}\right),
    \end{equation}
    %where $\vx^{(t,i)} \defeq (\vx_1^{(t,i)}, \dots, \vx_n^{(t,i)})$ is the strategy produced by the players at step $i$ of task $t$, 
    where $\opt^{(t)}$ is the optimal social welfare attainable at game $\game^{(t)}$ and $\widetilde{O}(\cdot)$ hides logarithmic terms.\looseness-1
\end{theorem}

The first term in the right-hand side of~\eqref{eq:meta-sw} is the average robust PoA in the sequence of games, while the third term vanishes as $T \to \infty$. The orchestrated learning dynamics reach approximately optimal equilibria much faster when the underlying task similarity is small; without meta-learning one would instead obtain the $m^{-1}$ rate known from the work of~\citet{Syrgkanis15:Fast}. \Cref{theorem:sw} is established by first providing a refined guarantee for the \emph{sum of the players regrets} (\Cref{theorem:sum-regs}), and then translating that guarantee in terms of the social welfare using the smoothness condition for each game (\Cref{prop:sum-sw}). Our guarantees are in fact more general, and apply for any suitable linear combination of players' utilities (see \Cref{cor:weight-regs}).\looseness-1

\subsection{Stackelberg (Security) Games}\label{sec:stack}

To conclude our theoretical results, we study meta-learning in repeated Stackelberg games. Following the convention of \citet{balcan2015commitment}, we present our results in terms of Stackelberg security games, although our results apply to general Stackelberg games as well (see~\citep[Section 8]{balcan2015commitment} for details on how such results extend).\looseness-1

\paragraph{Stackelberg security games} A repeated Stackelberg security game is a sequential interaction between a defender and $m$ attackers. In each round, the defender commits to a mixed strategy over $d$ targets to protect, which induces a \emph{coverage probability vector} $\vx \in \Delta^d$ over targets. After having observed coverage probability vector, the attacker \emph{best responds} by attacking some target $b(\vx) \in \range{d}$ in order to maximize their utility in expectation. Finally, the defender's utility is some function of their coverage probability vector $\vx$ and the target attacked $b(\vx)$.\looseness-1

It is a well-known fact that no-regret learning in repeated Stackelberg games is not possible without any prior knowledge about the sequence of followers~\citep[Section 7]{balcan2015commitment}, so we study the setting in which each attacker belongs to one of $k$ possible \emph{attacker types}. We allow sequence of attackers to be adversarially chosen from the $k$ types, and assume the attacker's type is revealed to the leader after each round.
We adapt the methodology of \citet{balcan2015commitment} to our setting by meta-learning the initialization and learning rate of the multiplicative weights update (henceforth \texttt{MWU}) run over a finite (but exponentially-large) set of \emph{extreme points} $\cE \subset \Delta^d$.\footnote{This is likely unavoidable, as \citet{li2016catcher} show computing a Stackelberg strategy is strongly \NP-Hard.} Each point $\vx \in \cE$ corresponds to a leader mixed strategy, and $\cE$ can be constructed in such a way that it will always contain a mixed strategy which is arbitrarily close to the optima-in-hindsight for each task.\footnote{For a precise definition of how to construct $\cE$, we point the reader to~\citep[Section 4]{balcan2015commitment}.}\looseness-1

Our results are given in terms of guarantees on the task-average \emph{Stackelberg regret}, which measures the difference in utility between the defender's deployed sequence of mixed strategies and the optima-in-hindsight, given that the attacker best responds.\looseness-1

\begin{definition}[Stackelberg Regret]
    Denote attacker $\pv{f}{t,i}$'s best response to mixed strategy $\vx$ as $b_{\pv{f}{t,i}}(\vx)$. The Stackelberg regret of the attacker in a repeated Stackelberg security game $t$ is 
    \begin{equation*}
        \pv{\sreg}{t,m}(\rvx^{(t)}) = \sum_{i=1}^m \langle \pv{\rvx}{t}, \pv{\vu}{t}(b_{\pv{f}{t,i}}(\pv{\rvx}{t})) \rangle - \langle \pv{\vx}{t,i}, \pv{\vu}{t}(b_{\pv{f}{t,i}}(\pv{\vx}{t,i})) \rangle.
    \end{equation*}
\end{definition}

In contrast to the standard notion of regret (\Cref{def:regret}), Stackelberg regret takes into account the extra structure in the defender's utility in Stackelberg games; namely that it is a function of the defender's current mixed strategy (through the attacker's best response).\looseness-1

\begin{theorem}[Informal; Detailed Version in~\Cref{thm:stack-full}]
    \label{theorem:informal-Stack}
Given a sequence of $T$ repeated Stackelberg security games with $d$ targets, $k$ attacker types, and within-game time-horizon $m$, running \texttt{MWU} over the set of extreme points $\cE$ as defined in \citet{balcan2015commitment} with suitable initialization and sequence of learning rates achieves task-averaged expected Stackelberg regret\looseness-1
\begin{equation*}
    \frac{1}{T} \sum_{t=1}^T \E [\pv{\sreg}{t,m}] = O(\sqrt{H(\bvy) m}) + o_T(\poly(m, |\cE|)),
\end{equation*}
where the sequence of attackers in each task can be adversarially chosen, the expectation is with respect to the randomness of \texttt{MWU}, $\bvy := \frac{1}{T} \sum_{t=1}^T \pv{\rvy}{t}$, where $\pv{\rvy}{t}$ is the optimum-in-hindsight distribution over mixed strategies in $\cE$ for game $t$, $H(\bvy)$ is the Shannon entropy of $\bvy$, and $o_T(1)$ suppresses terms which decay with $T$.\looseness-1
\end{theorem}

$H(\bvy) \leq \log |\cE|$, so in the worst-case our algorithm asymptotically matches the $O(\sqrt{m\log |\cE|})$ performance of the algorithm of \citet{balcan2015commitment}. Entropy $H(\bvy)$ is small whenever the same small set of mixed strategies are optimal for the sequence of $T$ Stackelberg games. For example, if in each task the adversary chooses from $s \ll k$ attacker types who are only interested in attacking $u \ll d$ targets (unbeknownst to the meta-learner), $H(\bvy) = O(s^2 u \log (su))$. In Stackelberg security games $|\cE| = O((2^d + kd^2)^d d^k)$, so $\log |\cE| = O(d^2 k \log (dk))$. Finally, the distance between the set of optimal strategies does not matter, as $\bvy$ is a categorical distribution over a discrete set of mixed strategies.\looseness-1

%\khcomment{If we have time, add extension to bandit feedback}
\section{Experiments}
\label{sec:exp}
\newcommand{\makelegend}{%
    \begin{tikzpicture}%
        \draw[ultra thick, plotblue] (0.5, 0) -- +(.5, 0) node[right, black] (mla) {\small Meta learning (avg. best-in-hindsight)};%
        \draw[ultra thick, plotgreen, dashdotted] ($(mla.east) + (.2,0)$) -- +(.5, 0) node[right, black] (mll) {\small Meta learning (last iterate)};%
        \draw[ultra thick, plotorange, dashed] ($(mll.east) + (.2,0)$) -- +(.5, 0) node[right, black] (nml) {\small No meta learning};%
        \draw[semithick, black!30!white, rounded corners] (0.3, -.25) rectangle ($(nml.east) + (.05, .25)$);%
    \end{tikzpicture}%
}%

\begin{figure}[t]
    \centering
    
    \begin{table}[H]
        \newcommand{\Hearts}[1]{{\textcolor{red!80!black}{#1\ding{170}}}~}
        \newcommand{\Diamonds}[1]{{\textcolor{red!80!black}{#1\ding{169}}}~}
        \newcommand{\Spades}[1]{{\textcolor{black}{#1\ding{171}}}~}
        \newcommand{\Clubs}[1]{{\textcolor{black}{#1\ding{168}}}~}
        
        \sisetup{table-number-alignment=right,group-separator={,},group-minimum-digits=4}
        \scalebox{.94}{\begin{tabular}{llrrrrrr}
            &  &  & \multicolumn{2}{c}{\bf Sequences} & \multicolumn{2}{c}{\bf Decision Points} & {\bf Payoff Matrix}\\
            {\bf Game} & {\bf Board} & {\bf Pot} & {Pl. 1} & {Pl. 2} & {Pl. 1} & {Pl. 2} & {num. nonzeros}\\
            \toprule
                Endgame \textsf{A} & \Spades{J}\Spades{K}\Clubs{5}\Spades{Q}\Diamonds{7} & \num{3700} & \num{18789} & \num{19237} &  \num{6710} &  \num{6870} & \num{14718298}\\
                Endgame \textsf{B} & \Spades{4}\Hearts{8}\Clubs{10}\Hearts{9}\Spades{2} & \num{500}  & \num{46875} & \num{47381} & \num{16304} & \num{16480} & \num{62748525}\\
            \bottomrule
        \end{tabular}}
    \end{table}\vspace{-3mm}
    
    \includegraphics[width=.99\textwidth]{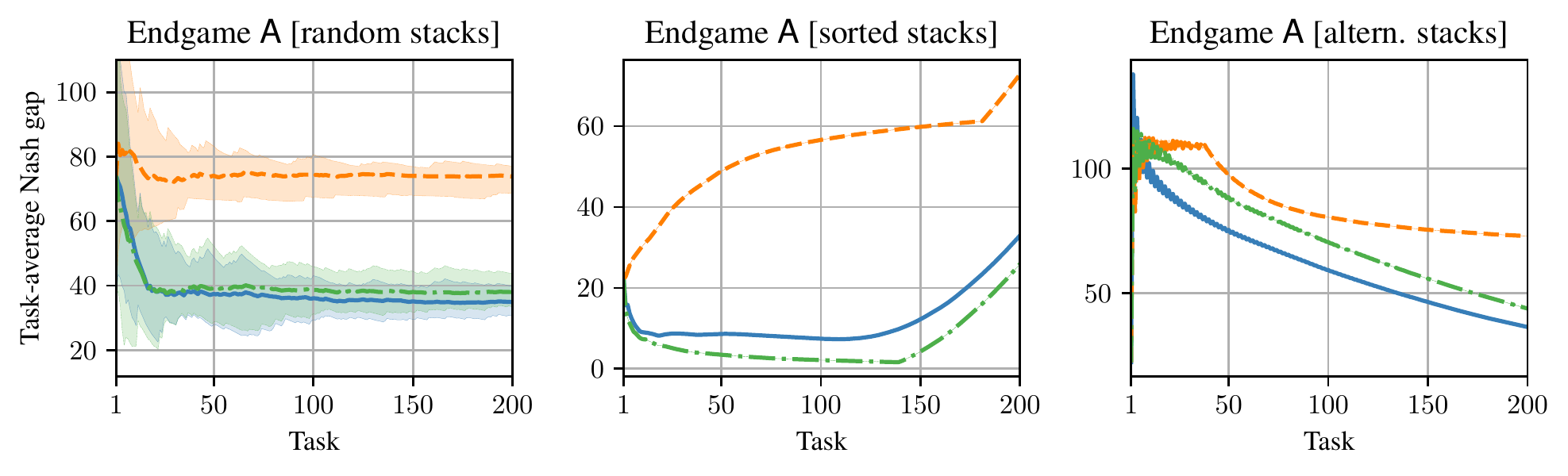}\\[-2mm]
    \includegraphics[width=.99\textwidth]{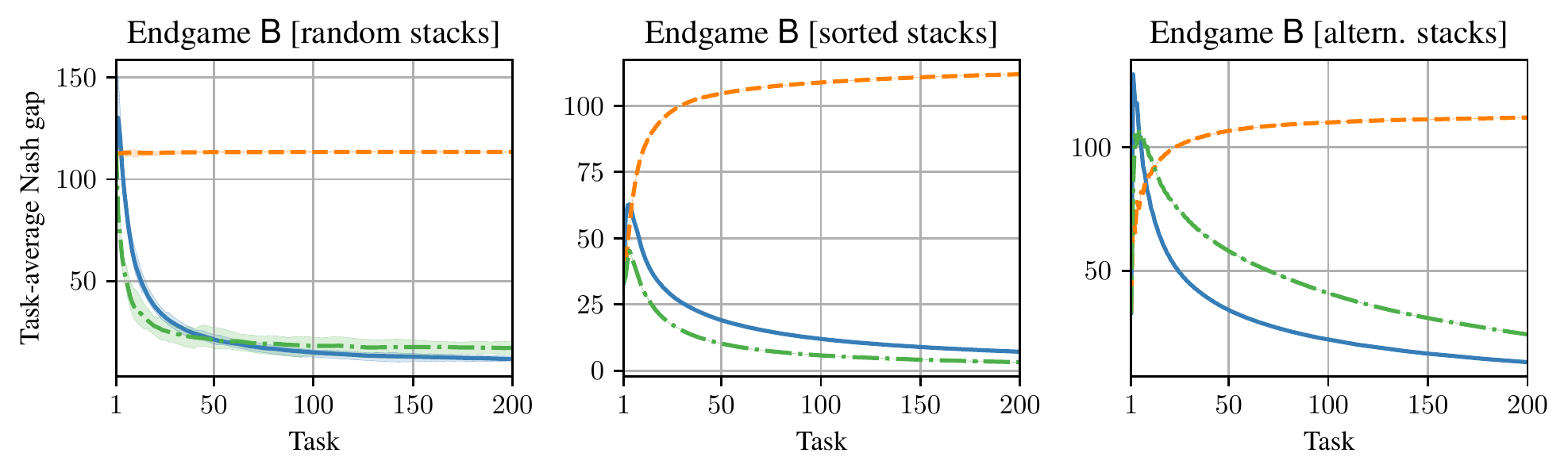}\\
    \makelegend
    \caption{(Top) Parameters of each endgame. (Bottom) The task-averaged NE gap of the players' average strategies across $200$ tasks, $2$ endgames, and $3$ different stack orderings. Both players use $\OGD$ with $\eta \defeq 0.01$. For the random stacks, we repeated each experiment $10$ times with different random seeds. The plots show the mean (thick line) as well as the minimum and maximum values. We see that across all task sequencing setups, meta-learning the initialization (using either technique) leads to up to an order of magnitude better performance compared to vanilla $\OGD$. When stacks are sorted, initializing to the last iterate of the previous game obtains the best performance, whereas when stacks are alternated or random, initializing according to~\protect\Cref{theorem:informal-dualgap} performs best.\looseness-1}
    \label{fig:experiments}
\end{figure}

In this section, we evaluate our meta-learning techniques in two River endgames that occurred in the \textit{Brains vs AI} competition~\citep{Brown18:Superhuman}. We use the two public endgames that were released by the authors,\footnote{Obtained from \url{https://github.com/Sandholm-Lab/LibratusEndgames}.} denoted `Endgame \textsf{A}' and `Endgame \textsf{B},' each corresponding to a zero-sum extensive-form game. For each of these endgames, we produced $T \defeq 200$ individual tasks by varying the size of the stacks of each player according to three different \emph{task sequencing setups}:\looseness-1\footnote{While in the general meta-learning setup it is assumed that the number of tasks is large but per-task data is limited (\emph{i.e.}, $T \gg m$), we found that setting $T\defeq 200$ was already sufficient to see substantial benefits.}\looseness-1
\begin{enumerate}[itemsep=.1mm,left=3mm,nosep]
    \item (\emph{random} stacks) In each task we select stack sizes for the players by sampling uniformly at random a multiple of $100$ in the range $[1000, 20000]$.% We note that the strategy set of each player remains the same across different tasks generated from the same endgame, since in our game representation actions are expressed as percentages of the overall stack, making our theoretical results applicable.
    \item (\emph{sorted} stacks) Task $t \in \{1, \dots, 200\}$ corresponds to solving the endgame where the stack sizes are set to the amount $t \times 100$ for each player.
    \item (\emph{alternating} stacks) We sequence the stack amounts of the players as follows: in task $1$, the stacks are set to $100$; in task $2$ to $200,000$; in task $3$ to $200$; in task $4$ to $199,900$; and so on.
\end{enumerate}

For each endgame, we tested the performance when both players (1) employ $\OGD$ while meta-learning the initialization (\Cref{theorem:informal-dualgap}) with $\vm_x^{(t,1)} = \vec{0}_{d_x}$ and $\vm_y^{(t,1)} = \vec{0}_{d_y}$, (2) employ $\OGD$ while setting the initialization equal to the last iterate of the previous task (see \Cref{remark:other-initializations}), and (3) use the vanilla initialization of $\OGD$---\emph{i.e.}, the players treat each game separately. For each game, players run $m \defeq 1000$ iterations. The $\ell_2$ projection to the \emph{sequence-form polytope}~\citep{Romanovskii62:Reduction,Koller92:The}---the strategy set of each player in extensive-form games---required for the steps of $\OGD$ is implemented via an algorithm originally described by~\citet{Gilpin12:First}, and further clarified in~\citep[Appendix B]{Farina22:Near}. We tried different learning rates for the players selected from the set $\{0.1, 0.01, 0.001\}$. \Cref{fig:experiments} illustrates our results for $\eta \defeq 0.01$, while the others are deferred to~\Cref{appendix:experiments}. In~the table at the top of~\Cref{fig:experiments} we highlight several parameters of the endgames including the board configuration, the dimensions of the players' strategy sets---\emph{i.e.}, the sequences---and the number of nonzero elements in each payoff matrix. Because of the scale of the games, we used the \emph{Kronecker sparsification} algorithm of~\citet[Technique A]{Farina22:Fast} in order to accelerate the training.\looseness-1
\section{Conclusions and Future Research}
In this paper, we introduced the study of meta-learning in games. In particular, we considered many of the most central game classes---including zero-sum games, potential games, general-sum multi-player games, and Stackelberg security games---and obtained provable performance guarantees expressed in terms of natural measures of similarity between the games. %Along the way, we also obtain new results for several single-task settings which may be of interest to the wider community on learning in games. 
Experiments on several sequences of poker endgames that were actually played in the \textit{Brains vs AI} competition~\citep{Brown18:Superhuman} show that meta-learning the initialization improves performance even by an order of magnitude.\looseness-1

Our results open the door to several exciting directions for future research, including meta-learning in other settings for which single-game results are known, such as general nonconvex-nonconcave min-max problems~\citep{Suggala20:Fast}, the nonparametric regime~\citep{Daskalakis22:Fast}, and partial feedback (such as bandit) models~\citep{Wei18:More,Hsieh22:No,balcan2022meta, osadchiy2022online}. %In particular, a pair of recent works~\citep{balcan2022meta, osadchiy2022online} obtain meta-learning guarantees under adversarial bandit feedback; it would be interesting to cast their similarity measures in terms of more interpretable game-theoretic quantities. 
Another interesting, yet challenging, avenue for future research would be to consider strategy sets that can vary across tasks.\looseness-1%While this direction could have concrete applications in certain problems, formulating such questions in a meaningful way requires further work.

\section*{Acknowledgements}
We are grateful to the anonymous ICLR reviewers for valuable feedback. KH is supported by a NDSEG Fellowship. IA and TS are supported by NSF grants IIS-1901403 and CCF-1733556, and the ARO under award W911NF2010081. MK is supported by a Meta Research PhD Fellowship. ZSW is supported in part by the NSF grant FAI-1939606, a Google Faculty Research Award, a J.P. Morgan Faculty Award, a Meta Research Award, and a Mozilla Research Grant. The authors would like to thank Nina Balcan for helpful discussions throughout the course of the project. IA is grateful to Ioannis Panageas for insightful discussions regarding \Cref{sec:stochastic}.

\newpage 
\bibliographystyle{plainnat}
\bibliography{refs}

\newpage
\appendix
\section{Additional Related Work}\label{appendix:comparison}

In this section, we provide an additional discussion on related work. Let us first compare in more detail our setting with that considered in~\citep{Zhang22:No}. \citet{Zhang22:No} study a setting which is more general than ours, thus making their algorithms applicable in the meta-learning setting we consider. Intuitively, their algorithms should not perform as well in the meta-learning setting, as they do not use knowledge of the game boundaries. More precisely, we begin by introducing the notion of dynamic Nash Equilibrium (NE) regret from \citet{Zhang22:No}, and show that in the meta-learning setting it corresponds to an unnormalized version of the maximum task average regret with respect to both players.\footnote{The results of \citet{Zhang22:No} are only applicable to two-player zero-sum games, so we will focus only on that setting in our comparison.}

\begin{definition}[Dynamic NE regret, \citet{Zhang22:No}]
Given a sequence of two-player zero-sum games characterized by payoff matrices $\pv{\mat{A}}{1}, \ldots, \pv{\mat{A}}{\tau}$ and player strategy sets $\Delta^{d_x}$ and $\Delta^{d_y}$,
\begin{equation*}
    \nereg := \left| \sum_{s=1}^\tau (\pv{\vx}{s})^\top \pv{\mat{A}}{s} \pv{\vy}{s} -  \min_{\vx \in \Delta^{d_x}} \max_{\vy \in \Delta^{d_y}} \vx^\top \pv{\mat{A}}{s} \vy \right|.
\end{equation*}
\end{definition}

In our meta-learning setting, $\tau = T \cdot m$, and $\pv{\mat{A}}{t,i} = \pv{\mat{A}}{t}, \forall i \in \range{m}, t \in \range{T}$. Using this information, we can rewrite $\nereg$ in our setting as
\begin{equation*}
    \nereg = \left| \sum_{t=1}^T \sum_{i=1}^m (\pv{\vx}{t,i})^\top \pv{\mat{A}}{t} \pv{\vy}{t,i} - \min_{\vx \in \Delta^{d_x}} \max_{\vy \in \Delta^{d_y}} \vx^\top \pv{\mat{A}}{t} \vy \right|.
\end{equation*}
Therefore, 
\begin{equation*}
    \frac{1}{T} \nereg = \max \left\{ \frac{1}{T} \sum_{t=1}^T \pv{\reg}{t,m}_x, \frac{1}{T} \sum_{t=1}^T \pv{\reg}{t,m}_y \right\}.
\end{equation*}

Given this characterization, we now informally restate Theorem 6 of \citet{Zhang22:No}, written in terms of task-average regret.

\begin{theorem}[Informal, \citet{Zhang22:No}]
When the x-player follows Algorithm 1 of \citet{Zhang22:No} and the y-player follows Algorithm 2 of \citet{Zhang22:No},

\begin{equation*}
    \max \left\{ \frac{1}{T} \sum_{t=1}^T \pv{\reg}{t,m}_x, \frac{1}{T} \sum_{t=1}^T \pv{\reg}{t,m}_y \right\} \leq \widetilde{O}\left( \frac{1}{T}\min\{ \sqrt{(1 + V)(1 + P)} + P, 1 + W \} \right),
\end{equation*}
where $V \in \mathbb{R}_{\geq 0}$ is a measure of the path-length variation of $(\pv{\mat{A}}{t} )_{1 \leq t \leq T}$, $P \in \mathbb{R}_{\geq 0}$ is a measure of the variation of the corresponding Nash equilibrium strategies, $W \in \mathbb{R}_{\geq 0}$ is a measure of the variance of $(\pv{\mat{A}}{t} )_{1 \leq t \leq T}$, and $\widetilde{O}(\cdot)$ hides logarithmic-in-$T$ factors.
\end{theorem}

It is easy to construct a sequence of games for which $V, P, W$ are all $\Omega(T)$ (\emph{e.g.}, consider an alternating sequence of two games with different payoff matrices and NE strategies for each player). Under such a setting, when players play according to the algorithms of \citet{Zhang22:No}, their task average regret guarantee will actually \emph{grow} with the number of tasks (albeit at a logarithmic rate). This is to be expected, as their algorithms are designed for a more general setting and therefore do not take the game boundaries into consideration. This is also in contrast to our results for two-player zero-sum games, where 
the parts of our bounds which explicitly depend on the number of games \emph{decrease} as the number of games grows large (\emph{e.g.}, \Cref{theorem:informal-dualgap}).

\paragraph{Broader context} Moving beyond the line of work on learning in dynamic games, our results are related to the literature on algorithms with predictions; \emph{e.g.}, see the excellent survey of~\citet{Mitzenmacher20:Algorithms}, and the many references therein. More broadly, our setting can be viewed as a specific instance of online learning under more structured sequences, a topic that has received extensive attention in the literature ~\citep{Block22:Smoothed,Rakhlin13:Online,Haghtalab22:Oracle,Cesa-Bianchi07:Improved,Chiang12:Online,Hazan10:Extracting,Hazan11:Better,Luo15:Achieving}. Finally, our work provides a number of new insights on the last-iterate convergence of optimistic no-regret learning algorithms in a variety of important settings; that line of work was pioneered by~\citet{Daskalakis18:Training}, and has thereafter witnesses rapid progress (see \citep{Mokhtari20:A,Mertikopoulos19:Optimistic,Golowich20:Last,Golowich20:Tight,wei2020linear,Azizian21:The,Fiez21:GLobal} for a highly incomplete list).
\section{Proofs from \texorpdfstring{\Cref{sec:sw}}{Section 3.2.1}: Meta-Learning Approximately Optimal Equilibria}
\label{appendix:sw}

In this section, we study how meta-learning can improve the convergence rate of learning algorithms to approximately optimal equilibria, establishing the proofs omitted from \Cref{sec:sw}. First, we provide some key preliminary results for meta-learning, which will be used throughout this paper beyond the proofs of \Cref{sec:sw}, and in particular~\Cref{appendix:zero-sum}. Then, equipped with those ingredients, we leverage the smoothness condition (\Cref{def:smooth}) to eventually arrive at the main theorem of \Cref{sec:sw} (\Cref{theorem:sw}), as well as extensions thereof.  

\subsection{Bounding the Social Regret}

Here, we establish refined meta-learning bounds for the sum of the players' regrets---oftentimes referred to as \emph{social regret}---under appropriate learning dynamics. A central ingredient of our analysis will be the so-called property of \emph{regret bounded by variation in utilities (RVU)}, crystallized by~\citet{Syrgkanis15:Fast}, a refined regret bound known to be satisfied by optimistic learning algorithms, such as $\OMD$ and \emph{optimistic follow the regularized leader ($\OFTRL$)}~\citep{Syrgkanis15:Fast}. For our purposes, it will be crucial to use an RVU bound parameterized in terms of the initialization. Indeed, the regret guarantee below can be readily extracted from~\citep{Rakhlin13:Optimization,Syrgkanis15:Fast}; we remark that our assumptions are indeed compatible with the ones made in~\citep{Syrgkanis15:Fast}. We further clarify that, in order to reduce the notational burden, when the underlying player or task are not necessary for our statement (as is below), we drop those dependencies from our notation.

\begin{theorem}[Initialization-Dependent RVU Bound~\citep{Syrgkanis15:Fast}]
\label{thm:rvu}
    Suppose that we employ $\OMD$ with a $1$-strongly convex regularizer $\calR$ with respect to the norm $\|\cdot\|$. For any observed sequence of utilities $(\pv{\vu}{i})_{1 \leq i \leq m}$, the regret $\reg^{(m)}$ of $\OMD$ up to time $m \in \N$ and initialized at $\vx^{(0)} \in \X$ can be bounded as 
    \begin{equation}
        \label{eq:rvu}
        \pv{\reg}{m} \leq \frac{1}{\eta} \brg{}{\rvx}{\pv{\vx}{0}} + \eta \sum_{i=1}^{m} \| \pv{\vu}{i} - \pv{\vec{m}}{i}\|_*^2 - \frac{1}{8\eta} \sum_{i=1}^m \| \pv{\vx}{i} - \pv{\vx}{i-1} \|^2,
    \end{equation}
    where $\rvx \in \X$ is an optimal strategy in hindsight, and $(\vx^{(i)})_{1 \leq i \leq m}$ is the sequence of strategies produced by $\OMD$.
\end{theorem}

A similar regret bound applies for $\OFTRL$ as well~\citep{Syrgkanis15:Fast}, but with the important caveat that the first term in the right-hand side of~\eqref{eq:rvu} is not refined in terms of the initialization. That deficiency is a crucial impediment towards providing provable meta-learning guarantees under $\OFTRL$, although under certain assumptions the mirror-descent viewpoint is known to be equivalent to the follow-the-regularider-leader one~\citep{McMahan11:Follow}. Also, it is worth noting that \Cref{thm:rvu} also applies under a broader class of prediction mechanisms---beyond the ``one-step recency bias,'' which consists of $\vec{m}^{(i)} \defeq \vu^{(i-1)}$---without altering qualitatively the regret bound, as formalized by~\citet{Syrgkanis15:Fast}; such extensions will not be made precise here as they are direct. 

Now, as it turns out, the RVU bound (\Cref{thm:rvu}) is a powerful tool for obtaining $O(1)$ bounds for the social regret, as was first shown by~\citet{Syrgkanis15:Fast}. Below, we follow their approach to give an initialization-dependent bound for the social regret.

\begin{corollary}
    \label{cor:sum-regs}
    Fix any $t \in \range{T}$, and suppose that $L^{(t)}$ is the Lipschitz parameter of $\game^{(t)}$---in the sense of~\eqref{eq:Lip}. If all players employ $\OGD$ with learning rate $\eta \leq \frac{1}{4  L^{(t)} \sqrt{n-1}}$, then
    \begin{equation*}
        \sum_{k=1}^n \pv{\reg}{t, m}_k \leq \frac{1}{\eta} \sum_{k=1}^n \brg{k}{\pv{\rvx}{t}_k}{\pv{\vx}{t,0}_k} - \frac{1}{16\eta} \sum_{k=1}^n \sum_{i=1}^m \| \pv{\vx_k}{t,i} - \pv{\vx_k}{t,i-1} \|_2^2.
    \end{equation*}
\end{corollary}

Before we proceed with the proof, let us point out that we generally make the very mild assumption that players know the total number of players $n$ and the Lipschitz parameter $L^{(t)}$ of each game $\game^{(t)}$---or, equivalently, reasonable upper bounds thereof, in order to tune the learning rate according to \Cref{cor:sum-regs}. If that is not the case, one can still obtain similar guarantees using the by now standard ``doubling trick'' to estimate a suitable learning rate. Concretely, for a given learning rate, each player can compute at every iteration $i \in \range{m}$ the growth of the second and the third term in~\eqref{eq:rvu} based on its ``local'' information. If at any iteration the sum of those terms is (strictly) positive, then we can simply halve the learning rate, and subsequently proceed by repeating the previous process; \Cref{cor:sum-regs} guarantees termination in a logarithmic (in the range of the parameters) number of repetitions, while only incurring a negligible increase in the social regret. While this protocol is not full uncoupled, given that some additional communication is required to notify the players to halve the learning rate, this is only a mild limitation that will not be addressed further in this work. The same discussion also applies to our results in the sequel.

\begin{proof}[Proof of \Cref{cor:sum-regs}]
    By the Lipschitz continuity assumption \eqref{eq:Lip}, we have that for any player $k \in \range{n}$ and iteration $i \in \range{m}$,
    \begin{equation*}
        \|\vu_k(\vx_{-k}^{(t,i)}) - \vu_k(\vx_{-k}^{(t,i-1)}) \|_2 \leq L^{(t)} \| \vx_{-k}^{(t,i)} - \vx_{-k}^{(t,i-1)} \|_2.
    \end{equation*}
    Thus, combining with \Cref{thm:rvu} under $\|\cdot\| = \|\cdot\|_2$ (since we use $\OGD$), we have that the regret $\pv{\reg}_k{t, m}(\rvx_k^{(t)})$ of player $k \in \range{n}$ is upper bounded by
    \begin{align*}
         &\frac{1}{\eta} \brg{k}{\pv{\rvx_k}{t}}{\pv{\vx_k}{t,0}} + \eta \sum_{i=1}^m \|\vu_k^{(i)} - \vu_k^{(i-1)}\|_2^2 - \frac{1}{8\eta} \sum_{i=1}^m \| \pv{\vx_k}{t,i} - \pv{\vx_k}{t,i-1} \|_2^2 \\
        &\leq \frac{1}{\eta} \brg{k}{\pv{\rvx_k}{t}}{\pv{\vx_k}{t,0}} + \eta (L^{(t)})^2 \sum_{k' \neq k} \sum_{i=1}^{m} \| \pv{\vx_{k'}}{t,i} - \pv{\vx_{k'}}{t,i-1}\|_2^2 - \frac{1}{8\eta} \sum_{i=1}^m \| \pv{\vx_k}{t,i} - \pv{\vx_k}{t,i-1} \|_2^2,
    \end{align*}
    since $\| \vx_{-k}^{(t,i)} - \vx_{-k}^{(t,i-1)} \|^2_2 = \sum_{k' \neq k} \| \pv{\vx_{k'}}{t,i} - \pv{\vx_{k'}}{t,i-1}\|_2^2 $ (Pythagorean theorem). In turn, this implies that the sum of the players' regrets $\sum_{k=1}^n \pv{\reg}_k{t, m}(\rvx_k^{(t)})$ is upper bounded by
    \begin{align*}
        &\frac{1}{\eta} \sum_{k=1}^n \brg{k}{\pv{\rvx_k}{t}}{\pv{\vx_k}{t,0}} + \left( \eta (L^{(t)})^2 (n-1) - \frac{1}{8\eta}  \right) \sum_{k=1}^n \sum_{i=1}^m \| \pv{\vx_k}{t,i} - \pv{\vx_k}{t,i-1} \|_2^2 \\
        &\leq \frac{1}{\eta} \sum_{k=1}^n \brg{k}{\pv{\rvx_k}{t}}{\pv{\vx_k}{t,0}} - \frac{1}{16\eta} \sum_{k=1}^n \sum_{i=1}^m \| \pv{\vx_k}{t,i} - \pv{\vx_k}{t,i-1} \|_2^2,
    \end{align*}
    where the last bound follows since $\eta \leq \frac{1}{4  L^{(t)} \sqrt{n-1}}$ (by assumption). This concludes the proof.
\end{proof}

We remark that in the proof above we used the assumption that $\vec{m}_k^{(t, 1)} \defeq \vu_k(\vx_{-k}^{(t, 0)})$, which is also needed in~\citep{Syrgkanis15:Fast}---although it was not explicitly mentioned by the authors. That assumption can be circumvented using standard techniques; as such, it will be lifted in~\Cref{appendix:lowerbounds}.

Armed with \Cref{cor:sum-regs}, we are now ready to state the key result of this subsection: a refined guarantee for the sum of the players' regrets, parameterized in terms of the similarity of the optimal in hindsight of each player. The meta-version of $\OGD$ that we consider is summarized in~\Cref{alg:meta-ogd}.

\begin{algorithm}[t]
\SetAlgoLined
\KwData{\begin{itemize}[noitemsep,leftmargin=.5cm]
    \item Number of players $n \in \N \setminus \{1\} $;
    \item Dimension of strategy set $d_k \in \N$ for each $k \in \range{n}$; and
    \item Lipschitz constant $L > 0$.
\end{itemize}}
Initialize $\pv{\vx_k}{1,0} = \frac{1}{d_k} \mathbbm{1}_{d_k} \in \X_k $ for each player $k \in \range{n}$ \;
\For{game $t = 1, \ldots, T$}{
 Each player $k$ runs $\OGD$ for $m$ iterations using initialization $\pv{\vx_k}{t,0}$ and learning rate $\eta = \frac{1}{4L \sqrt{n-1} }$\;
Each player $k$ computes the next initialization as $\pv{\vx_k}{t+1,0} = \frac{1}{t} \sum_{s=1}^t \pv{\rvx_k}{t}$\;
}
\caption{Meta-$\OGD$}\label{alg:meta-ogd}
\end{algorithm}

\begin{theorem}
    \label{theorem:sum-regs}
    Suppose that each player $k \in \range{n}$ employs $\OGD$ with learning rate $\eta = \frac{1}{4L\sqrt{n-1}}$, with $L \defeq \max_{1 \leq t \leq T} L^{(t)}$, and initialization $\vx_k^{(t,0)} \defeq \sum_{s < t} \rvx_k^{(s)} \big/ {(t-1)}$, for $t \geq 2$. Then,
    \begin{equation}
        \label{eq:sumregs}
        \frac{1}{T} \sum_{t=1}^T \sum_{k=1}^n \pv{\reg}_k{t, m} \leq 2L\sqrt{n-1} \sum_{k=1}^n \hinsim^2_k + \frac{4 L \sqrt{n-1} (1 + \log T)}{T} \sum_{k=1}^n \diam_{\X_k}^2.
    \end{equation}
\end{theorem}

In light of \Cref{cor:sum-regs}, the main ingredient for establishing this theorem is a meta-algorithm that determines the initialization of $\OGD$. As we shall see, the initialization seen in \Cref{theorem:sum-regs}, namely $\vx_k^{(t,0)} \defeq \sum_{s < t} \rvx_k^{(s)} \big/ {(t-1)}$, is what comes out of \emph{follow the leader} ($\FTL$)---the ``unregularized'' version of $\FTRL$. To justify this, we will need the following auxiliary results. We note that, for convenience, the guarantee below is stated in terms of an underlying sequence of losses, instead of utilities; naturally, those two viewpoints are equivalent.

\begin{proposition}[\citep{Khodak19:Adaptive}]
    \label{proposition:Khodak}
    Let $\calR : \X \to \R$ be a $1$-strongly convex regularizer with respect to the norm $\|\cdot\|$. Then, for any sequence $\vx^{(1)}, \dots, \vx^{(T)} \in \X$, $\FTL$ run on the sequence of losses $\alpha^{(1)} \brg{}{\vx^{(1)}}{\cdot}, \dots, \alpha^{(T)} \brg{}{\vx^{(T)}}{\cdot}$, where $\vec{\alpha} \in \R^T_{> 0}$, has regret $\reg^{(T)}$ bounded as 
    \begin{equation*}
        \reg^{(T)} \leq 2 C \diam_{\X} \sum_{t=1}^T \frac{(\alpha^{(T)})^2 G^{(T)}}{\alpha^{(t)} + 2 \sum_{s < t} \alpha^{(s)}},
    \end{equation*}
    where $C \in \R_{> 0}$ is such that $\|\vx\| \leq C \|\vx\|_2$ for any $\vx \in \X$; $\diam_{\X}$ is the $\ell_2$-diameter of $\X$; and $G^{(t)}$ is the Lipschitz constant of the function $\brg{}{\vx^{(t)}}{\cdot}$ with respect to $\|\cdot\|$.
\end{proposition}

In particular, in this section we will be using this proposition for the Bregman divergence induced by the Euclidean regularizer $\calR : \vx \mapsto \frac{1}{2} \| \vx \|_2^2$, in which case the sequence of losses encountered by the $\FTL$ algorithm happens to be strongly convex; for that special case, logarithmic regret under $\FTL$ is well-known from earlier works~\citep{Cesa-bianchi2006:Prediction}. \Cref{proposition:Khodak} can then be further simplified so that $C = 1$ and $G^{(t)} \leq \diam_{\X}$, for any $t \in \range{T}$, where we recall that $\diam_{\X}$ is the $\ell_2$-diameter of $\X$. To see this, we simply note that for any $\vx, \vx' \in \X$,
\begin{equation*}
    \left| \brg{}{\vx^{(t)}}{\vx} - \brg{}{\vx^{(t)}}{\vx'} \right| = \left| \left\langle \vx - \vx', \frac{\vx + \vx'}{2} - \vx^{(t)} \right\rangle \right| \leq \diam_{\X} \|\vx - \vx'\|_2,
\end{equation*}
since $\frac{\vx + \vx'}{2} \in \X$ (by convexity). Further, $\FTL$ takes a natural form, as implied by the following simple claim. We recall again that $\X$ it always assumed to be nonempty convex and compact.

\begin{claim}
    \label{claim:ftl}
    Consider a sequence of points $\vx^{(1)}, \dots, \vx^{(T)} \in \X \subseteq \R^d$, for $d \in \N$. Then, the function 
    \begin{equation*}
        \X \ni \vx \mapsto \frac{1}{2} \sum_{t=1}^T \alpha^{(t)} \|\vx^{(t)} - \vx\|_2^2,
    \end{equation*}
    with $\vec{\alpha} \in \R^T_{> 0}$, attains its (unique) minimum at $\xstar \defeq \sum_{t=1}^T \alpha^{(t)} \vx^{(t)} \big/ \sum_{t=1}^T \alpha^{(t)}$.
\end{claim}

We point out that, by convexity, $\xstar$ is indeed a feasible point in $\X$. Such a characterization (\Cref{claim:ftl}) is known to be extended beyond Euclidean regularization~\citep{Banerjee05:Clustering}, which will be used in the proofs of~\Cref{theorem:cce,theorem:ce}. We are now ready to establish \Cref{theorem:sum-regs}.

\begin{proof}[Proof of \Cref{theorem:sum-regs}]
    First, by \Cref{cor:sum-regs}, we have that for any task $t \in \range{T}$,
    \begin{equation}
        \label{eq:cor-sum}
        \sum_{k=1}^n \pv{\reg}_k{t, m} \leq 4 L \sqrt{n-1} \sum_{k=1}^n \brg{k}{\pv{\rvx_k}{t}}{\pv{\vx_k}{t,0}},
    \end{equation}
    where we used the fact that $\eta \defeq \frac{1}{4L\sqrt{n-1}} \leq \frac{1}{4L^{(t)}\sqrt{n-1}}$ for each task $t \in \range{T}$, by definition of $L \defeq \max_{1 \leq t \leq T} L^{(t)}$, thereby satisfying the precondition of \Cref{cor:sum-regs}.
    Thus,
    \begin{align}
        \frac{1}{T} \sum_{t=1}^T \sum_{k=1}^n \pv{\reg}_k{t, m} &\leq \frac{4 L \sqrt{n-1}}{T} \sum_{k=1}^n \sum_{t=1}^T \brg{k}{\pv{\rvx_k}{t}}{\pv{\vx_k}{t, 0}} \label{align:cor-sum} \\
        &= \frac{2 L \sqrt{n-1}}{T} \sum_{k=1}^n \sum_{t=1}^T \|\rvx_k^{(t)} - \vx_k^{(t,0)}\|_2^2 \label{align:breg-Eucl} \\
        &= \frac{2 L \sqrt{n-1}}{T} \sum_{k=1}^n \min_{\vx_k \in \X} \sum_{t=1}^T \|\rvx_k^{(t)} - \vx_k \|_2^2 \notag \\
        &+ \frac{2 L \sqrt{n-1}}{T} \sum_{k=1}^n \left(\sum_{t=1}^T \|\rvx_k^{(t)} - \vx_k^{(t,0)}\|_2^2 - \min_{\vx_k \in \X} \sum_{t=1}^T \|\rvx_k^{(t)} - \vx_k \|_2^2 \right) \notag \\
        &\leq \frac{2 L \sqrt{n-1}}{T} \sum_{k=1}^n \min_{\vx_k \in \X} \sum_{t=1}^T  \|\rvx_k^{(t)} - \vx_k \|_2^2 + \frac{8 L \sqrt{n-1} (1 + \log T)}{T} \sum_{k=1}^n \diam_{\X_k}^2  \label{align:FTL} \\
        &= \frac{2 L \sqrt{n-1}}{T} \sum_{k=1}^n \sum_{t=1}^T \|\rvx_k^{(t)} - \bvx_k \|_2^2 + \frac{8 L \sqrt{n-1} (1 + \log T)}{T} \sum_{k=1}^n \diam_{\X_k}^2 \label{align:min} \\
        &= 2L\sqrt{n-1} \sum_{k=1}^n V_k^{2} + \frac{8 L \sqrt{n-1} (1 + \log T)}{T} \sum_{k=1}^n \diam_{\X_k}^2, \label{align:final-not}
    \end{align}
    where 
    \begin{itemize}
        \item \eqref{align:cor-sum} follows from \eqref{eq:cor-sum}; 
        \item \eqref{align:breg-Eucl} uses the fact that $\OGD$ corresponds to $\OMD$ with $\calR_k : \X_k \ni \vx_k \mapsto \frac{1}{2} \|\vx_k\|_2^2$, thereby implying that $\brg{k}{\rvx^{(t)}}{\vx_k^{(t,0}} = \frac{1}{2} \|\rvx^{(t)} - \vx_k^{(t,0)}\|^2_2$;
        \item \eqref{align:FTL} uses \Cref{proposition:Khodak} with $\|\cdot\| = \|\cdot\|_2$. In particular, the quantity
        \begin{equation}
            \label{eq:reg-FTL}
            \frac{1}{2} \sum_{t=1}^T \left( \|\rvx_k^{(t)} - \vx_k^{(t,0)}\|_2^2  - \min_{\vx_k \in \X} \|\rvx_k^{(t)} - \vx_k \|_2^2 \right)
        \end{equation}
        can be recognized as the regret incurred by the $\FTL$ algorithm used by each player $k \in \range{n}$, which in turn follows since we have initialized as $\vx_k^{(t,0)} \defeq \sum_{s < t} \rvx_k^{(s)} \big/ {(t-1)}$, for $t \geq 2$, which by \Cref{claim:ftl} is exactly the update of $\FTL$ under $\alpha^{(1)} = \dots = \alpha^{(t)} = 1$. As a result, by \Cref{proposition:Khodak}, the regret of $\FTL$~\eqref{eq:reg-FTL} for each player $k \in \range{n}$ can be upper bounded by
        \begin{equation*}
            2 \Omega^2_{\X} \sum_{t=1}^T \frac{1}{2t - 1} \leq 2 \Omega^2_{\X} \sum_{t=1}^T \frac{1}{t} \leq 2 \Omega^2_{\X} (1 + \log T),
        \end{equation*}
        where the last bound uses the well-known inequality that the $t$-th harmonic number $\mathcal{H}^{(t)}$ is upper bounded by $1 + \log T$, where $\log (\cdot)$ here denotes the natural logarithm.
        \item \eqref{align:min} uses the fact that the function $\vx_k \mapsto \sum_{t=1}^T \|\rvx^{(t)} - \vx_k\|_2^2$ is minimized at $\Bar{\vx}_k \defeq \sum_{t=1}^T \rvx_k^{(t)} \big/ T$ (by \Cref{claim:ftl}); and
        \item \eqref{align:final-not} uses the notation $\hinsim_k^2 \defeq \frac{1}{T} \sum_{t=1}^T \|\rvx_k^{(t)} - \bvx_k \|_2^2$ for the similarity of the optimal in hindsight of player $k \in \range{n}$.
    \end{itemize}
\end{proof}

\begin{remark}
    For the sake of simplicity in the exposition, in \Cref{theorem:sum-regs} we have assumed that players use the same learning rate for each task, but that guarantee can be further improved if players use different learning rates, as long as \Cref{cor:sum-regs} can be applied. Indeed, \Cref{proposition:Khodak} is versatile enough to capture the case where the Bregman divergences have different weights, though the induced expression in that case is rather cumbersome.
\end{remark}

\begin{remark}
    \label{remark:sum-weights}
    \Cref{theorem:sum-regs} can be directly applied if the left-hand side of~\eqref{eq:sumregs} is replaced by $\frac{1}{T} \sum_{t=1}^T \alpha^{(t)} \sum_{k=1}^n \reg_k^{(t,m)}$, in which case the right-hand side of~\eqref{eq:sumregs} ought to be multiplied by $\max_{1 \leq t \leq T} \alpha^{(t)}$, for some vector $\vec{\alpha} \in \R_{> 0}^T$. This simple fact will be exploited in the proof of \Cref{theorem:sw}.
\end{remark}

\begin{remark}
    \label{remark:other-initializations}
    From \eqref{align:breg-Eucl}, if each player $k \in \range{n}$ initializes to $\pv{\vx}{t,0}_k = \pv{\rvx_k}{t-1}$, i.e., the optima-in-hindsight of the previous game, then $\frac{1}{T} \sum_{t=1}^T \sum_{k=1}^n \pv{\reg}_k{t, m} \leq \frac{2 L \sqrt{n-1}}{T} \sum_{t=1}^T \sum_{k=1}^n \|\rvx_k^{(t)} - \pv{\rvx_k}{t-1}\|_2^2$, where by convention $\rvx_k^{(0)} \defeq \vx_k^{(1,0)}$. Similarly, if each player initializes to $\pv{\vx}{t,0}_k = \pv{\vx_k}{t-1, m}$, i.e., the last iterate of the previous game, then $\frac{1}{T} \sum_{t=1}^T \sum_{k=1}^n \pv{\reg}_k{t, m} \leq \frac{2 L \sqrt{n-1}}{T} \sum_{t=1}^T \sum_{k=1}^n \|\rvx_k^{(t)} - \pv{\vx_k}{t-1, m}\|_2^2$, where by convention $\vx_k^{(0,m)}\defeq \vx_k^{(1,0)}$. While the bound in~\Cref{theorem:sum-regs} might be generally better if there is no sequential structure to the games encountered, it may be desirable to instead use one of the other initializations if some form of sequential structure is known to exist, as illustrated in our experiments (\Cref{sec:exp}).
\end{remark}

\subsection{Implications for the Social Welfare}
\label{appendix:impl-sw}

Having refined the guarantee in terms of the social regret (\Cref{theorem:sum-regs}), we can proceed with the proof of \Cref{theorem:sw}. Let us first further comment on \Cref{def:smooth}, which was introduced earlier in \Cref{sec:sw}. To do so, we first give some basic background on normal-form (or strategic-form) games; this will also make our setup in \Cref{sec:background} more concrete. In a normal-form game, each player $k \in \range{n}$ has a finite and nonempty set of available actions $\calA_k$. For a joint action profile $\vec{a} = (a_1, \dots, a_n) \in \bigtimes_{k=1}^n \calA_k$, there is a utility function $u_k : \vec{a} \mapsto u_k(\vec{a})$, assigning a utility to each player $k$ given $\vec{a}$. Players are allowed to randomize by selecting a probability distribution over their set of actions. Indeed, $\X_k \defeq \Delta(\calA_k)$ corresponds to the strategy set of player $k \in \range{n}$. In the setting described in \Cref{sec:background}, $u_k : \bigtimes_{k=1}^n \X_k \to \R$ is the mixed extension of the utility function: $u_k: \vx \mapsto \E_{\vec{a} \sim \vx} [u_k(\vec{a})]$. That is, players act so as to maximize their expected utility.

In this context, \citet{Roughgarden15:Intrinsic} introduced his notion of smoothness with respect to pure strategies:
\begin{equation}
    \label{eq:smooth-pure}
\sum_{k=1}^n u_k(\astar_k,\va_{-k}) \geq \lambda \opt - \mu \SW(\va),    
\end{equation}
for any two joint action profiles $\va, \astar \in \bigtimes_{k=1}^n \A_k$; in fact, $\astar$ can be restricted to be an action profile that maximizes social welfare~\citep{Roughgarden15:Intrinsic}. Thus, \Cref{def:smooth}, stated in terms of mixed strategies, can be obtained by taking an expectation in~\eqref{eq:smooth-pure} and restricting $\xstar$ to be a pure strategy---namely, $\astar$. 

For a $(\lambda,\mu)$-smooth game $\game$, we recall that the \emph{robust} price of anarchy is defined as $\rho \defeq \lambda/(1 + \mu)$; it was coined ``robust'' by~\citet{Roughgarden15:Intrinsic} because---among others---it gives a guarantee for any coarse correlated equilibrium of $\game$---not just the Nash equilibria. We note that the robust price of anarchy might be different than the price of anarchy---the ratio between the worst Nash equilibrium (in terms of social welfare)---and the optimal state---hypothetically imposed by a benevolent dictator. The importance of Roughgarden's smoothness framework is that for many classes of games $\rho$ is remarkably close to the actual price of anarchy~\citep{Roughgarden15:Intrinsic}. Indeed, we give in \Cref{tab:smoothness} a number of important settings for which $\rho$ is close to $1$, including games/mechanisms discussed by~\citet{Hartline15:No}.

\begin{table}[!ht]
    \small
    \centering
    \def\arraystretch{2.3}
    \begin{tabular}{m{6.0cm}c>{\small\arraybackslash}l}
        %\toprule
        \bf Game/Mechanism & $(\lambda,\mu)$ & \bf Reference \vspace{-2mm}\\[1.5mm]
        \toprule
        Simultaneous first-price auction with submodular bidders & $\displaystyle\left(1 - \frac{1}{e}, 1\right)$ & \citep{Syrgkanis13:Composable} \\ %\midrule 
        First-price multi-unit auction & $\displaystyle\left(1 - \frac{1}{e}, 1\right)$ & \citep{Syrgkanis13:Composable} \\ %\midrule
        Simultaneous second-price auctions & $(1,1)$ & \citep{Christodoulou16:Bayesian} \\ %\midrule
        Greedy combinatorial auction with\newline $d$-complements & $\displaystyle\left(1 - \frac{1}{e}, d\right)$ & \citep{Lucier10:Price} \\ %\midrule
        Valid utility games & $(1,1)$ & \citep{Vetta02:Nash} \\ %\midrule
        Congestion games with affine costs & $\displaystyle\left(\frac{5}{3}, \frac{1}{3}\right)$ & \small\citep{Christodoulou05:The} \\ \bottomrule
    \end{tabular}
    \caption{Smoothness parameters for well-studied games/mechanisms in the literature. For some of those settings, smoothness---in the sense of \Cref{def:smooth}---was only subsequently crystallized by~\citet{Roughgarden15:Intrinsic} based on the earlier arguments.}
    \label{tab:smoothness}
\end{table}

Now, for the sake of generality, let us treat the problem of maximizing welfare from a slightly broader standpoint. More precisely, for a vector $\Vec{\alpha} \in \R_{> 0}^n$, we define $\SW_{\vec{\alpha}}(\vx) \defeq \sum_{k=1}^n \alpha_k u_k(\vx)$. That is, the utility of each player $k \in \range{n}$ is weighted using a coefficient $\alpha_k > 0$; when $\alpha_1 = \dots = \alpha_n = 1$, we recover the standard notion of (utilitarian) social welfare we introduced earlier. We also let $\opt_{\vec{\alpha}} \defeq \max_{\vec{x} \in \X} \SW_{\vec{\alpha}}(\vx)$. In this context, we consider the following generalized notion of smoothness.

\begin{definition}[Extension of \Cref{def:smooth}]
    \label{def:ext-smooth}
    A game $\game$ is $(\lambda, \mu)$-smooth with respect to $\vec{\alpha} \in \R^n_{> 0}$, with $\lambda, \mu > 0$, if there exists a strategy profile $\xstar \in \bigtimes_{k = 1}^n \X_k $ such that for any $\vx \in \bigtimes_{k = 1}^n \X_k$,
    \begin{equation}
        \label{eq:ext-smooth}
        \sum_{k=1}^n \alpha_k u_k(\xstar_k, \vx_{-k}) \geq \lambda \opt_{\vec{\alpha}} - \mu \SW_{\vec{\alpha}}(\vx).
    \end{equation}
\end{definition}

In the sequel, when $\vec{\alpha}$ will remain unspecified it will be implied that $\alpha_1 = \dots = \alpha_n = 1$. Below we show a simple extension of an observation due to~\citet{Roughgarden15:Intrinsic}.

\begin{proposition}
    \label{prop:sum-sw}
    Suppose that each player $k \in \range{m}$ incurs regret at most $\reg_k^{(m)}$ in a $(\lambda, \mu)$-smooth game $\game$ with respect to $\vec{\alpha} \in \R^n_{> 0}$. Then,
    \begin{equation*}
        \frac{1}{m} \sum_{i=1}^m \SW_{\vec{\alpha}}(\vx^{(i)}) \geq \frac{\lambda}{1 + \mu} \opt_{\vec{\alpha}} - \frac{1}{1 + \mu} \frac{1}{m} \sum_{k=1}^n \alpha_k \reg_k^{(m)},
    \end{equation*}
    where $\lambda/(1 + \mu)$ is the (robust) price of anarchy.
\end{proposition}

\begin{proof}
    Suppose that $\xstar \in \bigtimes_{k=1}^n \X_k $ is the strategy profile for which~\eqref{eq:ext-smooth} is satisfied. Then, by definition, for any player $k \in \range{n}$ and iteration $i \in \range{m}$,
    \begin{equation*}
         \alpha_k \reg_k^{(m)} \geq \sum_{i=1}^m \alpha_k u_k(\xstar_k, \vx^{(i)}_{-k} ) - \alpha_k \sum_{i=1}^m u_k(\vx^{(i)}), 
    \end{equation*}
    since $\alpha_k > 0$ (by assumption). So, summing over all players,
    \begin{equation*}
        \sum_{k=1}^n \alpha_k \reg_k^{(m)} \geq \lambda m \opt - \mu \sum_{i=1}^m \SW_{\vec{\alpha}}(\vx^{(i)}) - \sum_{i=1}^m \SW_{\vec{\alpha}}(\vx^{(i)}),
    \end{equation*}
    and rearranging the last inequality concludes the proof.
\end{proof}

Next, we first focus on the case where $\alpha_{1} = \dots = \alpha_{n} = 1$, and we then give the more general result. In particular, below we show \Cref{theorem:sw}, the informal version of which was stated earlier in \Cref{sec:sw}.  

\begin{theorem}[Detailed Version of \Cref{theorem:sw}]
    \label{theorem:sw-formal}
    If all players use $\OGD$ with learning rate $\eta = \frac{1}{4L\sqrt{n-1}}$, with $L \defeq \max_{1 \leq t \leq T} L^{(t)}$, and initialization $\vx_k^{(t,0)} \defeq \sum_{s < t} \rvx_k^{(s)} \big/ {(t-1)}$, for $t \geq 2$ and $k \in \range{n}$, in a sequence of $T$ games $(\game^{(t)})_{1 \leq t \leq T}$, each of which is $(\lambda^{(t)}, \mu^{(t)})$-smooth, then
    \begin{equation*}
        \frac{1}{m T} \sum_{t=1}^T \sum_{i=1}^m \SW(\vx^{(t,i)}) \geq \frac{1}{T} \sum_{t=1}^T \frac{\lambda^{(t)}}{1 + \mu^{(t)}} \opt^{(t)} - \frac{2L\sqrt{n-1}}{m} \sum_{k=1}^n \hinsim^2_k - \widetilde{O}\left(\frac{1}{mT}\right),
    \end{equation*}
    where $\vx^{(t,i)} \defeq (\vx_1^{(t,i)}, \dots, \vx_n^{(t,i)})$ is the strategy produced by the players at iteration $i$ of task $t$, and $\opt^{(t)}$ is the optimal social welfare attainable at game $\game^{(t)}$.
\end{theorem}

\begin{proof}
    First, using \Cref{prop:sum-sw} we have that for any task $t \in \range{T}$,
    \begin{equation}
        \label{eq:sw-pertask}
        \frac{1}{m} \sum_{i=1}^m \SW(\vx^{(t,i)}) \geq \frac{\lambda^{(t)}}{1 + \mu^{(t)}} \opt^{(t)} - \frac{1}{1 + \mu^{(t)}} \frac{1}{m} \sum_{k=1}^n \reg_k^{(t,m)},
    \end{equation}
    since each game $\game^{(t)}$ is assumed to be $(\lambda^{(t)}, \mu^{(t)})$-smooth with optimal social welfare $\opt^{(t)}$. Hence, taking the average of~\eqref{eq:sw-pertask} over all $t \in \range{T}$ yields that
    \begin{equation*}
        \frac{1}{m T} \sum_{t=1}^T \sum_{i=1}^m \SW(\vx^{(t,i)}) \geq \frac{1}{T} \sum_{t=1}^T \frac{\lambda^{(t)}}{1 + \mu^{(t)}} \opt^{(t)} - \frac{1}{m T} \sum_{t=1}^T \frac{1}{1 + \mu^{(t)}} \sum_{k=1}^n \reg_k^{(t,m)}.
    \end{equation*}
    
    Finally, using \Cref{theorem:sum-regs} along with the refinement discussed in \Cref{remark:sum-weights}---with $\alpha^{(t)} \defeq \frac{1}{1 + \mu^{(t)}} \in [0,1]$---concludes the proof.
\end{proof}

For the more general result in which we allow any weight vector $\vec{\alpha} \in \R_{> 0}^n$, it suffices to refine \Cref{cor:sum-regs} as follows.

\begin{corollary}[Extension of \Cref{cor:sum-regs}]
    \label{cor:weight-regs}
    Fix any $t \in \range{T}$, and suppose that $L^{(t)}$ is the Lipschitz parameter of $\game^{(t)}$---in the sense of~\eqref{eq:Lip}. If all players employ $\OGD$ with learning rate 
    \begin{equation}
        \label{eq:learn-alpha}
    \eta \leq \frac{1}{2\sqrt{2} L^{(t)}} \max_{k \in \range{n}} \sqrt{\frac{\alpha_k}{\sum_{k' \neq k} \alpha_{k'}}},    
    \end{equation}
    where $\vec{\alpha} \in \R_{> 0}^n$, then
    \begin{equation}
        \label{eq:weightreg}
        \sum_{k=1}^n \alpha_k \pv{\reg}{t, m}_k \leq \frac{1}{\eta} \sum_{k=1}^n \alpha_k \brg{k}{\pv{\rvx}{t}_k}{\pv{\vx}{t,0}_k}.
    \end{equation}
\end{corollary}

The proof is analogous to \Cref{cor:sum-regs}, but we include it below for completeness. We remark that if $\alpha_k = 1$, for some player $k \in \range{n}$, and the rest of the coefficients approach to $0$, \Cref{cor:weight-regs} gives a vacuous guarantee: the learning rate required in~\eqref{eq:learn-alpha} approaches to $0$, thereby making the right-hand side of~\eqref{eq:weightreg} unbounded. This is precisely the reason why obtaining near-optimal $\widetilde{O}(1)$ bounds for $\max_{k \in \range{n}} \reg_k^{(m)}$, instead of $\sum_{k = 1}^n \reg_k^{(m)}$, requires more refined techniques~\citep{Daskalakis21:Near}. 

\begin{proof}[Proof of \Cref{cor:weight-regs}]
    Similarly to~\Cref{cor:weight-regs}, by $L^{(t)}$-Lipschitz continuity~\eqref{eq:Lip} and \Cref{thm:rvu}, we have that $\alpha_k
    \pv{\reg}_k{t, m}(\rvx_k^{(t)})$ can be upper bounded by
    \begin{equation*}
         \frac{\alpha_k}{\eta} \brg{k}{\pv{\rvx_k}{t}}{\pv{\vx_k}{t,0}} + \eta \alpha_k (L^{(t)})^2 \sum_{k' \neq k} \sum_{i=1}^{m} \| \pv{\vx_{k'}}{t,i} - \pv{\vx_{k'}}{t,i-1}\|_2^2 - \frac{\alpha_k}{8\eta} \sum_{i=1}^m \| \pv{\vx_k}{t,i} - \pv{\vx_k}{t,i-1} \|_2^2,
    \end{equation*}
    for any player $k \in \range{k}$, where we used the fact that $\alpha_k > 0$. As a result, we have that $\sum_{k=1}^n \alpha_k
    \pv{\reg}_k{t, m}(\rvx_k^{(t)})$ can be in turn upper bounded by
    \begin{equation*}
         \frac{1}{\eta} \sum_{k=1}^n \alpha_k \brg{k}{\pv{\rvx_k}{t}}{\pv{\vx_k}{t,0}} + \sum_{k = 1}^n \left( \eta (L^{(t)})^2 \sum_{k' \neq k} \alpha_{k'} - \frac{\alpha_k}{8\eta} \right) \sum_{i=1}^m \| \pv{\vx_k}{t,i} - \pv{\vx_k}{t,i-1} \|_2^2.
    \end{equation*}
    Combining this with the bound on the learning rate~\eqref{eq:learn-alpha} completes the proof.
\end{proof}

\section{Proofs from \texorpdfstring{\Cref{sec:zero-sum}}{Section 3.1}: Meta-Learning in Zero-Sum Games}
\label{appendix:zero-sum}

In this section, we primarily focus on meta-learning in two-player zero-sum games. We also give a number of gradual extensions and generalizations.

\subsection{Bilinear Saddle-Point Problems}
\label{sec:bssp}

We begin by considering the bilinear saddle-point problem
\begin{equation}
    \label{eq:bssp}
    \min_{\vx \in \X} \max_{\vy \in \Y} \vx^\top \mat{A} \Vec{y},
\end{equation}
where we recall that $\X \subseteq \R^{d_x}$ is the nonempty convex and compact set of strategies of player $x$ with $d_x \in \N$; $\Y \subseteq \R^{d_y}$ is the nonempty convex and compact set of strategies of player $y$ with $d_y \in \N$; and $\mat{A} \in \R^{\dx \times \dy}$ is the coupling (payoff) matrix of the game. In what follows in the coming subsections, we will extend our results to more general settings.

We will first apply \Cref{theorem:sum-regs} (proven in \Cref{appendix:sw}) for the bilinear saddle-point problem~\eqref{eq:bssp}. To do so, we first note that the Lipschitz continuity assumption~\eqref{eq:Lip} for a fixed task $t$ is satisfied with $L^{(t)} \defeq \|\mat{A}^{(t)} \|_2$, where $\|\mat{A}^{(t)} \|_2$ is the spectral norm of $\mat{A}^{(t)}$. Indeed, for player $x$ it holds that $\vu^{(t)}_x(\vy) \defeq - \mat{A}^{(t)} \vy$, thereby implying that $\| \vu_x^{(t)}(\vy) - \vu_x^{(t)}(\vy') \|_2 \leq \|\mat{A}^{(t)} \|_2 \|\vy - \vy'\|_2$, for any $\vy, \vy' \in \Y$, by definition of the spectral norm $\mat{A}^{(t)}$; similar reasoning applies for the utility vectors observed by player $y$ given that $\|(\mat{A}^{(t)})^\top\|_2 = \|\mat{A}^{(t)}\|_2$.  

\begin{theorem}[Sum of Regrets in BSPPs]
    \label{theorem:sum-zerosum}
    Suppose that both players employ $\OGD$ with initialization $\vx_k^{(t,0)} \defeq \sum_{s < t} \rvx_k^{(s)} \big/ {(t-1)}$, for $t \geq 2$, and learning rate $\eta \defeq \frac{1}{4L}$, where $L \defeq \max_{t \in [T]} \|\mat{A}^{(t)} \|_2 $. Then, the average sum of the players' regrets over all tasks can be upper bounded by
    \begin{equation*}
        \frac{1}{T} \sum_{t=1}^T \left( \reg_x^{(t, m)} + \reg_y^{(t, m)} \right) \leq 2 L \left( \hinsim^2_x + \hinsim^{2}_y \right) + \frac{8L(1 + \log T)}{T} \left( \diam^2_{\X} + \diam^2_{\Y} \right).
    \end{equation*}
\end{theorem}
Here, we recall that we use the notation $\hinsim_y^{2}, \hinsim_y^{2}$ for the task similarity in terms of the optimal in hindsight for player $x$ and $y$, respectively. We will also obtain results that depend on the similarity of the Nash equilibria. In \Cref{theorem:sum-zerosum} it is assumed that players know the Lipschitz constant of the game; such an assumption can be met by first rescaling the payoff matrix, so that $L$ will be a universal constant. Alternatively, one can employ the doubling trick (see our discussion after \Cref{cor:sum-regs}).

In light of the well-known connection that the sum of the players' regrets drives the rate of convergence of the average strategies to the set of Nash equilibria (\emph{e.g.}, see \citep{Freund97:A}), \Cref{theorem:sum-zerosum} yields the following consequence. 

\begin{corollary}[Average Duality Gap in BSPPs; Detailed Version of \Cref{theorem:informal-dualgap}]
    \label{cor:sum-zero-sum}
    In the setting of \Cref{theorem:sum-zerosum}, let $\Bar{\vx}^{(t)}, \bar{\vy}^{(t)}$ be the average strategy of the two players, respectively, at each task $t \in \range{T}$. Then,
    \begin{equation*}
        \frac{1}{T} \sum_{t=1}^T \dualgap^{(t)}(\bar{\vx}^{(t)}, \bar{\vy}^{(t)}) \leq \frac{2 L}{m} \left( \hinsim^{2}_x + \hinsim^{2}_y \right) + \frac{8L(1 + \log T)}{m T} \left( \diam^2_{\X} + \diam^2_{\Y} \right).
    \end{equation*}
\end{corollary}
Here, we used the notation
\begin{equation*}
    \dualgap^{(t)} : \X \times \Y \ni \hat{\vx} \times \hat{\vy} \mapsto \max_{\vy \in \Y} \hat{\vx}^{\top} \mat{A}^{(t)} \vy - \min_{\vx \in \X} \vx^\top \mat{A}^{(t)} \hat{\vy}.
\end{equation*}

\begin{proof}[Proof of \Cref{cor:sum-zero-sum}]
    The claim follows from \Cref{theorem:sum-zerosum}, using the fact that for any task $t \in \range{T}$,
    \begin{align*}
        \frac{1}{m} \left( \reg_x^{(t, m)} + \reg_y^{(t,m)} \right) &= \frac{1}{m} \left( \min_{\ystar \in \Y} \left\langle \ystar, \sum_{i=1}^m (\mat{A}^{(t)})^{\top} \vx^{(t,i)} \right\rangle - \max_{\xstar \in \X} \left\langle \xstar, \sum_{i=1}^m \mat{A}^{(t)} \vy^{(t,i)} \right\rangle  \right)  \\ 
        &= \dualgap^{(t)}\left( \frac{\sum_{i=1}^m \vx^{(t,m)}}{m}, \frac{\sum_{i=1}^m \vy^{(t,m)}}{m} \right). 
    \end{align*}
\end{proof}

This guarantee can significantly improve over the standard $m^{-1}$ rates in zero-sum games (when each game is treated separately)~\citep{Daskalakis15:Near}. We next provide bounds in terms of the individual regret of each player using a simple observation in~\citep{anagnostides2022last}. First, we connect the so-called second-order path lengths of the dynamics with the first term in the RVU bound (\Cref{thm:rvu}):  

\begin{proposition}[Bounded Second-Order Path Length]
    \label{prop:boundedpaths}
    Fix any task $t \in \range{T}$ associated with the bilinear saddle-point problem~\eqref{eq:bssp} under matrix $\mat{A}^{(t)}$. If both players employ $\OGD$ with learning rate $\eta \leq \frac{1}{4  \|\mat{A}^{(t)}\|_2}$, then
    \begin{equation*}
        \sum_{i=1}^m \| \pv{\vx}{t,i} - \pv{\vx}{t,i-1} \|_2^2 + \sum_{i=1}^m \| \pv{\vy}{t,i} - \pv{\vy}{t,i-1} \|_2^2 \leq 8 \left( \| \rvx^{(t)} - \vx^{(t,0)} \|_2^2 + \| \rvy^{(t)} - \vy^{(t,0)} \|_2^2 \right).
    \end{equation*}
\end{proposition}

\begin{proof}
    The claim is immediate from \Cref{cor:sum-regs}, using the fact that $\reg_x^{(t, m)} + \reg_y^{(t,m)} = \reg_x^{(t, m)}(\rvx^{(t)}) + \reg_y^{(t, m)}(\rvy^{(t)}) \geq 0$, by definition of $\rvx^{(t)} \in \X $ and $\rvy^{(t)} \in \Y$.
\end{proof}

When combined with the RVU bound, this proposition can be used to derive the following refined guarantee for the individual regret experienced by each player.

\begin{corollary}[Individual Per-Player Regret]
    \label{cor:ind-reg}
    In the setting of \Cref{theorem:sum-zerosum}, it holds that
    \begin{equation*}
        \frac{1}{T} \sum_{t=1}^T \max\left\{ \reg_x^{(t,m)}, \reg_y^{(t,m)} \right\} \leq 4 L \left( \hinsim_x^{2} + \hinsim_y^{2} \right) + \frac{16L(1 + \log T)}{T} \left( \diam_{\X}^2 + \diam_{\Y}^2 \right).
    \end{equation*}
\end{corollary}

\begin{proof}
    Let us first fix any task $t \in \range{T}$. Combining \Cref{thm:rvu} and \Cref{prop:boundedpaths}, we have that
    \begin{equation*}
        \max\left\{ \reg_x^{(t,m)}, \reg_y^{(t,m)} \right\} \leq 4 L \left( \|\rvx^{(t)} - \vx^{(t,0)}\|_2^2 + \|\rvy^{(t)} - \vy^{(t,0)}\|_2^2 \right).
    \end{equation*}
    Thus, the statement follows similarly to \Cref{theorem:sum-regs}.
\end{proof}

Analogous guarantees apply to games beyond two-player zero-sum, such as strategically zero-sum and zero-sum polymatrix games, using the same technique~\citep{anagnostides2022last}.

\subsubsection{Last-Iterate Bounds}

We next switch gears and focus on obtaining bounds on the number of iterations required so that the iterates of $\OGD$---instead of the average iterates---reach approximate Nash equilibria, the definition of which is recalled below for general multiplayer games.

\begin{definition}[Approximate Nash Equilibria]
    \label{def:NE}
    A joint strategy profile $\vx^{\star} \in \bigtimes_{k=1}^n \X_k$ is an $\epsilon$-approximate Nash equilibrium, with $\epsilon \geq 0$, if for any player $k \in \range{n}$ and any possible deviation $\vx_k \in \X_k$,
    \begin{equation*}
        u_k(\vx^{\star}) \geq u_k(\vx^{\star}_{-k}, \vx_k) - \epsilon.
    \end{equation*}
\end{definition}

Nash equilibria are known to exist under general assumptions~\citep{Rosen65:Existence}. In this context, we will need the following refined RVU bound, which can be readily extracted from~\citep{Syrgkanis15:Fast,Rakhlin13:Optimization}.

\begin{theorem}[\citep{Syrgkanis15:Fast,Rakhlin13:Optimization}]
\label{thm:refinedrvu}
    For any observed sequence of utilities $(\pv{\vu}{i})_{1 \leq i \leq m}$, the regret of $\OMD$ up to time $m \in \N$ can be bounded as 
    \begin{equation*}
        \pv{\reg}{m} \leq \frac{1}{\eta} \brg{}{\rvx}{\pv{\vx}{0}} + \eta \sum_{i=1}^{m} \| \pv{\vu}{i} - \pv{\vm}{i}\|_*^2 - \frac{1}{2\eta} \sum_{i=1}^m \| \pv{\vx}{i} - \pv{\hvx}{i} \|^2 - \frac{1}{2\eta} \sum_{i=1}^m \| \pv{\vx}{i} - \pv{\hvx}{i-1} \|^2,
    \end{equation*}
    where $\rvx \in \X$ is an optimal strategy in hindsight; $(\vx^{(i)})_{1 \leq i \leq m}$ is the primary sequence of strategies of $\OMD$, while $(\hvx^{(i)})_{1 \leq i \leq m}$ is the secondary sequence of strategies of $\OMD$.
\end{theorem}

Using this refined bound, the guarantee below follows analogously to \Cref{prop:boundedpaths}. For convenience, we will use the notation $\vz \defeq (\vx, \vy) \in \X \times \Y \eqqcolon \calZ$ in order to concatenate the strategies of the two players. The primary importance of the following refinement is that the second-order path length bound is now parameterized in terms of \emph{any Nash equilibrium} of the game.

\begin{corollary}[Refinement of~\Cref{prop:boundedpaths}]
    \label{cor:refined-paths}
    In the setting of~\Cref{prop:boundedpaths},
    \begin{equation*}
    \sum_{i=1}^m \|\vz^{(t, i)} - \hvz^{(t, i)}\|_2^2 + \sum_{i=1}^m \|\vz^{(t, i)} - \hvz^{(t, i-1)}\|_2^2 \leq 2 \| \vz^{(t, \star)} - \vz^{(t,0)} \|_2^2,
    \end{equation*}
    for any $\vz^{(t, \star)} \in \mathcal{Z}^{(t, \star)}$---the set of Nash equilibria of the BSPP associated with $\mat{A}^{(t)}$.
\end{corollary}

\begin{proof}
    Let $\reg^{(t,m)}(\vz) = \reg_x^{(t,m)}(\vx) + \reg_y^{(t,m)}(\vy)$, where $\vz = (\vx, \vy) \in \calZ$ and a fixed task $t \in \range{T}$. We claim that $\reg^{(t,m)}(\vz^{(t, \star)}) \geq 0$ for any Nash equilibrium pair $\vz^{(t, \star)} \in \mathcal{Z}^{(t, \star)}$. Indeed, for any step $i \in \range{m}$,
    \begin{equation}
        \label{eq:pointwise-regret}
        (\vx^{(t, i)})^\top \mat{A}^{(t)} \vy^{(t,i)} -(\vec{x}^{(t, \star)})^{\top} \mat{A}^{(t)} \vy^{(t,i)} + (\Vec{x}^{(t,i)})^\top \mat{A}^{(t)} \vy^{(t, \star)} - (\vx^{(t, i)})^\top \mat{A}^{(t)} \vy^{(t,i)} \geq 0.
    \end{equation}
    Here, if $v^{(t)}$ is the value of the game associated with $\mat{A}^{(t)}$, which exists by the minimax theorem, we used the fact that $(\vx^{(t, \star)})^\top \mat{A}^{(t)} \vy \leq v^{(t)}$, for any equilibrium strategy $\vx^{(t, \star)}$ for player $x$ and any $\vy \in \Y$, and similarly, $\vx^\top \mat{A}^{(t)} \vy^{(t, \star)} \geq v^{(t)}$, for any equilibrium strategy $\vy^{(t,\star)}$ for player $y$ and any $\vx \in \X$. Thus, summing~\eqref{eq:pointwise-regret} for all $i \in \range{m}$ verifies our claim: $\reg^{(t,m)}(\vz^{(t, \star)}) \geq 0$ for any $\vz^{(t, \star)} \in \mathcal{Z}^{(t, \star)}$. Finally, the statement follows analogously to \Cref{prop:boundedpaths}.
\end{proof}

Next, we combine this guarantee with \Cref{proposition:Khodak} to obtain an iteration-complexity bound that depends on the worst-case similarity of the equilibria, defined as
\begin{equation}
    \label{eq:worst-NE}
    \NEsim^{2} \defeq \max_{\vz^{(1,\star)}, \dots, \vz^{(T, \star)}} \min_{\bar{\vz} \in \calZ} \sum_{t=1}^T \| \vz^{(t, \star)} - \bar{\vz}\|_2^2,
\end{equation}
subject to the constraint that $\vz^{(t, \star)} \in \calZ^{(t, \star)}$ for each task $t \in \range{T}$. We recall that the set of Nash equilibria in BSSPs is nonempty convex and compact, and so \eqref{eq:worst-NE} is indeed well-defined. It is worth noting that in \emph{generic zero-sum games}---roughly speaking, any zero-sum game perturbed with random noise---there is a unique Nash equilibrium~\citep{VanDamme87:Stability}. In that case, there are no equilibrium selection issues, and \eqref{eq:worst-NE} happens to reduce to the smallest possible deviation of the Nash equilibria.

\begin{theorem}[Detailed Version of \Cref{theorem:informal-last}]
    \label{theorem:last-worst}
    Suppose that both players employ $\OGD$ with learning rate $\eta \leq \frac{1}{4L}$, where $L \defeq \max_{1 \leq t \leq T} \|\mat{A}^{(t)}\|_2$, and initialization $\vz^{(t, 0)} \defeq \sum_{s < t} \vz^{(t, \star)} \big/ (t-1)$, for any $\vz^{(t, \star)} \in \calZ^{(t, \star)}$ and $t \geq 2$. Then, for an average game $t \in \range{T}$
    \begin{equation*}
        \left\lceil \frac{2 \NEsim^{2}}{\epsilon^2} + \frac{8(1 + \log T) }{T\epsilon^2} \left( \diam^2_{\X} + \diam^2_{\Y} \right) \right\rceil
    \end{equation*}
    iterations suffice to reach an $\epsilon\left( \frac{2 \max\{\diam_{\X}, \diam_{\Y} \}}{\eta} + \|\mat{A}^{(t)} \|_2 \right)$-approximate Nash equilibrium, where $\NEsim^{2}$ is defined as in~\eqref{eq:worst-NE}.
\end{theorem}

\begin{proof}
    First, using \Cref{cor:refined-paths} for each game $t \in \range{T}$, 
    \begin{equation}
        \label{eq:avg-path}
    \frac{1}{T} \sum_{t=1}^T \left( \sum_{i=1}^m \|\vz^{(t, i)} - \hvz^{(t, i)}\|_2^2 + \sum_{i=1}^m \|\vz^{(t, i)} - \hvz^{(t, i-1)}\|_2^2 \right) \leq \frac{2}{T} \sum_{t=1}^T \| \vz^{(t, \star)} - \vz^{(t,0)} \|_2^2,
    \end{equation}
    for any sequence of Nash equilibria $(\vz^{(t,\star)})_{1 \leq t \leq T}$. Thus, analogously to \Cref{theorem:sum-regs}, the right-hand side of~\eqref{eq:avg-path} can be in turn upper bounded by
    \begin{equation*}
        2 \NEsim^{2} + \frac{8(1 + \log T) }{T} \left( \diam^2_{\X} + \diam^2_{\Y} \right).
    \end{equation*}
    As a result, for an average game
    \begin{equation*}
        \left\lceil \frac{2 \NEsim^{2}}{\epsilon^2} + \frac{8(1 + \log T) }{T\epsilon^2} \left( \diam^2_{\X} + \diam^2_{\Y} \right) \right\rceil
    \end{equation*}
    iterations suffice to reach an iterate such that $\|\vz^{(t,i)} - \hat{\vz}^{(t,i)}\|_2, \|\vz^{(t,i)} - \hat{\vz}^{(t,i-1)}\|_2 \leq \epsilon$, which in turn implies that $\vz^{(t,i)}$ is an $\epsilon\left( \frac{2 \max\{\diam_{\X}, \diam_{\Y} \}}{\eta} + \|\mat{A}^{(t)}\|_2 \right)$-approximate Nash equilibrium~\citep[Claim A.14]{anagnostides2022last}.
\end{proof}

Here, we have assumed that after the termination of each game the players obtain an exact Nash equilibrium of that game---in order to implement the initialization of~\Cref{theorem:last-worst}. If that is not the case, players have still learned an $O(1/m)$-approximate Nash equilibrium of the game after $m$ iterations of learning. So, \Cref{theorem:last-worst} can be modified by considering the task similarity metric~\eqref{eq:worst-NE} with respect to $O(1/m)$-approximate Nash equilibrium (\Cref{def:NE}) (instead of exact Nash equilibria). 

\subsubsection{Improving the Task Similarity}
\label{sec:improved-task}

The notion of task similarity used in \Cref{theorem:last-worst}, namely~\eqref{eq:worst-NE}, depends on the sequence of the worst Nash equilibria of the games. We can improve those guarantees under the assumption that the players observe the game $\game^{(t)}$ at the end of each task. In particular, each player can employ a meta-learning algorithm that observes as loss after task $t$ the function
\begin{equation}
    \label{eq:dist-set}
    \vx^{(t)} \mapsto \frac{1}{2} \min_{\vx^{(t,\star)} \in \X^{(t, \star)}} \| \vx^{(t)} - \vx^{(t,\star)} \|_2^2,
\end{equation}
where $\X^{(t, \star)}$ is the set of Nash equilibria of game $\game^{(t)}$ projected to $\X$. The function~\eqref{eq:dist-set} is easily seen to be convex, and its gradient can be computed if we have a (Euclidean) projection oracle for the set $\X^{(t,\star)}$. So, the meta-learning algorithm will perform at least as good as
\begin{equation}
    \label{eq:good-Nash}
    \min_{\vx \in \X} \frac{1}{2} \sum_{t=1}^T \min_{\vx^{(t,\star)} \in \X^{(t, \star)}} \| \vx - \vx^{(t,\star)} \|_2^2,
\end{equation}
modulo an $o(T)$ additive term (from the regret bound). Similar reasoning applies for player $y$. 

\paragraph{An illustrative example} To demonstrate the difference between the notions of task similarity---based on Nash equilibria---we have considered so far, we study a simple sequence of games. In particular, we let
\begin{equation}
    \label{eq:game-example}
    \game \defeq 
    \begin{bmatrix}
    1 & -1 & 0 \\
    -1 & 1 & 0 \\
    0 & 0 & 0
    \end{bmatrix} \quad \textrm{and} \quad \game' \defeq
    \begin{bmatrix}
    1.1 & -1.1 & 0 \\
    -1 & 1 & 0 \\
    0 & 0 & 0
    \end{bmatrix}
\end{equation}
be two zero-sum games in normal form described by their payoff matrix. Then, for $t \in \range{T}$, we let
\[
\game^{(t)} \defeq 
\begin{cases} 
\game \quad \textrm{ if } t \mod 2 = 1, \\
\game' \quad \textrm{if } t \mod 2 = 0.
\end{cases}
\]

The point of this example is that the games $\game$ and $\game'$ have a common Nash equilibrium, but also different ones:

\begin{claim}
    The strategy $((0,0,1), (0,0,1))$ is a Nash equilibrium in both $\game$ and $\game'$, defined in~\eqref{eq:game-example}. On the other hand, the pair of strategies $((\frac{1}{2}, \frac{1}{2}, 0),(\frac{1}{2}, \frac{1}{2}, 0))$ is a Nash equilibrium for $\game$, but it is $0.05$-far from being a Nash equilibrium in $\game'$. Conversely, the pair of strategies $((\frac{10}{21},\frac{11}{21}, 0), (\frac{1}{2}, \frac{1}{2}, 0))$ is a Nash equilibrium for $\game'$, but it is $\frac{1}{21}$-far from being a Nash equilibrium for $\game$. 
\end{claim}

Here, we say that a pair of strategies is $\alpha$-far from being a Nash equilibrium if there exists a unilateral deviation with (additve) benefit at least $\alpha > 0$ for that player. This simple example shows that the task similarity based on~\eqref{eq:good-Nash} can be $0$, while the improved task similarity~\eqref{eq:worst-NE} can be $\Omega(1)$.

\subsubsection{Further Refinements}

Moreover, we obtain further refinements when either the strategy set of each player corresponds to a probability simplex, or when the game is strongly convex-concave. In particular, let us treat the more general min-max optimization problem where $f(\vx,\vy)$ is an $L$-smooth, convex-concave and differentiable function; a more general setting that encompasses such problems is discussed in more detail in \Cref{sec:MVI}. We begin by noting the following general property of $\OMD$. For convenience, we use a prediction based on the secondary sequence of $\OMD$: $\vec{m}_x^{(i)} \defeq - \nabla_{\vx} f(\hvx^{(i-1)},\hvy^{(i-1)})$ and $\vec{m}_y^{(i)} \defeq \nabla_{\vy} f(\hvx^{(i-1)},\hvy^{(i-1)})$. 

\begin{proposition}[$\OMD$ Approaches the Set of NE]
    \label{prop:approach}
    For any learning rate $\eta \leq \frac{1}{4L}$ and for any iteration $i \in \range{m}$, $\OMD$ with $\vec{m}_x^{(i)} \defeq - \nabla_{\vx} f(\hvx^{(i-1)},\hvy^{(i-1)})$ and $\vec{m}_y^{(i)} \defeq \nabla_{\vy} f(\hvx^{(i-1)},\hvy^{(i-1)})$ satisfies
    \begin{equation}
        \label{eq:OMD-pot}
        \brg{x}{\xstar}{\hvx^{(i-1)}} + \brg{y}{\ystar}{\hvy^{(i-1)}} - \brg{x}{\xstar}{\hvx^{(i)}} - \brg{y}{\ystar}{\hvy^{(i)}} \geq 0,
    \end{equation}
    for any Nash equilibrium pair $(\xstar, \ystar) \in \X \times \Y$. In particular, equality in~\eqref{eq:OMD-pot} holds if and only $\hvx^{(i)} = \hvx^{(i-1)}$ and $\hvy^{(i)} = \hvy^{(i-1)}$.
\end{proposition}

Interestingly, this property follows by relying on the analysis of the RVU bound (\Cref{thm:rvu}), but only for a single iteration of $\OMD$.

\begin{proof}[Proof of \Cref{prop:approach}]
    First, for player $x$ we have that for any $\xstar \in \X$, the term $\langle \vx^{(i)} - \xstar,\nabla_{\vx} f(\vx^{(i)}, \vy^{(i)})$ is upper bounded by
    \begin{align}
        \label{eq:1rvu-x}
        \frac{1}{\eta} \left( \brg{x}{\xstar}{\hvx^{(i-1)}} - \brg{x}{\xstar}{\hvx^{(i)}} \right) &+ \langle \hvx^{(i)} - \vx^{(i)}, \nabla_{\vx} f(\hvx^{(i-1)}, \hvy^{(i-1)}) - \nabla_{\vx} f(\vx^{(i)}, \vy^{(i)}) \rangle \notag \\
        &- \frac{1}{2\eta} \left( \| \vx^{(i)} - \hvx^{(i-1)}\|^2 + \|\vx^{(i)} - \hvx^{(i)} \|^2 \right).
    \end{align}
    Similarly, we have that for any $\ystar \in \Y$ the term $\langle \ystar - \vy^{(i)}, \nabla_{\vy} f(\vx^{(i)}, \vy^{(i)}) \rangle$ is upper bounded by
    \begin{align}
        \label{eq:1rvu-y}
        \frac{1}{\eta} \left( \brg{y}{\ystar}{\hvy^{(i-1)}} - \brg{y}{\ystar}{\hvy^{(i)}} \right) &+ \langle \hvy^{(i)} - \vy^{(i)}, 
        \nabla_{\vy} f(\vx^{(i)}, \vy^{(i)}) - \nabla_{\vy} f(\hvx^{(i-1)}, \hvy^{(i-1)}) \rangle \notag \\
        &- \frac{1}{2\eta} \left( \| \vy^{(i)} - \hvy^{(i-1)}\|^2 + \|\vy^{(i)} - \hvy^{(i)} \|^2 \right).
    \end{align}
    As a result, for $\eta \leq \frac{1}{4L}$ it follows that $\langle \vx^{(i)} - \xstar,\nabla_{\vx} f(\vx^{(i)}, \vy^{(i)}) + \langle \ystar - \vy^{(i)}, \nabla_{\vy} f(\vx^{(i)}, \vy^{(i)}) \rangle$ is upper bounded by
    \begin{align*}
        \frac{1}{\eta} \left( \brg{x}{\xstar}{\hvx^{(i-1)}} - \brg{x}{\xstar}{\hvx^{(i)}} \right) + \frac{1}{\eta} \left( \brg{y}{\ystar}{\hvy^{(i-1)}} - \brg{y}{\ystar}{\hvy^{(i)}} \right) \\
        - \frac{1}{4\eta} \left( \| \vx^{(i)} - \hvx^{(i-1)}\|^2 + \|\vx^{(i)} - \hvx^{(i)} \|^2 \right) - \frac{1}{4 \eta} \left( \| \vy^{(i)} - \hvy^{(i-1)}\|^2 + \|\vy^{(i)} - \hvy^{(i)} \|^2 \right),
    \end{align*}
    due to \eqref{eq:1rvu-x} and \eqref{eq:1rvu-y}. Finally, the proof follows since $\langle \vx^{(i)} - \xstar,\nabla_{\vx} f(\vx^{(i)}, \vy^{(i)}) \rangle + \langle \ystar - \vy^{(i)}, \nabla_{\vy} f(\vx^{(i)}, \vy^{(i)}) \rangle \geq f(\vx^{(i)}, \ystar) - f(\xstar, \vy^{(i)}) \geq 0$, by convexity-concavity of $f$ along with the fact that $(\xstar, \ystar)$ is assumed to be a Nash equilibrium.
\end{proof}

While this proposition guarantees that $\OMD$ approaches the set of Nash equilibria---in the sense of~\eqref{eq:OMD-pot}, the improvement could be arbitrarily small. So, below we impose further structure on the problem. 

First, we assume that the objective of player $x$ is $\mu$-strongly convex when the strategy of player $y$ is fixed, and the objective of player $y$ is $\mu$-strongly concave when the strategy of player $x$ is fixed. Under that assumption, \eqref{eq:OMD-pot} can be further strengthened in that the improvement is not only strict, but increases with the modulus of strong convexity.
\begin{corollary}
    \label{cor:strong-conv}
    Let $\mu > 0$ be the modulus of strong convexity and strong concavity of $f(\cdot, \vy)$ and $f(\vx, \cdot)$, respectively. Then, in the setting of \Cref{prop:approach} with $\eta \leq \min \left\{ \frac{1}{4L}, \frac{1}{2\mu} \right\}$,
    \begin{equation}
        \label{eq:strong-pot}
        \left( \|\xstar - \hvx^{(i-1)} \|_2^2 + \|\ystar - \hvy^{(i-1)} \|_2^2 \right) \geq \left( 1 + \frac{\mu}{2} \right) \left( \|\xstar - \hvx^{(i)} \|_2^2 + \|\ystar - \hvy^{(i)} \|_2^2 \right),
    \end{equation}
    for any Nash equilibrium pair $(\xstar, \ystar) \in \X \times \Y$.
\end{corollary}
In proof, the key difference is that now $\langle \vx^{(i)} - \xstar,\nabla_{\vx} f(\vx^{(i)}, \vy^{(i)}) \rangle \geq f(\vx^{(i)}, \vy^{(i)}) - f(\xstar, \vy^{(i)}) + \frac{\mu}{2} \|\xstar - \vx^{(i)} \|_2^2$, by $\mu$-strong convexity of $f(\cdot, \vy^{(i)})$. Further, $\frac{1}{4\eta} \|\vx^{(i)} - \hvx^{(i)} \|_2^2 + \frac{\mu}{2} \|\vx^{(i)} - \xstar \|_2^2 \geq \frac{\mu}{4} \| \hvx^{(i)} - \xstar\|_2^2$, by Young's inequality and the fact that $\eta \leq \frac{1}{2\mu}$. Similar reasoning applies for player $y$, leading to \Cref{cor:strong-conv}.

Using~\eqref{eq:strong-pot} inductively yields that
\begin{equation*}
    \left( \|\xstar - \hvx^{(m)} \|_2^2 + \|\ystar - \hvy^{(m)} \|_2^2 \right) \leq \left( \frac{1}{1 +\mu/2} \right)^m \left( \|\xstar - \hvx^{(0)} \|_2^2 + \|\ystar - \hvy^{(0)} \|_2^2 \right),
\end{equation*}
for any Nash equilibrium pair $(\xstar, \ystar)$, in turn implying that $(\hvx^{(m)}, \hvy^{(m)})$ converges to the projection of $\calZ^{\star}$ with a linear rate of $1/(1 + \mu/2) \in (0,1)$. Perhaps surprisingly, linear rate is also achievable without any strong convexity assumptions in games with polyhedral sets~\citep[Theorem 8]{wei2020linear}, although the rate there depends on condition number-like quantities and can be arbitrarily slow even in $2 \times 2$ games.

From a meta-learning standpoint, that property of $\OMD$---converging to the projection to the set of Nash equilibria---is useful, and can partially address the bad example we show earlier in~\eqref{eq:game-example}. In particular, as long as the number of iterations is large enough, the dynamics will project to the set of Nash equilibria of $\game$ and $\game'$ in tandem, gradually approaching the common Nash equilibrium. Yet, if the ``angle'' between those sets is small it would take a large number of tasks $T$ so that the dynamics reach close to the common Nash equilibrium. 

\begin{remark}[Alternating Updates]
    So far, we have studied the setting where both players update their strategies simultaneously---as in the definition of~$\OMD$. A different approach that has received extensive attention in the literature~\citep{Wibisono22:Alternating}, not least due to its practical superiority~\citep{Tammelin14:Solving}, consists of performing the update rule in an alternating fashion. Interestingly, within the framework of optimistic mirror descent, alternation can be captured through the predictions: the first player uses the standard optimistic prediction (if any), but the player who updates second has a perfect prediction; that is, the prediction does not correspond to the previous strategy of the opponent, but the current one. All of our guarantees immediately apply under alternating updates using this simple observation.
\end{remark}

\subsection{Beyond Bilinear Saddle-Point Problems: A VI Perspective}
\label{sec:MVI}

In this subsection, we extend our scope beyond the bilinear saddle-point problems of~\eqref{eq:bssp}. In particular, we take a broader variational inequality (VI) perspective, which is commonly espoused in the context of min-max optimization.

Let us suppose that $F : \calZ \to \calZ$ is a singe-valued operator; many of our results below readily apply even if $F$ is multi-valued, such as the subdifferential operator of a non-differentiable function. We begin by making some standard assumptions on the operator $F$ below. We will then explain how the assumptions below immediately capture BSSPs. We remark that both of those assumptions will be weakened in the sequel: Lipschitz continuity (\Cref{assumption:F-Lip}) is relaxed in \Cref{appendix:Holder} where we only assume H\"older continuity, while the MVI property (\Cref{assumption:MVI}) is relaxed in \Cref{sec:weakMVI} where we consider the so-called weak MVI property.

\begin{assumption}
    \label{assumption:F-Lip}
    $F$ is $L$-Lipschitz continuous, in the sense that for any $\vz, \vz' \in \calZ$, it holds that $\|F(\vz) - F(\vz')\|_2 \leq L \|\vz - \vz'\|_2$.
\end{assumption}

\begin{assumption}[MVI Property~\citep{Mertikopoulos19:Optimistic,Gidel19:A}]
    \label{assumption:MVI}
    There exists a point $\vz^{\star} \in \calZ$ such that $\langle F(\vz), \vz - \vz^{\star} \rangle \geq 0$ for any $\vz \in \calZ$.
\end{assumption}

To relate those assumptions to the bilinear saddle-point problem we considered earlier in \Cref{sec:bssp}, we take $F : \vz \mapsto ( \mat{A} \vy, - \mat{A}^\top \vx)$. Then, given that $\langle \vz, F(\vz) \rangle = 0$ for any $\vz \in \calZ$, \Cref{assumption:MVI} requests the existence of a point $(\xstar, \ystar) \in \calZ$ so that $\vx^\top \mat{A} \vy^{\star} - (\vx^{\star})^\top \mat{A} \vy \geq 0$, for any $(\vx, \vy) \in \X \times \Y$, which is a well-known consequence of the minimax theorem; \Cref{assumption:F-Lip} also follows immediately for BSPPs. 

More broadly, \Cref{assumption:F-Lip,assumption:MVI} induce a standard setup in min-max optimization, where $F \defeq (\nabla_{\vx} f(\vx, \vy), - \nabla_{\vy} f(\vx, \vy))$ for certain differentiable and smooth objective functions $f : \X \times \Y \to \R$. Importantly, our techniques in \Cref{sec:bssp} can be readily applied to this more general setting. In particular, \Cref{cor:refined-paths} can be cast as follows.

\begin{corollary}
    \label{cor:MVI}
    Consider an operator $F : \calZ \to \calZ$ that satisfies \Cref{assumption:F-Lip,assumption:MVI}. Then, under $\OGD$ with learning rate $\eta \leq \frac{1}{4L}$,
    \begin{equation*}
    \sum_{i=1}^m \|\vz^{(t, i)} - \hvz^{(t, i)}\|_2^2 + \sum_{i=1}^m \|\vz^{(t, i)} - \hvz^{(t, i-1)}\|_2^2 \leq 2 \| \vz^{(t, \star)} - \vz^{(t,0)} \|_2^2,
    \end{equation*}
    where $\vz^{(t, \star)} \in \calZ$ is any point that satisfies the MVI property (\Cref{assumption:MVI}).
\end{corollary}

To analyze $\OGD$ in this more general setup, we define the regret up to time $m \in \N$ as
\begin{equation}
    \label{eq:linreg}
    \lreg{m} \defeq \max_{\zstar \in \calZ} \left \{ \sum_{i=1}^m \langle  \vz^{(i)} - \zstar, F(\vz^{(i)}) \rangle \right\}.
\end{equation}
When $F \defeq (\nabla_{\vx} f(\vx, \vy), - \nabla_{\vy} f(\vx, \vy))$, \eqref{eq:linreg} is precisely the sum of the ``linearized'' regrets incurred by the two players.

\begin{proof}[Proof of \Cref{cor:MVI}]
    First, analogously to~\Cref{thm:refinedrvu}, it follows that
    \begin{align*}
        \lreg{m}(\vz^{(t, \star)}) \leq \frac{1}{2 \eta} \|\vec{z}^{(t, \star)} - \vz^{(t, 0)}\|_2^2 + &\eta \sum_{i=1}^m \| F(\vz^{(t, i)}) - F(\vz^{(t, i-1)})\|_2^2 \\ 
        &- \frac{1}{2\eta} \sum_{i=1}^m \left( \|\vz^{(t,i)} - \hat{\vz}^{(t,i)}\|_2^2 + \|\vz^{(t,i)} - \hat{\vz}^{(t,i-1)}\|_2^2 \right).
    \end{align*}
    Next, using the $L$-Lipschitz continuity of $F$ (\Cref{assumption:F-Lip}), we have that for $\eta \leq \frac{1}{4L}$,
    \begin{equation}
        \label{eq:MVI}
        \lreg{m}(\vz^{(t, \star)}) \leq \frac{1}{2 \eta} \|\vec{z}^{(t, \star)} - \vz^{(t, 0)}\|_2^2 - \frac{1}{4\eta} \sum_{i=1}^m \left( \|\vz^{(t,i)} - \hat{\vz}^{(t,i)}\|_2^2 + \|\vz^{(t,i)} - \hat{\vz}^{(t,i-1)}\|_2^2 \right).
    \end{equation}
    But, if a point $\vz^{(t, \star)}$ satisfies the MVI property (\Cref{assumption:MVI}), it follows that
    \begin{equation*}
        \langle \vz^{(t,i)} - \vz^{(t, \star)}, F(\vz^{(t,i)}) \rangle \geq 0,
    \end{equation*}
    for any iteration $i \in \range{m}$, and summing over all $i \in \range{m}$ implies that $\lreg{m}(\vz^{(t, \star)}) \geq 0$. Thus, the result follows immediately by rearranging~\eqref{eq:MVI}.
\end{proof}

Before we proceed with the analog of \Cref{theorem:last-worst}, let us first clarify our solution concept. We will say that a point $\vz^{\star} \in \calZ$ is an $\epsilon$-approximate solution to the Stampacchia variational inequality (SVI) problem if 
\begin{equation}
    \label{eq:SVI}
    \tag{SVI}
    \langle \vz - \vz^{\star}, F(\vz^{\star}) \rangle \geq - \epsilon, \quad \forall \vz \in \calZ.
\end{equation}

We point out that when $F$ is a monotone operator, meaning that $\langle F(\vz) - F(\vz'), \vz - \vz' \rangle \geq 0$ for any $\vz, \vz' \in \calZ$, then an $\epsilon$-approximate solution to~\eqref{eq:SVI} also satisfies $\langle \vz - \vz^{\star}, F(\vz) \rangle \geq - \epsilon$, for any $\vz \in \calZ$; a point $\zstar$ satisfying the later property is referred to as an $\epsilon$-approximate \emph{weak} solution to the variational inequality problem. We will use the following property, which follows analogously to~\citep[Claim A.14]{anagnostides2022last}.

\begin{claim}
    \label{claim:approx-stat}
    Suppose that the sequences $(\vz^{(i)})_{0 \leq i \leq m}$ and $(\hvz^{(i)})_{0 \leq i \leq m}$ are updated using $\OGD$ with learning rate $\eta > 0$. If $\|\vz^{(i)} - \hvz^{(i-1)}\|_2, \|\vz^{(i)} - \hvz^{(i)}\|_2 \leq \epsilon$ for some iteration $i \in \range{m}$, then $\vz^{(i)} \in \calZ$ is an $\epsilon \left( \frac{2 \Omega_{\calZ}}{\eta} + \|F(\vz^{(i)}) \|_2 \right) $-approximate solution to the~\eqref{eq:SVI} problem.
\end{claim}

Now in a complete analogy to \Cref{eq:worst-NE}, we will use the following basic notion of task similarity.

\begin{equation}
    \label{eq:MVI-sim}
    \MVIsim^{2} \defeq \sup_{\vz^{(1,\star)}, \dots, \vz^{(T, \star)}} \min_{\bar{\vz} \in \calZ} \sum_{t=1}^T \| \vz^{(t, \star)} - \bar{\vz}\|_2^2,
\end{equation}
where $\vz^{(1, \star)}, \dots, \vz^{(T, \star)}$ are constrained to be in the (nonempty) set of points that satisfy the MVI property for the corresponding game.

\begin{theorem}[Extension of~\Cref{theorem:last-worst}]
    \label{theorem:last-MVI}
Consider an operator $F : \calZ \to \calZ $ that satisfies \Cref{assumption:F-Lip,assumption:MVI}. Suppose further that we employ $\OGD$ with learning rate $\eta \leq \frac{1}{4L}$, where $L \defeq \max_{1 \leq t \leq T} L^{(t)}$, and initialization $\vz^{(t, 0)} \defeq \sum_{s < t} \vz^{(t, \star)} \big/ (t-1)$, for any $\vz^{(t, \star)} \in \calZ^{(t, \star)}$ and $t \geq 2$, where $ \emptyset \neq \calZ^{(t, \star)} \subseteq \calZ$ is the set of points that satisfy the MVI property (\Cref{assumption:MVI}). Then, for an average game $t \in \range{T}$
    \begin{equation*}
        \left\lceil \frac{2 \MVIsim^{2}}{\epsilon^2} + \frac{8(1 + \log T) }{T\epsilon^2} \diam^2_{\calZ} \right\rceil
    \end{equation*}
    iterations suffice to reach an $\epsilon \left( \frac{2 \Omega_{\calZ}}{\eta} + \|F(\vz^{(t, i)}) \|_2 \right)$-approximate solution to the~\eqref{eq:SVI}, where $\MVIsim^{2}$ is defined as in~\eqref{eq:MVI-sim}.
\end{theorem}

It is also worth pointing out how to obtain in this setting an improved guarantee for the average sequence of $\OGD$. Let $\vz^{\star} \in \calZ$. Assuming that $F$ is a monotone operator, the average regret---in the sense of \eqref{eq:linreg}---can be lower bounded as
\begin{equation*}
    \frac{1}{m} \lreg{m}(\vz^{\star}) = \frac{1}{m} \sum_{i=1}^m \langle \vz^{(i)} - \vz^{\star}, F(\vz^{(i)}) \rangle \geq \frac{1}{m} \sum_{i=1}^m \langle \vz^{(i)} - \vz^{\star}, F(\vz^{\star}) \rangle = \langle \bar{\vz} - \vz^{\star}, F(\vz^{\star}) \rangle,
\end{equation*}
where $\bar{\vz} \defeq \sum_{i=1}^m \vz^{(i)} \big/ m$. Further, assuming that $F$ is Lipchitz continuous (\Cref{assumption:F-Lip}), following the proof of \Cref{theorem:last-MVI} we have that the average regret decays as $O(1/m)$ (assuming that $\calZ$ is bounded). In other words, assuming monotonicity we have that 
\begin{equation*}
    \langle \bar{\vz} - \vz, F(\vz) \rangle \leq O\left( \frac{1}{m} \right), \quad \forall \vz \in \calZ.
\end{equation*}
That is, $\bar{\vz}$ is an $O(1/m)$-approximate weak solution to the variational inequality problem.

\subsubsection{Weak MVI Property}
\label{sec:weakMVI}

Nevertheless, there are important problems for which the MVI property fails. For that reason, \citet{Diakonikolas21:Efficient} introduced a weaker property, recalled below.

\begin{assumption}[Weak MVI Property]
    \label{assumption:weakmvi}
    The operator $F : \calZ \to \calZ$ is said to satisfy the weak MVI property if there exists $\vz^{\star} \in \calZ$ such that 
    \begin{equation}
        \label{eq:weakMVI}
        \langle F(\vz), \vz - \vz^{\star} \rangle \geq - \frac{\rho}{2} \|F(\vz)\|_2^2, \quad \forall \vz \in \calZ,
    \end{equation}
    for a sufficiently small parameter $\rho > 0$.
\end{assumption}

We note that we use the $\ell_2$-norm in the right-hand side of~\eqref{eq:weakMVI} for convenience, although that definition can be stated more broadly. This is related to a condition known in the literature as \emph{cohypomonotonicity}~\citep{Bauschke21:Generalized,Combettes04:Proximal}.

\begin{definition}[Cohypomonotonicity]
    We say that an operator $F : \calZ \to \calZ$ is $\rho$-cohypomonotone, with $\rho > 0$, if for any $\vz, \vz' \in \calZ$,
    \begin{equation}
        \label{eq:cohy}
        \langle \vz - \vz', F(\vz) - F(\vz') \rangle \geq - \rho \| F(\vz) - F(\vz')\|_2^2.
    \end{equation}
\end{definition}

It is evident that~\eqref{eq:cohy} recovers the notion of monotonicity when $\rho = 0$, but cohypomonotonicity is a non-trivial extension of monotonicity. An operator that satisfies~\eqref{eq:cohy} is also sometimes referred to as \emph{$-\rho$-comonotone}.

As in~\citep{Diakonikolas21:Efficient}, here we focus on the unconstrained setting: $\calZ \defeq \R^{d}$, for some $d \in \N$. In that case, $\OGD$ can be simplified as
\begin{align*}
    \vz^{(i)} &\defeq \proj_{\calZ} \left( \hat{\vz}^{(i-1)} - \eta F(\vz^{(i-1)}) \right) = \hat{\vz}^{(i-1)} - \eta F(\vz^{(i-1)}), \\
    \hat{\vz}^{(i)} &\defeq \proj_{\calZ} \left( \hat{\vz}^{(i-1)} - \eta F(\vz^{(i)}) \right) = \hat{\vz}^{(i-1)} - \eta F(\vz^{(i)}).
\end{align*}

In this context, below we establish a bound on the number of iterations requires to make the norm of the operator at most $\epsilon$. We point out that in the special case of min-max optimization, $F(\vz) = (\nabla_{\vx} f(\vx, \vy), - \nabla_{\vy} f(\vx, \vy))$, we have that $\|F(\vz)\|_2^2 = \| \nabla_{\vx} f(\vx, \vy) \|_2^2 + \| \nabla_{\vy} f(\vx, \vy) \|_2^2$, which is perhaps the most natural measure of convergence in unconstrained min-max optimization.

\begin{theorem}[$\OGD$ under the weak MVI]
    \label{theorem:weak-MVI}
    If $F : \calZ \to \calZ$ is $L$-Lipschitz continuous (\Cref{assumption:F-Lip}) and satisfies the weak MVI property with parameter $0 < \rho < \frac{1}{8L}$ (\Cref{assumption:weakmvi}), then for $2 \rho < \eta < \frac{1}{4L}$, $\OGD$ satisfies
    \begin{equation*}
        \sum_{i=1}^{m-1} \|F(\vz^{(i)}) \|_2^2 \leq \frac{2}{\eta (\eta - 2 \rho)} \|\vec{z}^{ \star} - \vz^{(0)}\|_2^2 + \frac{2 \rho}{\eta - 2 \rho} \|F(\vz^{(m)}) \|_2^2,
    \end{equation*}
    where $\zstar \in \calZ$ is any point that satisfies the weak MVI property. In particular, there is an iterate $\vz^{(i)} \in \calZ$ such that 
    \begin{equation*}
        \|F(\vz^{(i)}) \|_2 \leq \frac{1}{\sqrt{m-1}} \sqrt{\frac{2}{\eta (\eta - 2 \rho)} \|\vec{z}^{ \star} - \vz^{(0)}\|_2^2 + \frac{2 \rho}{\eta - 2 \rho} \|F(\vz^{(m)}) \|_2^2}.
    \end{equation*}
\end{theorem}

\begin{proof}
    Similarly to the proof of \Cref{cor:MVI}, it follows that
    \begin{equation}
        \label{eq:weak-MVI}
        \lreg{m}(\vz^{\star}) \leq \frac{1}{2 \eta} \|\vec{z}^{ \star} - \vz^{(0)}\|_2^2 - \frac{1}{4\eta} \sum_{i=1}^m \left( \|\vz^{(i)} - \hat{\vz}^{(i)}\|_2^2 + \|\vz^{(i)} - \hat{\vz}^{(i-1)}\|_2^2 \right).
    \end{equation}
    But, the weak MVI property (\Cref{assumption:weakmvi}) implies that there exists $\zstar \in \calZ$ such that
    \begin{equation*}
        \lreg{m}(\vz^{\star}) \geq \sum_{i=1}^m \langle \vz^{(i)} - \zstar, F(\vz^{(i)}) \rangle \geq  - \frac{\rho}{2} \sum_{i=1}^m \|F(\vz^{(i)}) \|_2^2.
    \end{equation*}
    Thus, combining with~\eqref{eq:weak-MVI},
    \begin{equation*}
        0 \leq \frac{1}{2 \eta} \|\vec{z}^{ \star} - \vz^{(0)}\|_2^2 + \frac{\rho}{2} \sum_{i=1}^{m} \|F(\vz^{(i)}) \|_2^2 - \frac{\eta}{4} \sum_{i=1}^{m} \|F(\vz^{(i-1)}) \|_2^2,
    \end{equation*}
    since $\vz^{(i)} = \hat{\vz}^{(i-1)} - \eta F(\vz^{(i-1)})$. Thus,
    \begin{equation*}
        \left( \frac{\eta}{4} - \frac{\rho}{2} \right) \sum_{i=1}^{m-1} \|F(\vz^{(i)}) \|_2^2 \leq \frac{1}{2 \eta} \|\vec{z}^{ \star} - \vz^{(0)}\|_2^2 + \frac{\rho}{2} \|F(\vz^{(m)}) \|_2^2.
    \end{equation*}
    The proof then follows by rearranging, given that $ \eta > 2\rho$ (by assumption).
\end{proof}

In light of this theorem, the meta-learning version is completely analogous to \Cref{theorem:last-MVI}, and it is therefore omitted. An interesting question is to extend our analysis in~\Cref{theorem:weak-MVI} in the constrained setting.\footnote{Certain extensions were sebsequently reported by~\citet{Cai22:Accelerated}.}

\iffalse
Here we will show how such a guarantee can be extended beyond monotone operators, 
\fi

\subsection{Weighting the Strategies}

A common strategy used in order to compute faster (approximate) saddle-points consists of weighting the players' strategies in a non-uniform way (\emph{e.g.}, \citep{Zhang22:Equilibrium,Brown19:Solving,Gao21:Increasing}). In this subsection, we formalize the fact that the guarantees we established for the ergodic average (\Cref{cor:sum-zero-sum}) can still be applied under a broad class of weighting schemes. To establish this, we first provide a refined analysis of \OMD, but with the twist that our guarantee will apply for a \emph{weighted} notion of regret. Namely, for a sequence $\alpha^{(1)}, \dots, \alpha^{(m)} \in \R_{> 0}$, we define \emph{alpha regret} as
\begin{equation}
    \label{eq:a-reg}
    \areg^{(m)} \defeq \max_{\rvx \in \X} \left\{ \sum_{i=1}^m \alpha^{(i)} \langle \rvx - \vx^{(i)}, \vu^{(i)} \rangle \right\}.
\end{equation}

This generalized notion of regret has been proved useful, for example, in obtaining accelerates rates for the Frank-Wolfe algorithm in convex optimization~\citep{Abernethy18:Faster}; see also~\citep{Wang21:No}. We will insist on ensuring that $\sum_{i=1}^m \alpha^{(i)} = m$. Typical choices used in the literature include the so-called \emph{linear averaging}, in which $\alpha^{(i)} \defeq \frac{2i}{m+1}$, or \emph{quadratic averaging}, in which $\alpha^{(i)} \defeq \frac{6i^2}{(m+1)(2m+1)}$. Below, we show that $\OMD$ enjoys similar RVU-type guarantees (\Cref{thm:rvu}) under a broad family of weighting sequences; this is perhaps unexpected given that $\OMD$ is designed to minimize (unweighted) regret.

\begin{theorem}
    \label{theorem:weighted-rvu}
    Consider any $\vec{\alpha} \in \R^m_{> 0}$ with $\alpha^{(1)} \leq \dots \leq \alpha^{(m)}$, and suppose that we use $\OMD$ with a $1$-strongly convex regularizer $\calR$ with respect to the norm $\|\cdot\|$. For any observed sequence of utilities $(\pv{\vu}{i})_{1 \leq i \leq m}$, the alpha regret $\areg^{(m)}$ of $\OMD$ up to time $m \in \N$ and initialized at $\hvx^{(0)} \in \X$ can be upper bounded by
    \begin{align*}
        \frac{\alpha^{(1)}}{\eta} \brg{}{\rvx}{\hvx^{(0)}} + \frac{(\alpha^{(m)} - \alpha^{(1)})}{\eta} \max_{1 \leq i \leq m-1} \brg{}{\rvx}{\hvx^{(i)}} + \eta \sum_{i=1}^m \alpha^{(i)} \| \vec{u}^{(i)} - \vec{m}^{(i)} \|_*^2 \\
        - \frac{1}{2\eta} \sum_{i=1}^m \alpha^{(i)} \| \vx^{(i)} - \hvx^{(i-1)} \|^2-\frac{1}{2\eta} \sum_{i=1}^m \alpha^{(i)} \| \vx^{(i)} - \hvx^{(i)} \|^2,
    \end{align*}
    for any $\rvx \in \X$.
\end{theorem}

The proof of this theorem closely follows the argument due to~\citet{Rakhlin13:Optimization}, but it is included below for completeness. 

\begin{proof}
    Fix any $\xstar \in \X$. First, by $1$-strong convexity of the Bregman divergence (with respect to the first argument), we have that for any iteration $i \in \range{m}$,
    \begin{equation}
        \label{eq:strong1}
        \langle \vx^{(i)}, \vec{m}^{(i)} \rangle - \frac{1}{2\eta} \| \vx^{(i)} - \hvx^{(i-1)} \|^2 - \langle \hvx^{(i)}, \vec{m}^{(i)} \rangle + \frac{1}{\eta} \brg{}{\hvx^{(i)}}{\hvx^{(i-1)}} \geq \frac{1}{2\eta} \| \vec{x}^{(i)} - \hvx^{(i)} \|^2.
    \end{equation}
    Similarly, for any iteration $i \in \range{m}$,
    \begin{equation}
        \label{eq:strong2}
        \langle \hvx^{(i)}, \vec{u}^{(i)} \rangle - \frac{1}{\eta} \brg{}{\hvx^{(i)}}{\hvx^{(i-1)}} - \langle \rvx, \vec{u}^{(i)} \rangle + \frac{1}{\eta} \brg{}{\rvx}{\hvx^{(i-1)}} \geq \frac{1}{\eta} \brg{}{\rvx}{\hvx^{(i)}}.
    \end{equation}
    Further, for any step $i \in \range{m}$ it holds that
    \begin{equation*}
    \langle \rvx - \vec{x}^{(i)}, \vec{u}^{(i)} \rangle = \langle \rvx, \vec{u}^{(i)} \rangle - \langle \hvx^{(i)}, \vec{u}^{(i)} \rangle - \langle \vec{x}^{(i)}, \vec{m}^{(i)} \rangle + \langle \hvx^{(i)}, \vec{m}^{(i)} \rangle + \langle \hvx^{(i)} - \vec{x}^{(i)}, \vec{u}^{(i)} - \vec{m}^{(i)} \rangle.
\end{equation*}
    Combining with \eqref{eq:strong1} and \eqref{eq:strong2},
    \begin{align}
        \areg^{(m)}(\rvx) \leq &\frac{1}{\eta} \sum_{i=1}^m \alpha^{(i)} \left( \brg{}{\rvx}{\hvx^{(i-1)}}  - \brg{}{\rvx}{\hvx^{(i)}} \right) + \sum_{i=1}^m \alpha^{(i)} \langle \hvx^{(i)} - \vec{x}^{(i)}, \vec{u}^{(i)} - \vec{m}^{(i)} \rangle \notag \\
        &- \frac{1}{2\eta} \sum_{i=1}^m \alpha^{(i)} \| \vx^{(i)} - \hvx^{(i-1)} \|^2-\frac{1}{2\eta} \sum_{i=1}^m \alpha^{(i)} \| \vx^{(i)} - \hvx^{(i)} \|^2.\label{eq:alpha-first}
    \end{align}
    It is also easy to see that $\| \hvx^{(i)} - \vx^{(i)}\| \leq \eta \|\vu^{(i)} - \vec{m}^{(i)}\|_*$. So, it suffices to appropriately bound the first term in the right-hand side of~\eqref{eq:alpha-first}. We have
    \begin{align*}
        \sum_{i=1}^m \alpha^{(i)} \left( \brg{}{\rvx}{\hvx^{(i-1)}}  - \brg{}{\rvx}{\hvx^{(i)}} \right) = \alpha^{(1)} \brg{}{\rvx}{\hvx^{(0)}} - \alpha^{(m)} \brg{}{\rvx}{\hvx^{(m)}} \\
        + \sum_{i=2}^{m} \brg{}{\rvx}{\hvx^{(i-1)}} (\alpha^{(i)} - \alpha^{(i-1)}), 
    \end{align*}
    which in turn is upper bounded by
    \begin{equation*}
        \alpha^{(1)} \brg{}{\rvx}{\hvx^{(0)}} + (\alpha^{(m)} - \alpha^{(1)}) \max_{1 \leq i \leq m-1} \brg{}{\rvx}{\hvx^{(i)}}.
    \end{equation*}
    Here, we used the fact that $\alpha^{(m)} \geq 0$ and $\brg{}{\cdot}{\cdot} \geq 0$, as well as the assumption that $\alpha^{(i)} - \alpha^{(i-1)} \geq 0$ for any $2 \leq i \leq m$, which results in the telescopic sum above. Combing with~\eqref{eq:alpha-first} concludes the proof.
\end{proof}

This characterization has some interesting implications. First, it lifts many of the results of~\citet{Syrgkanis15:Fast} to any class of monotone weighting schemes (\Cref{theorem:weighted-rvu}), even though the underlying algorithm $(\OMD)$ remains the same. For example, we state the following immediate consequence of the sum of the players' regets.

\begin{corollary}
    Suppose that all players employ $\OGD$ with learning rate $\eta \leq \frac{1}{4  L \sqrt{n-1}}$.
    \begin{itemize}[leftmargin=1cm]
        \item If $\alpha^{(i)} \defeq \frac{2i}{m+1}$, then $\frac{1}{m} \sum_{k=1}^n \areg_k^{(m)}\leq 4L\sqrt{n-1}\sum_{k=1}^n \diam^2_{\X_k} \frac{1}{m}$.
        \item If $\alpha^{(i)} \defeq \frac{6i^2}{(m+1)(2m+1)}$, then $\frac{1}{m} \sum_{k=1}^n \areg_k^{(m)}\leq 12 L\sqrt{n-1}\sum_{k=1}^n \diam^2_{\X_k} \frac{m^2 + 1}{m(m+1)(2m+1)}$.
    \end{itemize}
\end{corollary}

This is important since, as we explained earlier, linear or quadratic averaging have shown to substantially improve performance in practice. 
From a meta-learning standpoint, it is possible that assigning different weights to different tasks could improve performance in practice, although this is not pursued in this paper.

\subsection{Adaptive Regularization}

In this subsection, we extend our scope to $\OGD$~\citep{Chiang12:Online,Rakhlin13:Optimization} under adaptive regularization, or \emph{preconditioning}. That is, we will endow $\OGD$ with an \emph{adaptive} regularizer, leading to $\optada$, an extension of the popular $\adagrad$ algorithm~\citep{Duchi11:Adaptive}. To this end, let $\preco^{(i)} \in \R^{d \times d}$ be a positive definite and symmetric matrix, for any $i \in \range{m}$, that will serve as the \emph{preconditioner}. We also let $\| \vec{x} \|_{\preco^{(i)}} \defeq \sqrt{\vec{x}^\top \preco^{(i)} \vec{x}}$ be the (Mahalanobis) norm induced by $\preco^{(i)}$, and $\calR^{(i)}(\vec{x}) \defeq \frac{1}{2} \| \vec{x}\|^2_{\preco^{(i)}}$ be the associated \emph{regularizer}. If we denote by $\proj^{\preco^{(t)}}_{\X}(\vec{x}) \defeq \argmin_{\hvx \in \X} \|\vec{x} - \hvx \|^2_{\preco^{(i)}}$, $\optada$ can be expressed via the following update rule for all $i \in \N$.
\begin{equation}
    \label{eq:optada}
    \tag{OptAdaGrad}
    \begin{split}
        \vec{x}^{(i)} &\defeq \proj_{\X}^{\preco^{(i)}} \left( \hvx^{(i-1)} + (\preco^{(i)})^{-1} \vec{m}^{(i)} \right), \\
        \hvx^{(i)} &\defeq \proj_{\X}^{\preco^{(i)}} \left( \hvx^{(i-1)} + (\preco^{(i)})^{-1} \vec{u}^{(i)} \right).
    \end{split}
\end{equation}
Further, we define $\vec{x}^{(0)} = \hvx^{(0)} \defeq \min_{\vec{x} \in \X} \|\vec{x}\|_2^2$. When $\preco^{(t)} = \frac{1}{\eta} \mat{I}_{d \times d}$, for some \emph{learning rate} $\eta > 0$, $\optada$ reduces to $\OGD$. While this algorithm can be impractical in high dimensions, if we restrict the preconditioner $\preco^{(i)}$ to be, for example, a diagonal matrix, the inverse can be computed in linear time. A central theme in \adagrad~\citep{Duchi11:Adaptive} is that the preconditioner $\preco^{(i)}$ has to vary slowly over time. In our setting, we formalize that requirement in the following condition.
\begin{condition}
    \label{condition:smooth-preco}
Fix some time horizon $m \geq 2$. The sequence of preconditioners $(\preco^{(i)})_{1 \leq i \leq m}$ is such that $\sum_{i=1}^{m-1} \|\preco^{(i+1)} - \preco^{(i)} \|_2 \leq \spnorm(m)$, for some function $\spnorm(m) = o_m(m)$.
\end{condition}
Before we proceed with the regret analysis of $\optada$, let us make some observations. First, $\cR^{(i)}$ is $1$-strongly convex with respect to the norm $\|\cdot\|_{\preco^{(i)}}$ since
\begin{equation*}
    \cR^{(i)}(\vec{x}) \geq \cR^{(i)}(\vx') + \langle \nabla \cR^{(i)}(\vec{x}'), \vec{x} - \vec{x}' \rangle + \frac{1}{2} \| \vec{x} - \vec{x}' \|^2_{\preco^{(i)}},
\end{equation*}
for any $i \in \range{m}$ and $\vx, \vx' \in \X$. In this setup, it will be convenient to express $\optada$ in the following form for $i \in \range{m}$:
\begin{equation}
    \label{eq:OMD}
    \begin{split}
        \Vec{x}^{(i)} &\defeq \argmax_{\Vec{x} \in \X} \left\{ \Phi^{(i)} \defeq \langle \Vec{x}, \Vec{m}^{(i)} \rangle - \frac{1}{2} (\Vec{x} - \hvx^{(i-1)})^\top \preco^{(i)} (\Vec{x} - \hvx^{(i-1)})  \right\}, \\
        \hvx^{(i)} &\defeq \argmax_{\hvx \in \X} \left\{ \Psi^{(i)} \defeq \langle \hvx, \vec{u}^{(i)} \rangle - \frac{1}{2} (\hvx - \hvx^{(i-1)})^\top \preco^{(i)} (\hvx - \hvx^{(i-1)}) \right\}.
    \end{split}
\end{equation}

\begin{theorem}
    \label{theorem:ada-reg}
Suppose that \Cref{condition:smooth-preco} holds for some function $\sigma(m)$. Then, the regret of $\optada$ up to time $m \in \N$ can be bounded as
\begin{align}
    \label{eq:opt-reg}
    \reg^{(m)}(\rvx) \leq \frac{1}{2} \| \rvx - \vx^{(0)} \|^2_{\mat{Q}^{(1)}} +\! \frac{\diam^2_{\X}}{2} \spnorm(m) + \sum_{i=1}^m \| \vec{u}^{(i)} - \vec{m}^{(i)} \|^2_{*, \preco^{(i)}} 
    - \frac{1}{2} \Sigma^{(m)},
\end{align}
where 
\begin{equation*}
    \Sigma^{(m)} \defeq \sum_{i=1}^m \left( \|\vec{x}^{(i)} - \hvx^{(i-1)} \|^2_{\preco^{(i)}} + \| \vec{x}^{(i)} - \hvx^{(i)} \|^2_{\preco^{(i)}} \right).
\end{equation*}
\end{theorem}
This regret bound differs from the analysis of $\adagrad$~\citep{Duchi11:Adaptive} in that we incorporate optimism. \Cref{theorem:ada-reg} also goes beyond the analysis of~\citep{Rakhlin13:Optimization} in that the regularizer varies over time, while we recover the guarantee of $\OGD$ in the special case where $\spnorm(m) = 0$.

\begin{proof}
First, given that $\cR^{(i)}$ is $1$-strongly convex with respect to the norm $\|\cdot\|_{\preco^{(i)}}$, it follows that both $\Phi^{(i)}$ and $\Psi^{(i)}$ (recall the definition in \eqref{eq:OMD}) are also $1$-strongly convex with respect to $\|\cdot\|_{\preco^{(i)}}$. In turn, this implies that
\begin{equation}
    \label{eq:x-strongly}
    \langle \vec{x}^{(i)}, \vec{m}^{(i)} \rangle - \frac{1}{2} \| \vec{x}^{(i)} - \hvx^{(i-1)} \|^2_{\preco^{(i)}} - \langle \hvx^{(i)}, \vec{m}^{(i)} \rangle + \frac{1}{2} \| \hvx^{(i)} - \hvx^{(i-1)} \|^2_{\preco^{(i)}} \geq \frac{1}{2} \| \vec{x}^{(i)} - \hvx^{(i)} \|^2_{\preco^{(i)}};
\end{equation}
and
\begin{equation}
    \label{eq:hatx-strongly}
    \langle \hvx^{(i)}, \vec{u}^{(i)} \rangle - \frac{1}{2} \| \hvx^{(i)} - \hvx^{(i-1)} \|^2_{\preco^{(i)}} - \langle \rvx, \vec{u}^{(i)} \rangle + \frac{1}{2} \| \vec{x}^* - \hvx^{(i-1)} \|^2_{\preco^{(i)}} \geq \frac{1}{2} \|  \hvx^{(i)} - \rvx \|^2_{\preco^{(i)}},
\end{equation}
for any $\rvx \in \X$. Furthermore, we have that
\begin{equation*}
    \langle \rvx - \vec{x}^{(i)}, \vec{u}^{(i)} \rangle = \langle \rvx, \vec{u}^{(i)} \rangle - \langle \hvx^{(i)}, \vec{u}^{(i)} \rangle - \langle \vec{x}^{(i)}, \vec{m}^{(i)} \rangle + \langle \hvx^{(i)}, \vec{m}^{(i)} \rangle + \langle \hvx^{(i)} - \vec{x}^{(i)}, \vec{u}^{(i)} - \vec{m}^{(i)} \rangle.
\end{equation*}
As a result, combining this identity with \eqref{eq:x-strongly} and \eqref{eq:hatx-strongly} we get that for any $\rvx \in \X$,
\begin{align}
    \reg^{(m)}(\rvx) \leq \frac{1}{2} &\sum_{i=1}^m \left( \| \rvx - \hvx^{(i-1)} \|^2_{\preco^{(i)}} - \| \rvx - \hvx^{(i)} \|^2_{\preco^{(i)}} \right) + \sum_{i=1}^m \| \hvx^{(i)} - \vec{x}^{(i)} \|_{\preco^{(i)}} \| \vec{u}^{(i)} - \vec{m}^{(i)} \|_{*, \preco^{(i)}} \notag \\
    &- \frac{1}{2} \sum_{i=1}^m \left( \|\vec{x}^{(i)} - \hvx^{(i-1)} \|^2_{\preco^{(i)}} + \| \vec{x}^{(i)} - \hvx^{(i)} \|^2_{\preco^{(i)}} \right), \label{eq:long-regbound}
\end{align}
where we also used that $\langle \hvx^{(i)} - \vec{x}^{(i)}, \vec{u}^{(i)}- \vec{m}^{(i)} \rangle \leq \|\hvx^{(i)} - \vec{x}^{(i)} \|_{\preco^{(i)}} \|\vec{u}^{(i)} - \vec{m}^{(i)} \|_{*, \preco^{(i)}}$ (by Cauchy-Schwarz inequality). Next, to further bound \eqref{eq:long-regbound} we will show the following simple claims.

\begin{claim}
    \label{claim:u-m}
It holds that
\begin{equation*}
    \|\hvx^{(i)} - \vec{x}^{(i)} \|_{\preco^{(i)}} \leq \|\vec{u}^{(i)} - \vec{m}^{(i)} \|_{*, \preco^{(i)}}.
\end{equation*}
\end{claim}

\begin{proof}
By $1$-strong convexity of $\Phi^{(i)}$ with respect to $\|\cdot\|_{\preco^{(i)}}$,
\begin{equation}
    \label{eq:x-hatx}
    \langle \vec{x}^{(i)}, \vec{m}^{(i)} \rangle - \frac{1}{2} \| \vec{x}^{(i)} - \hvx^{(i-1)} \|^2_{\preco^{(i)}} - \langle \hvx^{(i)}, \vec{m}^{(i)} \rangle + \frac{1}{2} \| \hvx^{(i)} - \hvx^{(i-1)} \|^2_{\preco^{(i)}} \geq \frac{1}{2} \| \vec{x}^{(i)} - \hvx^{(i)} \|^2_{\preco^{(i)}};
\end{equation}
Similarly,
\begin{equation}
    \label{eq:hatx-x}
    \langle \hvx^{(i)}, \vec{u}^{(i)} \rangle - \frac{1}{2} \| \hvx^{(i)} - \hvx^{(i-1)} \|^2_{\preco^{(i)}} - \langle \vec{x}^{(i)}, \vec{u}^{(i)} \rangle + \frac{1}{2} \| \vec{x}^{(i)} - \hvx^{(i-1)} \|^2_{\preco^{(i)}} \geq \frac{1}{2} \|  \hvx^{(i)} - \vec{x}^{(i)} \|^2_{\preco^{(i)}}.
\end{equation}
Hence, summing \eqref{eq:x-hatx} and \eqref{eq:hatx-x} yields that
\begin{equation*}
    \langle \vec{x}^{(i)} - \hvx^{(i)}, \vec{m}^{(i)} - \vec{u}^{(i)} \rangle \geq \| \vec{x}^{(i)} - \hvx^{(i)} \|^2_{\preco^{(i)}} \implies \| \vec{x}^{(i)} - \hvx^{(i)} \|_{\preco^{(i)}} \leq \| \vec{u}^{(i)} - \vec{m}^{(i)} \|_{*, \preco^{(i)}},
\end{equation*}
by Cauchy-Schwarz inequality.
\end{proof}

\begin{claim}
    \label{claim:smooth-preco}
Suppose that $\sum_{i=1}^{m-1} \|\preco^{(i+1)} - \preco^{(i)} \|_2 \leq \spnorm(m)$. Then, 
\begin{equation*}
    \frac{1}{2} \sum_{i=1}^m \left( \| \rvx - \hvx^{(i-1)} \|^2_{\preco^{(i)}} - \| \rvx - \hvx^{(i)} \|^2_{\preco^{(i)}} \right) \leq \frac{1}{2} \| \rvx - \hvx^{(0)} \|^2_{\mat{Q}^{(1)}} + \frac{\diam^2_{\X}}{2} \spnorm(m).
\end{equation*}
\end{claim}

\begin{proof}
First, we observe that
\begin{align*}
    \frac{1}{2} \sum_{i=1}^m \left( \| \rvx - \hvx^{(i-1)} \|^2_{\preco^{(i)}} - \| \rvx - \hvx^{(i)} \|^2_{\preco^{(i)}} \right) \leq \frac{1}{2} \sum_{i=2}^m (\rvx - \hvx^{(i-1)})^\top (\preco^{(i)} - \preco^{(i-1)}) (\rvx - \hvx^{(i-1)}) \\
    + \frac{1}{2} \| \rvx - \hvx^{(0)} \|^2_{\mat{Q}^{(1)}} - \frac{1}{2} \|\rvx - \hvx^{(m)} \|^2_{\preco^{(m)}}.
\end{align*}
The first term on the right-hand side can be further bounded as
\begin{align*}
    (\rvx - \hvx^{(i-1)})^\top (\preco^{(i)} - \preco^{(i-1)}) (\rvx - \hvx^{(i-1)}) &\leq \| \rvx - \hvx^{(i-1)} \|_2 \| (\preco^{(i)} - \preco^{(i-1)}) (\rvx - \hvx^{(i-1)}) \|_2 \\
    &\leq \| \preco^{(i)} - \preco^{(i-1)} \|_2 \|\rvx - \hvx^{(i-1)} \|_2^2 \\
    &\leq \diam_{\X}^2 \| \preco^{(i)} - \preco^{(i-1)} \|_2.
\end{align*}
As a result, we conclude that
\begin{equation*}
    \frac{1}{2} \sum_{i=1}^m \left( \| \rvx - \hvx^{(i-1)} \|^2_{\preco^{(i)}} - \| \rvx - \hvx^{(i)} \|^2_{\preco^{(i)}} \right) \leq \frac{1}{2} \| \rvx - \hvx^{(0)} \|^2_{\mat{Q}^{(1)}} + \frac{\diam^2_{\X}}{2} \spnorm(m).
\end{equation*}
\end{proof}
Finally, combining \Cref{claim:u-m} and \Cref{claim:smooth-preco} with \eqref{eq:long-regbound} completes the proof.
\end{proof}

An important advantage of this guarantee is that the term in~\eqref{eq:opt-reg} that captures the misprediction error depends on the sequence of preconditioners, which could be potentially selected in a dynamic way to minimize that term.

\subsection{Stochastic Games}
\label{sec:stochastic}

In this subsection, we study another settings that has received considerable attention, especially in recent years; namely, stochastic games~\citep{Shapley53:Stochastic}. In particular, we will extend our previous results in (infinite-horizon) two-player zero-sum stochastic games, but with an important caveat that will be explained in the sequel. Such games are not covered by our previous techniques, as it will become clear shortly.

For the sake of simplicity in the exposition, we will focus on a simple case of stochastic games that---in a certain sense---captures the difficulty of the problem, known as Von Neumann's ratio game~\citep{Neumann45:A}, defined as
\begin{equation}
    \label{eq:ratio}
    V(\vx, \vy) \defeq \frac{\vx^\top \mat{R} \vy}{\vx^\top \mat{S} \vy},
\end{equation}
where $\vx \in \Delta^{d_x}$; $\vy \in \Delta^{d_y}$; and $\mat{R}, \mat{S} \in \R^{d_x \times d_y}$. It is further assumed that for any $\vx \in \Delta^{d_x}$ and $\vy \in \Delta^{d_y}$ it holds that $\vx^\top \mat{S} \vy \geq \zeta > 0$, which ensures that the objective~\eqref{eq:ratio} is indeed well-defined. The ratio game~\eqref{eq:ratio} can be interpreted as a single-state stochastic game, in which the immediate reward under actions $(a_x, a_y) \in \calA_x \times \calA_y$ is given by $\mat{R}[a_x, a_y]$, while $\mat{S}[a_x,a_y]$ represents the probability of stopping at reach round. We point out that the subsequent analysis can be extended to general infinite-horizon (discounted) two-player zero-sum stochastic games by appropriately performing the analysis for each state. 

The ratio game~\eqref{eq:ratio} evidently captures bilinear saddle-point problems (studied in \Cref{sec:bssp}) in the special case when $\vx^\top \mat{S} \vy = 1$ for any $(\vx, \vy) \in \Delta^{d_x} \times \Delta^{d_y}$, but the objective~\eqref{eq:ratio} is in general nonconvex-nonconcave. In fact, as pointed out in~\citep[Proposition 2]{Daskalakis20:Independent}, the MVI property (\Cref{assumption:MVI}) could fail in such games. It is subsequently an open problem to understand whether $\OGD$ converges in such games~\citep[Open Problem 1]{Daskalakis20:Independent}. Here, we will show that there exists a learning rate schedule, time-varying but non-vanishing, that ensures that $\OGD$ reaches a minimax equilibrium of~\eqref{eq:ratio}; the main caveat is that we have not been able to identify such a learning rate schedule in an algorithmically useful way.

Let us first state some useful properties. Two-player zero-sum (discounted) stochastic games, and in particular Von Neumann's ratio game~\eqref{eq:ratio}, admit a minimax theorem, as was shown in the pioneering work of~\citet{Shapley53:Stochastic}:

\begin{fact}[\citep{Shapley53:Stochastic,Neumann45:A}]
    \label{fact:zero-dual}
    Let $V : \Delta^{d_x} \times \Delta^{d_y} \to \R$ be defined as in~\eqref{eq:ratio}. Then, 
    \begin{equation*}
        \min_{\vx^{\star} \in \X} \max_{\vy^{\star} \in \Y} V(\vx^{\star}, \vy^{\star}) = \max_{\vy^{\star} \in \Y} \min_{\vx^{\star} \in \X} V(\vx^{\star}, \vy^{\star}).
    \end{equation*}
\end{fact}

The second useful property we will require is the so-called \emph{gradient dominance property} (see~\citep{Agarwal21:On} and references therein), which is formalized below.

\begin{fact}[Gradient Dominance~\citep{Agarwal21:On}]
    \label{fact:gd}
    Let $V : \Delta^{d_x} \times \Delta^{d_y} \to \R$ be defined as in~\eqref{eq:ratio}. There is a parameter $C \in \R_{> 0}$ such that for any fixed $\vy \in \Delta^{d_y}$, 
    \begin{equation}
        \label{eq:gd}
        V(\vx, \vy) - \min_{\xstar \in \X} V(\xstar, \vy) \leq C \max_{\xstar \in \X} \langle \vx - \xstar, \nabla_{\vx} V(\vx, \vy) \rangle.
    \end{equation}
    Furthermore, an analogous relation holds for player $y$.
\end{fact}

Establishing~\eqref{eq:gd} for the ratio game is immediate, but that property holds generally---under relatively mild assumptions on the underlying stochastic game. The importance of~\Cref{fact:gd} it that it guarantees that approximate stationary points of gradient-descent-type algorithms are also globally optimal for each player---even though the optimization problem faced by each player is nonconvex (nonconcave). In light of the fact that the MVI property (\Cref{assumption:MVI}) fails even for the ratio game~\citep{Daskalakis20:Independent}, when $F \defeq (\nabla_{\vx} V(\vx, \vy), - \nabla_{\vy} V(\vx, \vy))$, we introduce a more general condition below.

\begin{property}[Generalized MVI for Min-Max Problems]
    \label{property:genMVI}
    Consider the operator $F \defeq (\nabla_{\vx} V(\vx, \vy), - \nabla_{\vy} V(\vx, \vy))$. For any sequence $\vz^{(1)}, \dots, \vz^{(m)} \in \calZ$, there exists $(\xstar, \ystar) \in \X \times \Y$ and sequences of weights $\alpha_x^{(1)}, \dots, \alpha_x^{(m)} \in \R_{> 0}$ and $\alpha_y^{(1)}, \dots, \alpha_y^{(m)} \in \R_{> 0}$, with $0 < \lb \leq \alpha_x^{(i)}, \alpha_y^{(i)} \leq \ub$, for each $i \in \range{m}$, such that
    \begin{equation}
        \label{eq:genMVI}
        \sum_{i=1}^m \alpha_x^{(i)} \langle \nabla_{\vx} V(\vx^{(i)}, \vy^{(i)}), \vx^{(i)} - \xstar \rangle + \sum_{i=1}^m \alpha_y^{(i)} \langle \nabla_{\vy} V(\vx^{(i)}, \vy^{(i)}), \ystar - \vy^{(i)} \rangle \geq 0.
    \end{equation}
    Further, for any $2 \leq i \leq m$, there exists $W \in \R_{> 0}$ such that
    \begin{equation}
        \label{eq:smooth-alpha}
        \frac{\alpha_x^{(i)}}{\alpha_x^{(i-1)}}, \frac{\alpha_y^{(i)}}{\alpha_y^{(i-1)}} \leq 1 + W \|\vz^{(i)} - \vz^{(i-1)}\|_2.
    \end{equation}
\end{property}

It is evident that if the MVI property holds, \eqref{eq:genMVI} and \eqref{eq:smooth-alpha} also follow with all the coefficients being equal to $1$, but, unlike the MVI property, as we shall see~\Cref{property:genMVI} it satisfied for stochastic games. The sequence of weights in~\eqref{eq:genMVI} is designed to take into account the distribution shift, a central challenge in reinforcement learning. We now observe that~\Cref{property:genMVI} is satisfied for the ratio game.

\begin{proposition}
    \label{prop:genMVI}
    Let $F \defeq (\nabla_{\vx} V(\vx, \vy), - \nabla_{\vy} V(\vx, \vy))$, where $V$ is the ratio game defined in~\eqref{eq:ratio}. Then, \Cref{property:genMVI} is satisfied with $\ell = \zeta/\smax$ and $h = \smax/\zeta$, where $\smax \defeq \max_{(\vx, \vy) \in \X \times \Y} \vx^\top \mat{S} \vy$ and $\min_{(\vx, \vy) \in \X \times \Y} \vx^\top \mat{S} \vy \geq \zeta$.
\end{proposition}

Before we proceed with the proof of \Cref{prop:genMVI}, we remark that \Cref{property:genMVI}---and subsequently \Cref{prop:genMVI}---can be directly extended to multi-state stochastic games as well. 

\begin{proof}[Proof of \Cref{prop:genMVI}]
    Consider any sequence $\vz^{(1)}, \dots, \vz^{(m)} \in \calZ$, for some $m \in \N$. We first see that for any $(\vx, \vy) \in \X \times \Y$,
    \begin{equation*}
        \nabla_{\vx} V(\vx, \vy) = \frac{\mat{R} \vy (\vx^\top \mat{S} \vy) - \mat{S} \vy (\vx^\top \mat{R} \vy)}{(\vx^\top \mat{S} \vy)^2},
    \end{equation*}
    and similarly,
    \begin{equation*}
        \nabla_{\vy} V(\vx, \vy) = \frac{\mat{R}^\top \vx (\vx^\top \mat{S} \vy) - \mat{S}^\top \vx (\vx^\top \mat{R} \vy)}{(\vx^\top \mat{S} \vy)^2}.
    \end{equation*}
    As a result, for any $\xstar \in \X$,
    \begin{align}
        \sum_{i=1}^m \left( V(\vx^{(i)}, \vy^{(i)}) - V(\xstar, \vy^{(i)}) \right) &= \sum_{i=1}^m \left( \frac{ (\vx^{(i)})^{\top} \mat{R} \vy^{(i)}}{(\vx^{(i)})^{\top} \mat{S} \vy^{(i)}} - \frac{(\xstar)^{\top} \mat{R} \vy^{(i)}}{(\xstar)^{\top} \mat{S} \vy^{(i)}}  \right) \notag \\
        &= \sum_{i=1}^m \left( \frac{ (\vx^{(i)})^{\top} \mat{R} \vy^{(i)} (\xstar)^{\top} \mat{S} \vy^{(i)} - (\vx^{(i)})^{\top} \mat{S} \vy^{(i)} (\xstar)^{\top} \mat{R} \vy^{(i)}}{(\vx^{(i)})^{\top} \mat{S} \vy^{(i)} (\xstar)^{\top} \mat{S} \vy^{(i)} }  \right) \notag \\
        &= \sum_{i=1}^m \frac{ (\vx^{(i)})^{\top} \mat{S} \vy^{(i)} }{(\xstar)^{\top} \mat{S} \vy^{(i)}} \langle \vx^{(i)} - \xstar, \nabla_{\vx} V(\vx^{(i)}, \vy^{(i)}) \rangle \notag \\
        &= \sum_{i=1}^m \alpha_x^{(i)} \langle \vx^{(i)} - \xstar, \nabla_{\vx} V(\vx^{(i)}, \vy^{(i)}) \rangle, \label{align-v1}
    \end{align}
    where
    \begin{equation}
        \label{eq:ax}
        \alpha_x^{(i)}(\xstar) \defeq \frac{ (\vx^{(i)})^{\top} \mat{S} \vy^{(i)} }{(\xstar)^{\top} \mat{S} \vy^{(i)}}.
    \end{equation}
    Similarly, for any $\ystar \in \Y$,
    \begin{equation}
        \label{eq:v2}
        \sum_{i=1}^m \left( V(\vx^{(i)}, \ystar ) - V(\vx^{(i)}, \vy^{(i)}) \right) = \sum_{i=1}^m \alpha_y^{(i)} \langle \ystar - \vy^{(i)}, \nabla_{\vy} V(\vx^{(i)}, \vy^{(i)}) \rangle,
    \end{equation}
    where
    \begin{equation}
        \label{eq:ay}
        \alpha_y^{(i)}(\ystar) \defeq \frac{ (\vx^{(i)})^{\top} \mat{S} \vy^{(i)} }{(\vx^{(i)})^{\top} \mat{S} \ystar}.
    \end{equation}
    But, optimizing over $(\xstar, \ystar)$, the sum of the left-hand sides of \eqref{align-v1} and \eqref{eq:v2} can be lower bounded by
    \begin{align*}
        \max_{\ystar \in \Y} \left\{ \sum_{i=1}^m V(\vx^{(i)}, \ystar) \right\} - \min_{\xstar \in \X} &\left\{ \sum_{i=1}^m V(\xstar, \vy^{(i)}) \right\} \\ 
        &\geq m \max_{\ystar \in \Y} \min_{\xstar \in \X} V(\xstar, \ystar) - m \min_{\xstar \in \X} \max_{\ystar \in \Y} V(\xstar, \ystar) \geq 0,
    \end{align*}
    by \Cref{fact:zero-dual}, which in turn implies that \Cref{property:genMVI} is satisfied with coefficients as defined in \eqref{eq:ax} and \eqref{eq:ay}. Indeed, bounding those coefficients is immediate, yielding \eqref{eq:genMVI}, while \eqref{eq:smooth-alpha} also follows directly from \eqref{eq:ax} and \eqref{eq:ay}.
\end{proof}

Importantly, for a min-max problem that enjoys \Cref{property:genMVI}, one can employ $\OGD$, but with learning rate that adapts to the sequence of weights associated with~\Cref{property:genMVI}. In particular, we analyze the following variant of $\OGD$, defined for any iteration $i \in \N$ as
\begin{align*}
    \vx^{(i)} &\defeq \proj_{\X} \left( \hvx^{(i-1)} - \eta \alpha_x^{(i-1)} \nabla_{\vx} f(\vx^{(i-1)}, \vy^{(i-1)}) \right),\\
    \hvx^{(i)} &\defeq \proj_{\X} \left( \hvx^{(i-1)} - \eta \alpha_x^{(i)} \nabla_{\vx} f(\vx^{(i)}, \vy^{(i)}) \right),
\end{align*}
where $\alpha_x^{(i-1)}, \alpha_x^{(i)}$ are defined as in~\eqref{eq:ax}. A similar update rule holds for player $y$.

\begin{theorem}
    \label{theorem:stochastic}
    Consider the ratio game $V(\vx,\vy)$ defined in~\eqref{eq:ratio}. If both players employ $\OGD$ with a suitable learning rate schedule, $O(1/\epsilon^2)$ iterations suffices to reach a point $(\vx^{(i)}, \vy^{(i)}) \in \X \times \Y$ such that
    \begin{equation*}
        V(\vx^{(i)}, \vy^{(i)}) - \min_{\vx \in \X} V(\vx, \vy^{(i)}) \leq \epsilon \quad \textrm{and} \quad \max_{\vy \in \Y} V(\vx^{(i)}, \vy) - V(\vx^{(i)}, \vy^{(i)}) \leq \epsilon.
    \end{equation*}
\end{theorem}

\begin{proof}[Sketch of Proof]
    First, for any iteration $i \in \range{m}$,
    \begin{align*}
        \alpha_x^{(i-1)} \langle \vx^{(i)}, \nabla_{\vx}^{(i-1)} \rangle - \frac{1}{2\eta} \| \vx^{(i)} - \hvx^{(i-1)} \|_2^2 - \alpha_x^{(i-1)} \langle \hvx^{(i)}, \nabla_{\vx}^{(i-1)} \rangle + \frac{1}{2 \eta} \| \hvx^{(i)} - \hvx^{(i-1)}\|_2^2 \\ \geq \frac{1}{2\eta} \| \vec{x}^{(i)} - \hvx^{(i)} \|_2^2,
    \end{align*}
    where we used the shorthand notation $\nabla^{(i)}_{\vx} \defeq \nabla_{\vx} f(\vx^{(i)}, \vy^{(i)})$. In turn, this implies that
    \begin{align}
        \alpha_x^{(i)} \langle \vx^{(i)}, \nabla_{\vx}^{(i-1)} \rangle - \frac{\alpha_x^{(i)}}{2 \alpha_x^{(i-1)} \eta} \| \vx^{(i)} - \hvx^{(i-1)} \|_2^2 - \alpha_x^{(i)} \langle \hvx^{(i)}, \nabla_{\vx}^{(i)} \rangle& + \frac{\alpha_x^{(i)}}{2 \alpha_x^{(i-1)}\eta} \| \hvx^{(i)} - \hvx^{(i-1)}\|_2^2 \notag \\ \geq &\frac{\alpha_x^{(i)}}{2\alpha_x^{(i-1)}\eta} \| \vec{x}^{(i)} - \hvx^{(i)} \|_2^2. \label{align:mis1}
    \end{align}
    Further, for any iteration $i \in \range{m}$,
    \begin{align}
        \alpha_x^{(i)} \langle \hvx^{(i)}, \nabla_{\vx}^{(i)} \rangle - \frac{1}{2\eta} \| \hvx^{(i)} - \hvx^{(i-1)} \|_2^2 - \alpha_x^{(i)} \langle \xstar, \nabla_{\vx}^{(i)} \rangle + \frac{1}{2 \eta} \| \xstar - \hvx^{(i-1)}\|_2^2 \notag \\ \geq \frac{1}{2\eta} \| \xstar - \hvx^{(i)} \|_2^2.\label{align:mis2}
    \end{align}
    Moreover, we have
    \begin{equation}
        \label{eq:stab-proj1}
        \| \vx^{(i)} - \hvx^{(i)} \|_2 \leq \eta \alpha_x^{(i-1)} \|\nabla_{\vx}f(\vx^{(i-1)}, \vy^{(i-1)})\|_2,
    \end{equation}
    since the Euclidean projection operator $\proj(\cdot)$ is nonexpansive with respect to $\|\cdot\|_2$. Similarly,
    \begin{equation}
        \label{eq:stab-proj2}
        \| \hvx^{(i)} - \hvx^{(i-1)} \|_2 \leq \eta \alpha_x^{(i)} \| \nabla_{\vx}f(\vx^{(i)}, \vy^{(i)})\|_2.
    \end{equation}
    Thus, combining \eqref{eq:stab-proj1} and \eqref{eq:stab-proj2} with~\eqref{eq:smooth-alpha} of \Cref{property:genMVI} implies that for a sufficiently small learning rate $\eta > 0$, it holds that $\frac{\alpha_x^{(i)}}{\alpha_x^{(i-1)}} \leq 1 + \gamma$, for a sufficiently small universal constant $\gamma$. Now when combining \eqref{align:mis1} and \eqref{align:mis2} there is a mismatch term of the form
    \begin{equation*}
        \frac{\alpha_x^{(i)}}{2 \alpha_x^{(i-1)} \eta} \| \hvx^{(i)} - \hvx^{(i-1)}\|_2^2 - \frac{1}{2\eta} \| \hvx^{(i)} - \hvx^{(i-1)}\|_2^2 \leq \frac{\gamma}{2\eta} \| \hvx^{(i)} - \hvx^{(i-1)}\|_2^2.
    \end{equation*}
    Thus, when $\frac{\alpha_x^{(i)}}{\alpha_x^{(i-1)}} \geq 1$ this mismatch term will be subsumed by the other terms in~\eqref{align:mis1}. As a result, we can obtain an RVU bound for the weighted regret $\sum_{i=1}^m \alpha_x^{(i)} \langle \vx^{(i)} - \xstar, \nabla_{\vx}^{(i)} \rangle$, and similar reasoning applies for player $y$. As a result, analogously to the proof of \Cref{theorem:last-MVI}, we conclude that $O(1/\epsilon^2)$ suffice so that $\|\vx^{(i)} - \hvx^{(i)}\|_2, \|\vx^{(i)} - \hvx^{(i-1)}\|_2, \|\vy^{(i)} - \hvy^{(i)}\|_2,\|\vy^{(i)} - \hvy^{(i)}\|_2 \leq \epsilon$ (using \Cref{prop:genMVI}). Finally, this implies that $\max_{x \in \X} \langle \vx^{(i)} - \vx, \nabla_{\vx}^{(i)} \rangle = O(\epsilon)$ and $\max_{\vy \in \Y} \langle \vy - \vy^{(i)}, \nabla_{\vy}^{(i)} \rangle = O(\epsilon)$, and the gradient dominance property (\Cref{fact:gd}) yields the statement.
\end{proof}

Finally, it is immediate to derive the meta-learning version of \Cref{theorem:stochastic}, parameterized in terms of the similarity of the minimax equilibria of the underlying stochastic games.

\subsection{H\"older Smooth Games}
\label{appendix:Holder}

So far our analysis has (either implicitly or explicitly) required some form of Lipschitz continuity in order to appropriate massage the RVU bound. In this subsection, we will relax that condition. In particular, we study variational inequality problems for which the operator is H\"older continuous, in the following precise sense.

\begin{definition}
    \label{def:Holder}
    Consider an operator $F : \calZ \to \calZ$. We say that $F$ is $\alpha$-H\"older continuous, with $\alpha \in (0,1]$, if for any $\vz, \vz' \in \calZ$,
    \begin{equation}
        \label{eq:Holder}
        \|F(\vz) - F(\vz')\|_2 \leq H \|\vz - \vz'\|_2^{\alpha},
    \end{equation}
    for some parameter $H > 0$.
\end{definition}

Naturally, \eqref{eq:Holder} reduces to Lipschitz continuity of $F$ when $\alpha = 1$. This class of problems was also studied in the seminal work of~\citet{Rakhlin13:Optimization}. Indeed, the following argument partially overlaps with~\citep[Lemma 3]{Rakhlin13:Optimization}. For convenience, below we use a prediction based on the secondary sequences of $\OGD$.

\begin{proposition}[Refinement of~\Cref{cor:refined-paths} under H\"older Smoothness]
    \label{prop:Holder}
    Suppose that $F : \calZ \to \calZ$ is $\alpha$-H\"older smooth, with $\alpha \in (0,1)$, and satisfies the MVI property (\Cref{assumption:MVI}). Then, for $\OGD$ with prediction $\vec{m}^{(i)} \defeq F(\hvx^{(i-1)})$ and learning rate $\eta > 0$ set as 
    \begin{equation}
        \label{eq:bound-Holder}
        \eta(m) \defeq \left( \frac{\|\zstar - \vz^{(0)}\|_2^2}{m H^{\frac{2}{1 - \alpha}} g(\alpha)} \right)^{\frac{1 - \alpha}{2}},
    \end{equation}
    where $g(\alpha) = (1 + \alpha) (2 + 2 \alpha)^{\frac{1 - \alpha}{1 + \alpha}}$,
    it holds that
    \begin{equation*}
    \sum_{i=1}^m \|\vz^{(i)} - \hvz^{(i)}\|_2^2 + \sum_{i=1}^m \|\vz^{(i)} - \hvz^{(i-1)}\|_2^2 \leq 2 \| \vz^{\star} - \vz^{(0)} \|_2^2,
    \end{equation*}
    where $\vz^{\star} \in \calZ$ is any point that satisfies the MVI property (\Cref{assumption:MVI}).
\end{proposition}

In the argument below we do not make any attempt to optimize factors that depend (solely) on $\alpha$---the modulus of H\"older continuity (\Cref{def:Holder}).

\begin{proof}
    First, analogously to the proof of \Cref{cor:refined-paths}, we have that for any $\zstar \in \calZ^{\star}$,
    \begin{align*}
        0 \leq \frac{1}{2 \eta} \|\vec{z}^{\star} - \vz^{(0)}\|_2^2 + &H \sum_{i=1}^m 
        \langle \vz^{(i)} - \hat{\vz}^{(i)}, F(\vz^{(i)}) - F(\hat{\vz}^{(i-1)}) \rangle   \\ &- \frac{1}{2\eta} \sum_{i=1}^m \left( \|\vz^{(i)} - \hat{\vz}^{(i)}\|_2^2 + \|\vz^{(i)} - \hat{\vz}^{(i-1)}\|_2^2 \right).
    \end{align*}
    Further, by Cauchy-Schwarz inequality we have that $\langle \vz^{(i)} - \hat{\vz}^{(i)}, F(\vz^{(i)}) - F(\hat{\vz}^{(i-1)}) \rangle \leq \| \vz^{(i)} - \hat{\vz}^{(i)}\|_2 \| F(\vz^{(i)}) - F(\hat{\vz}^{(i-1)})\|_2$, in turn implying that
    \begin{align}
        0 \leq \frac{1}{2 \eta} \|\vec{z}^{\star} - \vz^{(0)}\|_2^2 + &H \sum_{i=1}^m \|\vz^{(i)} - \hat{\vz}^{(i)}\|_2^{1 + \alpha} + H \sum_{i=1}^m \|\vz^{(i)} - \hat{\vz}^{(i-1)}\|_2^{1 + \alpha} \notag \\&- \frac{1}{2\eta} \sum_{i=1}^m \left( \|\vz^{(i)} - \hat{\vz}^{(i)}\|_2^2 + \|\vz^{(i)} - \hat{\vz}^{(i-1)}\|_2^2 \right),\label{align:fb}
    \end{align}
    where we used the fact that $\| F(\vz^{(i)}) - F(\hat{\vz}^{(i-1)})\|_2 \leq H \|\vz^{(i)} - \hat{\vz}^{(i-1)}\|_2^{\alpha}$ (by $\alpha$-H\"older smoothness), and Young's inequality, $a b \leq \frac{a^p}{p} + \frac{b^q}{q}$ for any $a,b \in \R_{\geq 0}$ and $p,q \in \R_{> 0}$ such that $\frac{1}{p} + \frac{1}{q} = 1$, implying that $\| \vz^{(i)} - \hat{\vz}^{(i)}\|_2 \| \vz^{(i)} - \hat{\vz}^{(i-1)} \|^{\alpha}_2 \leq \frac{1}{1 + \alpha} \| \vz^{(i)} - \hat{\vz}^{(i)} \|^{1 + \alpha}_2 + \frac{\alpha}{1 + \alpha} \| \vz^{(i)} - \hat{\vz}^{(i-1)} \|^{1 + \alpha}_2 \leq \| \vz^{(i)} - \hat{\vz}^{(i)} \|^{1 + \alpha}_2 + \| \vz^{(i)} - \hat{\vz}^{(i-1)} \|^{1 + \alpha}_2 $ (for $p=  1 + \alpha$ and $q = \frac{1 + \alpha}{\alpha}$). Let us now bound the middle terms in the right-hand side of~\eqref{align:fb}. We have
    \begin{align}
        H \sum_{i=1}^m \|\vz^{(i)} - \hat{\vz}^{(i)}\|^{1 + \alpha} &= \sum_{i=1}^m H \left( (2 + 2 \alpha) \eta \right)^{\frac{1 + \alpha}{2}} \left( \frac{\|\vz^{(i)} - \hat{\vz}^{(i)}\|}{\sqrt{(2 + 2 \alpha)\eta}}\right)^{1 + \alpha} \notag \\
        &\leq \left( \sum_{i=1}^m H^{\frac{2}{1 - \alpha}} ((2 + 2\alpha)\eta)^{\frac{1 + \alpha}{1 - \alpha}} \right)^{\frac{1 - \alpha}{2}} \left( \sum_{i=1}^m \frac{\|\vz^{(i)} - \hat{\vz}^{(i)} \|_2^2}{(2 + 2\alpha)\eta} \right)^{\frac{1 + \alpha}{2}} \label{align:Holder} \\
        &\leq \frac{1 - \alpha}{2} m H^{\frac{2}{1 - \alpha}} ((2 + 2\alpha)\eta)^{\frac{1 + \alpha}{1 - \alpha}} + \frac{1}{4\eta} \sum_{i=1}^m \|\vz^{(i)} - \hat{\vz}^{(i)}\|_2^2, \label{align:AM-GM}
    \end{align}
    where \eqref{align:Holder} uses H\"older's inequality with conjugate powers $(1 - \alpha)/2$ and $(1 + \alpha)/2$, and \eqref{align:AM-GM} uses the (weighted) AM-GM inequality. Similarly, 
    \begin{equation*}
        H \sum_{i=1}^m \|\vz^{(i)} - \hat{\vz}^{(i-1)}\|^{1 + \alpha} \leq \frac{1 - \alpha}{2} m H^{\frac{2}{1 - \alpha}} ((2 + 2\alpha)\eta)^{\frac{1 + \alpha}{1 - \alpha}} + \frac{1}{4\eta} \sum_{i=1}^m \|\vz^{(i)} - \hat{\vz}^{(i-1)}\|_2^2.
    \end{equation*}
    We now optimize the following function
    \begin{equation}
        \label{eq:opt-fun}
        \frac{1}{2\eta} \|\zstar - \vz^{(0)}\|_2^2 + (1 - \alpha) m H^{\frac{2}{1 - \alpha}} (2 + 2\alpha)^{\frac{1 + \alpha}{1 - \alpha}} \eta^{\frac{1 + \alpha}{1 - \alpha}}
    \end{equation}
    in terms of the learning rate $\eta > 0$. In particular, if we let $g(\alpha) \defeq (1 + \alpha) (2 + 2 \alpha)^{\frac{1 - \alpha}{1 + \alpha}}$, a simple calculation yields that the optimal value for $\eta$ is
    \begin{equation*}
        \eta(m) \defeq \left( \frac{\|\zstar - \vz^{(0)}\|_2^2}{2 m H^{\frac{2}{1 - \alpha}} g(\alpha)} \right)^{\frac{1 - \alpha}{2}},
    \end{equation*}
    which gives a value to~\eqref{eq:opt-fun} equal to
    \begin{equation*}
        \left( \frac{H \| \zstar - \vz^{(0)} \|^{1 + \alpha}}{2^{\frac{1 + \alpha}{2}} (g(\alpha))^{\frac{\alpha - 1}{2}}} \right) m^{\frac{1 - \alpha}{2}}.
    \end{equation*}
    Plugging this bound to~\eqref{align:fb} and rearranging concludes the proof.
\end{proof}

Above, we have assumed access to a point $\zstar \in \calZ$ that satisfies the MVI property in order to optimally tune the learning rate. That assumption can be sidestepped similarly to our approach in~\Cref{appendix:cce}. While \Cref{prop:Holder} still guarantees that the second-order path lengths are bounded (\eqref{eq:bound-Holder}), the difference with our previous bounds under Lipschitz continuity is that the learning rate $\eta$ has to decay with the time horizon $m$. This is reflected on a worse bound on the number of iterations required to reach an approximate solution to the VI problem, as we formalize below. To our knowledge, the following guarantee is the first of its kind.

\begin{theorem}[Rates for H\"older Smooth Functions]
    \label{theorem:rate-Holder}
    Consider an $\alpha$-H\"older continuous operator $F : \calZ \to \calZ$ that satisfies the MVI property (\Cref{assumption:MVI}). Then, after $O(m)$ iterations of $\OGD$ with the learning rate of \Cref{prop:Holder} there is an iterate that is an $m^{-\alpha/2}$-approximate strong solution to the VI problem.
\end{theorem}

\begin{proof}
    By~\Cref{prop:Holder}, after $m$ iterations there is an iterate $i \in \range{m}$ such that $\|\vz^{(i)} - \hvz^{(i)}\|_2, \|\vz^{(i)} - \hvz^{(i-1)}\|_2 \leq \frac{2\|\zstar - \vz^{(0)}\|_2}{\sqrt{m}}$. Thus, the statement follows from~\Cref{claim:approx-stat}.
\end{proof}

In the special case where $\alpha = 1$, the guarantee above recovers the recently established $m^{-1/2}$ rates of $\OGD$ under Lipschitz continuity, which is known to be tight~\citep{Golowich20:Tight,Golowich20:Last}. An interesting question is to derive lower bounds under the weaker H\"older continuity assumption studied in the present subsection. Finally, obtaining a meta-learning version of \Cref{theorem:rate-Holder} follows immediately from our previous techniques. We caution that one has to also meta-learn
the learning rate in this case as \Cref{prop:Holder} assumes knowledge of $\zstar \in \calZ$; this is analogous to our approach in~\Cref{theorem:cce}, presented in \Cref{appendix:cce}.

\subsection{The Extra-Gradient Method}
\label{appendix:EG}

In this subsection, we extend our results to the \emph{extra-gradient} method. Starting from the seminal work of~\citet{Korpelevich76:Extragradient}, that method---and variants thereof---has received tremendous attention in the literature; we refer to the excellent discussion of~\citet{Hsieh19:On} for further pointers. We will consider a generalized version of the extra-gradient method, defined with the following update rule for $i \in \N$.
\begin{equation}
    \label{eq:EG}
    \begin{split}
        \hvx^{(i)} &\defeq \argmax_{\hvx \in \X} \left\{ \langle \hvx, \vec{u}^{(i-1)} \rangle - \frac{1}{\eta} \brg{}{\hvx}{\vx^{(i-1)}} \right\}, \\
        \vec{x}^{(i)} &\defeq \argmax_{\vec{x} \in \X} \left\{ \langle \vec{x}, \hvu^{(i)} \rangle - \frac{1}{\eta} \brg{}{\vx}{\vx^{(i-1)}} \right\}.
    \end{split}
\end{equation}

The initialization is defined as in $\OGD$. For brevity, we will refer to this algorithm as $\EG$. We remark that the update rule presented in~\eqref{eq:EG} is more general than the standard extra-gradient method, which corresponds to~\eqref{eq:EG} with Euclidean regularization. While similar to $\OGD$, $\EG$ has one central difference: the update of $\vx^{(i)}$ uses an auxiliary feedback $\hvu^{(i)}$. For this reason, unlike $\OGD$, $\EG$ does not fit into the traditional online learning framework. In fact, any non-trivial modification of $\EG$ fails to provide meaningful regret guarantees when faced against an adversarially selected sequence of utilities~\citep{Golowich20:Tight}. That additional feedback $\hvu^{(i)}$ is also unsatisfactory since it requires two gradient-oracle calls per iteration---unlike single-call variants, such as $\OGD$~\citep{Hsieh19:On}. 

In this context, the main purpose of this subsection is to show that $\EG$ can be analyzed in the same framework as $\OGD$, which will ensure that the guarantees we have provided so far for $\OGD$ can be translated to $\EG$ as well. In particular, although $\EG$ does not lie in the no-regret framework, we will show that it still admits an RVU-type bound; this leads to a unifying analysis of $\OGD$ and $\EG$, and recovers several well-known results from the literature in a much simpler fashion. The first key idea is to consider the regret incurred by the secondary sequence of $\EG$ $(\hvx^{(i)})_{1 \leq i \leq m}$, and with respect to the auxiliary sequence of utilities $(\hvu^{(i)})_{1 \leq i \leq m}$:

\begin{equation}
    \label{eq:hreg}
    \hreg^{(m)}(\rvx) \defeq \sum_{i=1}^m \langle \rvx - \hvx^{(i)}, \hvu^{(i)} \rangle.
\end{equation}

Below, we show that $\EG$ also admits an RVU-type bound, but with respect to the notion of regret introduced in~\eqref{eq:hreg}.

\begin{theorem}[RVU-type Bound for Extra-Gradient]
    \label{theorem:rvu-EG}
    For the extra-gradient method \eqref{eq:EG} it holds that $\sum_{i=1}^m \langle \rvx - \hvx^{(i)}, \hvu^{(i)} \rangle$ can be upper bounded, for any $\rvx \in \X$, as
    \begin{equation*}
        \frac{1}{\eta} \brg{}{\rvx}{\vx^{(0)}} + \eta \sum_{i=1}^m \| \hvu^{(i)} - \vec{u}^{(i-1)} \|^2_* - \frac{1}{2 \eta} \sum_{i=1}^m \left( \|\hvx^{(i)} - \vec{x}^{(i)} \|^2 + \| \hvx^{(i)} - \vec{x}^{(i-1)} \|^2 \right).
    \end{equation*}
\end{theorem}

\begin{proof}
Fix any $i \in \range{m}$. By $1$-strong convexity of $\cR$ with respect to $\|\cdot\|$,
\begin{equation}
    \label{eq:EG-1}
    \langle \hvx^{(i)}, \vec{u}^{(i-1)} \rangle - \frac{1}{\eta} \brg{}{\hvx^{(i)}}{\vx^{(i-1)}} - \langle \vec{x}^{(i)}, \vec{u}^{(i-1)} \rangle + \frac{1}{\eta} \brg{}{\vx^{(i)}}{\vx^{(i-1)}} \geq \frac{1}{2\eta} \| \hvx^{(i)} - \vec{x}^{(i)} \|^2,
\end{equation}
and for any $\rvx \in \X$,
\begin{equation}
    \label{eq:EG-2}
    \langle \vec{x}^{(i)}, \hvu^{(i)} \rangle - \frac{1}{\eta} \brg{}{\vx^{(i)}}{\vx^{(i-1)}} - \langle \rvx, \hvu^{(i)} \rangle + \frac{1}{\eta} \brg{}{\rvx}{\vx^{(i-1)}} \geq \frac{1}{\eta} \brg{}{\rvx}{\vx^{(i)}}.
\end{equation}
Moreover, for any $i \in \range{m}$,
\begin{equation*}
    \langle \rvx - \hvx^{(i)}, \hvu^{(i)} \rangle = \langle \rvx, \hvu^{(i)} \rangle - \langle \vec{x}^{(i)}, \hvu^{(i)} \rangle + \langle \vec{x}^{(i)}, \vec{u}^{(i-1)} \rangle - \langle \hvx^{(i)}, \vec{u}^{(i-1)} \rangle + \langle \vec{x}^{(i)} - \hvx^{(i)}, \hvu^{(i)} - \vec{u}^{(i-1)} \rangle.
\end{equation*}
Combining with \eqref{eq:EG-1} and \eqref{eq:EG-2} yields that $\sum_{i=1}^m \langle \rvx - \hvx^{(i)}, \hvu^{(i)} \rangle$ is upper bounded by
\begin{equation}
    \label{eq:almost}
     \frac{1}{\eta} \brg{}{\rvx}{\vx^{(0)}} + \sum_{i=1}^m \| \vec{x}^{(i)} - \hvx^{(i)} \| \| \hvu^{(i)} - \vec{u}^{(i-1)} \|_* - \frac{1}{2 \eta} \sum_{i=1}^m \left( \|\hvx^{(i)} - \vec{x}^{(i)} \|^2 + \| \hvx^{(i)} - \vec{x}^{(i-1)} \|^2 \right),
\end{equation}
where we further used Cauchy-Schwarz inequality for the dual pair of norms $(\|\cdot\|, \|\cdot\|_*)$ to obtain that $\langle \vec{x}^{(i)} - \hvx^{(i)}, \hvu^{(i)} - \vec{u}^{(i-1)} \rangle \leq \| \vec{x}^{(i)} - \hvx^{(i)} \| \| \hvu^{(i)} - \vec{u}^{(i-1)} \|_*$, as well as the telescopic summation
\begin{equation*}
    \frac{1}{\eta} \sum_{i=1}^m \left( \brg{}{\rvx}{\vx^{(i-1)}} - \brg{}{\rvx}{\vx^{(i)}} \right) = \frac{1}{\eta} \brg{}{\rvx}{\vx^{(0)}} - \frac{1}{\eta} \brg{}{\rvx}{\vx^{(m)}} \leq \frac{1}{\eta} \brg{}{\rvx}{\vx^{(0)}},
\end{equation*}
since $\brg{}{\cdot}{\cdot} \geq 0$. Finally, we will need the following stability bound.

\begin{claim}
    \label{claim:EG-stab}
    For any $i \in \range{m}$,
    \begin{equation*}
        \| \vec{x}^{(i)} - \hvx^{(i)} \| \leq \eta \| \hvu^{(i)} - \vec{u}^{(i-1)} \|_*.
    \end{equation*}
\end{claim}

\begin{proof}[Proof of \Cref{claim:EG-stab}]
    By replacing $\rvx \defeq \hvx^{(i)}$ in~\eqref{eq:EG-2} and summing with~\eqref{eq:EG-1},
\begin{align*}
    \langle \hvx^{(i)}, \vec{u}^{(i-1)} \rangle - \langle \vec{x}^{(i)}, \vec{u}^{(i-1)} \rangle + \langle \vec{x}^{(i)}, \hvu^{(i)} \rangle - \langle \hvx^{(i)}, \hvu^{(i)} \rangle \geq \frac{1}{\eta} \| \hvx^{(i)} - \vec{x}^{(i)} \|^2,
\end{align*}
in turn implying that
\begin{equation*}
    \langle \vec{x}^{(i)} - \hvx^{(i)}, \hvu^{(i)} - \vec{u}^{(i-1)} \rangle \geq \frac{1}{\eta} \|\hvx^{(i)} - \vec{x}^{(i)} \|^2 \implies
    \|\hvx^{(i)} - \vec{x}^{(i)} \| \leq \eta \|\hvu^{(i)} - \vec{u}^{(i-1)} \|_*,
\end{equation*}
where the last bound follows from the Cauchy-Schwarz inequality.
\end{proof}
Finally, the statement follows by combining \Cref{claim:EG-stab} with~\eqref{eq:almost}.
\end{proof}

This regret bound allows us to automatically inherit the many and important consequences of the RVU bound for $\EG$ as well. For example, assuming Lipschitz continuity for the underlying operator $F$, \Cref{theorem:rvu-EG} guarantees that the sum of the regrets of the secondary sequences---in the sense of~\eqref{eq:hreg}---will be bounded by $O(1)$; assuming that $F$ is also monotone, that implies that the average secondary sequences yields an $O(1/m)$-approximate (weak) solution to the VI problem. Further, assuming the MVI property, it holds that the sum of the regrets of the secondary sequences is always nonnegative, thereby implying bounded second-order path lengths for $\EG$ (as in \Cref{cor:refined-paths}); in turn, that gives that $O(1/\epsilon^2)$ iterations suffice to reach an $\epsilon$-approximate (strong) solution to the VI problem. The implications we have derived earlier for $\OGD$ in the meta-learning setting also follow directly. 

\subsection{Lower Bounds}
\label{appendix:lowerbounds}

Finally, this subsection derives lower bounds on the players' regrets in the meta-learning setting. We begin by first recalling the lower bound for a single zero-sum game (in normal-form) due to~\citet{Daskalakis15:Near} (\Cref{prop:lower-bound}). We then extend their idea to the meta-learning setting.  

The basic idea is to consider the family of games for which the only nonzero entries of the $d \times d$ payoff matrix are at a single row. In particular, for an index $r \in \range{d}$, we let 
\begin{equation}
    \label{eq:class-lb}
    \mat{A}_r[a_x, a_y] = 
    \begin{cases}
        1 \quad \textrm{if   } a_x = r, \\
        0 \quad \textrm{otherwise}.
    \end{cases}
\end{equation}
We then consider the family of games described by $\{\mat{A}_r\}_{1 \leq r \leq d}$. We suppose for convenience that the entries of each matrix correspond to the utility of the row player. Now in any zero-sum game from that class, the column player has no impact on the game since each payoff matrix~\eqref{eq:class-lb} remains invariant under changes in the columns. Further, the row player is clearly just searching for the index $r$, which indicates the row with entries all-ones. Since the row player has no knowledge about the payoff matrix in the beginning of the game, there is a high probability that the row player will incur constant regret in the first iteration of the game. Conditioned on that event, this means that $\reg_x^{(m)} + \reg_y^{(m)} = \Omega(1)$, for any $m \in \N$, given that $\reg_y^{(m)} = 0$.

\begin{proposition}[\citep{Daskalakis15:Near}]
    \label{prop:lower-bound}
    There exists a class of zero-sum games for which $\E[\reg_x^{(m)} + \reg_y^{(m)}] = \Omega(1)$, where the expectation is with respect to the agents' randomization.
\end{proposition}

\paragraph{The construction is the meta-learning setting} Now let us imagine that one repeats the previous construction for a sequence of zero-sum games derived from the class in~\eqref{eq:class-lb}. If every such zero-sum game is selected uniformly at random, then \Cref{prop:lower-bound} still applies since there is no additional structure that can be exploited by the meta-learner. In this context, the key idea to refine that lower bound in the meta-learning setting is to assume that there is some prior distribution $\vec{p} \in \Delta^{d}$ over indexes, which is in fact assumed to be known by the agents. Then, the sequence of games is produced by selecting a zero-sum game $\mat{A}^{(t)}$ from the family $\{\mat{A}_r\}_{1 \leq r \leq d}$ according to distribution $\vec{p}$; those random draws in the course of the $t$ games are independent of each other. The following observation is a direct extension of~\Cref{prop:lower-bound}.

\begin{proposition}
    \label{prop:exte-Dask}
    Consider a sequence of zero-sum games $(\mat{A}^{(t)})_{1 \leq t \leq T}$ such that each matrix $\mat{A}^{(t)}$ is produced by selecting from the set $\{\mat{A}_r\}_{1 \leq r \leq d}$~\eqref{eq:class-lb} with probability according to $\vec{p}$. Then, $\frac{1}{T} \sum_{t=1}^T \E[\reg_x^{(t,m)} + \reg_y^{(t,m)}] \geq 1 - \max_{1 \leq r \leq d} \vec{p}[r]$.
\end{proposition}

Now, for any $\epsilon > 0$ and $\delta \in (0,1)$, if $T = \poly(d, 1/\epsilon,\log(1/\delta))$ is sufficiently large, standard arguments imply that the empirical frequency $\hat{\vec{p}} \in \Delta^d$ is $\epsilon$-close to $\vec{p}$ in terms of total variation distance from $\vec{p}$, meaning that $\frac{1}{2} \sum_{r=1}^d | \vec{p}[r] - \hat{\vec{p}}[r]| \leq \epsilon$, with probability at least $1 - \delta$. Moreover, from the perspective of player $x$, both the Nash equilibrium and the optimal in hindsight strategy is obvious: for a game $\mat{A}_r$, for some $r \in \range{d}$, select (with probability 1) the action indexed by $r$; the perspective of player $y$ is of no consequence since player $y$ does not incur any regret. Further, let $r^{\star} \in \range{d}$ be amongst the most frequent indeces of $\vec{p}$, that is $r^{\star} \in \argmax_{1 \leq r \leq d} \vec{p}[r]$, and let $\vec{\pi}_{r^{\star}} \in \Delta^d$ be such that $\vec{\pi}_{r^{\star}}[r^{\star}] = 1$. Thus, conditioned on the event that $\hat{\vec{p}}$ is $\epsilon$-close to $\vec{p}$ in terms of total variation distance,
\begin{align*}
V_x^2 \defeq \frac{1}{T} \min_{\vx \in \X} \sum_{t=1}^T \|\rvx^{(t)} - \vx\|_2^2 &\leq \frac{1}{T} \sum_{t=1}^T \|\vec{\pi}_{r^{(t)}} - \vec{\pi}_{r^{\star}} \|_2^2 \leq \frac{1}{T} \sum_{t=1}^T \max_{r \neq r'} \|\vec{\pi}_{r} - \vec{\pi}_{r'}\|_2^2 \mathbbm{1}\{ r^{\star} \neq r^{(t)}\} \\
&= 2 \sum_{r \neq r^{\star}} \hat{\vec
p}[r] \leq 2 (1 - \max_{1 \leq r \leq d} \vec{p}[r]) + 4 \epsilon,
\end{align*}
where $r^{(t)}$ is the index of the unique nonzero row of the payoff matrix $\mat{A}^{(t)}$. Combining this with \Cref{prop:exte-Dask}, we are now ready to establish our lower bound. Below, we point out that $\hinsim_y^2$ is $0$ since any possible strategy for player $y$ yields the same utility (recall our tie-breaking convention for $\rvy^{(t)}$ made after \Cref{def:regret}), while we make a similar convention for $\NEsim^2$ to avoid trivialities.

\begin{theorem}[Precise Verison of~\Cref{theorem:informal-lb}]
    \label{theorem:lower-bound}
    Consider a sequence of zero-sum games $(\mat{A}^{(t)})_{1 \leq t \leq T}$ such that each matrix $\mat{A}^{(t)}$ is produced by selecting from the set $\{\mat{A}_r\}_{1 \leq r \leq d}$~\eqref{eq:class-lb} with probability according to $\vec{p}$. Then, for any $\epsilon > 0$ and a sufficiently large $T = \poly(d, 1/\epsilon)$,
    \begin{equation*}
        \frac{1}{T} \sum_{t=1}^T \E[\reg_x^{(t,m)} + \reg_y^{(t,m)}] \geq \frac{1}{2} \left( \hinsim_x^2 + \hinsim_y^2 \right) -  \epsilon = \frac{1}{2} \NEsim^2 -  \epsilon.
    \end{equation*}
\end{theorem}

\section{Proofs from \texorpdfstring{\Cref{sec:general-sum}}{Section 3.2}: Meta-Learning in General-Sum Games}
\label{appendix:general-sum}

In this section, we switch our attention to general-sum games. Unlike zero-sum games, where computing a Nash equilibrium can be phrased as a linear program, there are inherent computational barriers for finding Nash equilibria in general-sum games~\citep{Chen09:Settling,Daskalakis09:The}. Instead, learning algorithms are known to converge to different notions of \emph{correlated equilibria}~\citep{Hart00:Adaptive,blum2007external}, which are more permissive than the Nash equilibrium~\citep{aumann1973subjecjwity}.

In our meta-learning context, it is natural to ask whether one can obtain refined regret bounds that are parameterized based on the similarity of the correlated equilibria of the games. However, that appears hard to achieve using uncoupled learning algorithms. Indeed, suppose that after a sufficient number of iterations the agents have converged---in terms of the average product distribution of play---to a, say, coarse correlated equilibrium (CCE) of the underlying game. Then, under uncoupled methods, any initialization at the next game will inevitably produce a product distribution, which could be far from the previous CCE. In other words, we may fail to appropriately leverage the learning of the previous task since the initialization alone is not enough to encode the previous CCE. The separation we are alluding to is summarized below.

\begin{proposition}[Separation between Zero-Sum and General-Sum]
    \label{prop:separation}
    Consider the two-player games $\game$ and $\game'$ with maximum pairwise difference~\citep{Candogan13:Dynamics} $d(\game, \game') \leq \epsilon$, for a sufficiently small $\epsilon \leq 1/\max\{\poly(\game'), \poly(\game)\}$.
    \begin{itemize}
        \item If $\game, \game'$ are both zero-sum, then no-regret learning in $\game$ allows the players to have $O(\epsilon)$ regret from the first iteration in $\game'$.
        \item On the other hand, if the games are general-sum, then even if the players have $0$ regret in $\game$, finding an initialization with $O(\epsilon)$ regret from the first iteration in $\game'$ is $\PPAD$-hard, even if the games are known.
    \end{itemize}
\end{proposition}

One way to bypass the aforedescribed difficulties is by incorporating some form of centralization into the learning process. Indeed, computing a correlated equilibrium can be phrased as a zero-sum game between a ``mediator,'' which will serve as the coordinating party, and the players (see, \emph{e.g.}, \citep{Hart89:Existence}). As such, one can leverage our previous results for zero-sum games even in general games, where it is more likely to obtain bounds that depend on the similarity of the correlated equilibria. Nonetheless, investigating this further is not in our scope since it deviates substantially from the online learning paradigm we follow in this paper.

Instead, in \Cref{appendix:cce,appendix:CE} we primarily focus on obtaining refined convergence bounds to correlated and coarse correlated equilibria that depend on the task similarity of the optimal in hindsight. But first, we remark that learning algorithms are known to lead to Nash equilibria in some ``structued'' general-sum games. Perhaps the most prominent example is that of potential games~\citep{Kleinberg09:Multiplicative,Hofbauer02:On,Candogan13:Dynamics}, which is pursued in \Cref{appendix:pot} below. Two other interesting directions for which obtaining meta-learning guarantees is left for future work are supermodular games~\citep{Milgrom90:Rationalizability} and games possessing \emph{strict Nash equilibria}~\citep{Giannou21:Survival}.

\subsection{Potential Games}
\label{appendix:pot}

Here we study meta-learning on potential games given in strategic form (the strategic-form representation was introduced in the beginning of \Cref{appendix:impl-sw}); we suspect that similar results apply more generally to Markov potential games~\citep{Leonardos22:Global}. We first recall the definition of a potential game we will use throughout this subsection.

\begin{definition}[Potential Game]
    \label{def:pot}
    A strategic-form game $\game$ is potential if there exists a function $\Phi : \bigtimes_{k=1}^n \X_k \to \R$ such that for any player $i \in \range{n}$, joint strategy $\vx \in \bigtimes_{k=1}^n \X_k$ and action $a_k \in \calA_k$,
    \begin{equation*}
        \frac{\partial \Phi (\vx)}{\partial \vx_k[a_k]} = u_k(a_k, \vx_{-k}).
    \end{equation*}
\end{definition}

We let $\phimax$ be an upper bound on the function $|\Phi(\vx)|$. In potential games, unlike zero-sum games, variants of mirror descent, such as gradient descent ($\GD$), are known to reach Nash equilibria. In this context, we assume that players face a sequence of potential games $(\Phi^{(t)})_{1 \leq t \leq T}$, described by their potential functions, and the goal will be to obtain parameterized regret bounds that depend on the similarity of the potential functions. In particular, we will use the following result~\citep[Corollary 4.4]{anagnostides2022last}, which gives an initialization-dependent bound on the second-order path length of $\GD$.
\begin{proposition}[\citep{anagnostides2022last}]
    \label{prop:init-pot}
    Suppose that all players employ $\GD$ in a potential game with a sufficiently small learning rate $\eta > 0$. Then, 
    \begin{equation}
        \label{eq:path-pot}
        \frac{1}{2 \eta} \sum_{i=1}^{m} \sum_{k=1}^n \|\vx_k^{(t,i)} - \vx_k^{(t,i-1)} \|_2^2 \leq \Phi^{(t)}(\vx^{(t,m)}) - \Phi^{(t)}(\vx^{(t,0)}),
    \end{equation}
    where $\Phi^{(t)} : \bigtimes_{k=1}^n \X_k \to \R$ is the potential function of the $t$-th game, and $\vx^{(t,i)} = (\vx_1^{(t,i)}, \dots, \vx_n^{(t,i)})$ is the players' joint strategy at iteration $i \in \range{m}$.
\end{proposition}

Now the main challenge here is that the potential function is in general nonconcave (and nonconvex). Thus, employing a meta-regret minimizer for learning the initilization appears to be computationally prohibitive. While there are interesting settings in which the potential function is concave (see \Cref{sec:concave-pot}), such as linear Fisher markets~\citep{Birnbaum11:Distributed}, we will bypass the need for concavity in the general setting by using a rather simpler initialization. In particular, let us denote by $\diff(\Phi, \Phi') \defeq \max_{\vx} ( \Phi(\vx) - \Phi'(\vx))$ the difference functional for two functions $\Phi, \Phi' : \bigtimes_{k=1}^n \X_k \to \R$. Then, we will use the following notion of similarity for the sequence of encountered potential games.

\begin{equation}
    \label{eq:diff-sim}
    \diffsim \defeq \frac{1}{T} \sum_{t=1}^{T-1} \diff(\Phi^{(t)}, \Phi^{(t+1)}).
\end{equation}

We are now ready to derive a refined bound for the second-order path lengths of $\GD$ in terms of~\eqref{eq:diff-sim}. Similar results apply more broadly for other variants of mirror descent.

\begin{corollary}
    \label{cor:potbound}
    Suppose that all players employ $\GD$ in a potential game with a sufficiently small learning rate $\eta > 0$ and initialization $\vx_k^{(t, 0)} \defeq \vx_k^{(t-1, m)}$, for all $k \in \range{n}$ and $t \geq 2$. Then, 
    \begin{equation*}
        \frac{1}{2\eta T} \sum_{t=1}^T \sum_{i=1}^{m} \|\vx^{(t,i)} - \vx^{(t,i-1)} \|_2^2 \leq \frac{2\phimax}{T} + \diffsim,
    \end{equation*}
    where $\diffsim$ is defined as in~\eqref{eq:diff-sim}, and $\vx^{(t,i)} = (\vx_1^{(t,i)}, \dots, \vx_n^{(t,i)})$ is the joint strategy at task $t \in \range{t}$ and iteration $i \in \range{m}$.
\end{corollary}

\begin{proof}
    By applying \Cref{prop:init-pot} for each task $t \in \range{T}$, we have    
    \begin{align}
    \frac{1}{2\eta} \sum_{t=1}^T \sum_{i=1}^{m} \|\vx^{(t,i)} - \vx^{(t,i-1)} \|_2^2 &\leq \sum_{t=1}^{T} \left( \Phi^{(t)}(\vx^{(t,m)}) - \Phi^{(t)}(\vx^{(t,0)}) \right) \notag \\
    &\leq 2\phimax + \sum_{t=1}^{T-1} \left( \Phi^{(t)}(\vx^{(t,m)}) - \Phi^{(t+1)}(\vx^{(t+1,0)}) \right) \notag \\
    &\leq 2\phimax + \sum_{t=1}^{T-1} \left( \Phi^{(t)}(\vx^{(t,m)}) - \Phi^{(t+1)}(\vx^{(t,m)}) \right) \label{align:init} \\
    &\leq 2 \phimax + \sum_{t=1}^{T-1} \Delta(\Phi^{(t)}, \Phi^{(t+1)}), \label{align:delta}
    \end{align}
    where \eqref{align:init} uses the initialization $\vx^{(t+1, 0)} \defeq \vx^{(t, m)}$ and the definition of $\phimax$, and \eqref{align:delta} follows from the definition of $\diff(\cdot, \cdot)$. The statement then follows by averaging~\eqref{align:delta} over the $T$ tasks and using~\eqref{eq:diff-sim}.
\end{proof}

Thus, $\GD$ under the initialization of~\Cref{cor:potbound} enjoys the meta-learning guarantee provided below. There are also consequences for the players' regrets, which we omit since they follow directly.

\begin{corollary}[Detailed Version of \Cref{theorem:informal-pot}]
    \label{cor:pot-last}
    Under the assumptions of \Cref{cor:potbound}, for an average potential game 
    \begin{equation*}
        m \defeq \left\lceil \frac{4\eta \phimax }{\epsilon^2 T} + \frac{2 \eta \diffsim}{\epsilon^2} \right\rceil
    \end{equation*}
    iterations suffice to reach an 
    \begin{equation*}
        \epsilon \left( \frac{\max_{k \in \range{n}} \diam_k}{\eta} + \sqrt{\max_{k \in \range{n}} \calA_k} \right)-\textrm{approximate Nash equilibrium}.
    \end{equation*}
\end{corollary}

\begin{proof}
    By \Cref{cor:potbound}, for an average task $t \in \range{T}$ it follows that $\left\lceil \frac{4\eta \phimax }{\epsilon^2 T} + \frac{2 \eta \diffsim}{\epsilon^2} \right\rceil$ iterations suffice so that $\|\vx_k^{(t,i)} - \vx_k^{(t,i-1)}\|_2 \leq \epsilon$, for some $i \in \range{m}$, for all $k \in \range{n}$. Then, the statement follows by~\citep[Claim B.7]{anagnostides2022last}.
\end{proof}

\subsubsection{Refinements under Concave Potentials}
\label{sec:concave-pot}

Moreover, as we pointed out earlier, there are important settings for which the potential function is concave, such as linear Fisher markets~\citep{Birnbaum11:Distributed}; such settings are not necessarily expressed in strategic form, but are indeed readily amenable to our techniques. Then, one can use a meta-algorithm for learning the initialization that receives after the termination of each task $t \in \range{T}$ the cost function
\begin{equation*}
    (\vx_1, \dots, \vx_n) \eqqcolon \vx \mapsto \Phi^{(t)}(\vx^{(t,m)}) - \Phi^{(t)}(\vx).
\end{equation*}
In particular, here it is assumed that $\Phi^{(t)}$ is also available after the end of the task---a rather stringent assumption compared to what is required in \Cref{cor:pot-last}. By concavity, the meta-learner will incur $o(T)$ regret, thereby leading to an $o_T(1)$ overhead in the per-task performance, which becomes negligible as $T \to +\infty$. Concavity also guarantees that the meta-learner can be implemented efficiently. Moreover, by virtue of \Cref{def:pot}, gradient-descent-type algorithms on the potential function can be in fact implemented in a full decentralized way by having each player perform a local update based on the observed utility. As a result, this leads to a task similarity of the form
\begin{equation*}
    \min_{\xstar} \sum_{t=1}^T \left( \Phi^{(t)}(\vx^{(t, m)}) - \Phi^{(t)}(\xstar) \right),
\end{equation*}   
where $\xstar \in \bigtimes_{k=1}^n \X_k$.

\subsection{Coarse Correlated Equilibria}
\label{appendix:cce}

In this subsection, we study meta-learning for coarse correlated equilibria. We will use optimistic Hedge ($\opthedge$) as our base learner. In particular, we will use the following guarantee~\citep[Theorem 3.1]{Daskalakis21:Near}.

\begin{theorem}[\citep{Daskalakis21:Near}]
    \label{theorem:Dask}
    There exist universal constants $C, C' > 0$ so that when all players employ $\opthedge$ with learning rate $\eta < \frac{1}{C n \log^4(m)}$, the regret of each player $k \in \range{n}$ is bounded by
    \begin{equation*}
        \reg_k^{(m)}(\rvx_k) \leq \frac{\kl{\rvx_k}{\vx_k^{(0)}}}{\eta} + \eta C' \log^5 (m).
    \end{equation*}
\end{theorem}

We remark that~\citep[Theorem 3.1]{Daskalakis21:Near} was stated slightly differently, but the statement we include here (\Cref{theorem:Dask}) follows readily from their analysis. Since $\kl{\cdot}{\cdot}$ is non-Lipschitz near the (relative) boundary of the strategy set, we initialize $\opthedge$ $\alpha$-away from the boundary in each task. Since the optimal learning rate will generally depend on the task similarity, we meta-learn the learning rate by running exponentially-weighted online optimization ($\EWOO$) \citep{hazan2007logarithmic} over a sequence of strongly convex regret-upper-bounds $U_1, \ldots, U_T$ (to be specified at a later point in the proof), similar to previous work in meta-learning (\emph{i.e.}, \citet{Khodak19:Adaptive, osadchiy2022online}). Specifically, we set the learning rate in task $t$ by 

\begin{equation*}
    \pv{\eta}{t} = \frac{\int_{\rho}^{\sqrt{D^2 + \rho^2 D^2}} \eta \exp(-\beta \sum_{s<t} U_s(\eta)) d\eta}{\int_{\rho}^{\sqrt{D^2 + \rho^2 D^2}} \exp(-\beta \sum_{s<t} U_s(\eta)) d\eta},
\end{equation*}
where $\beta = \frac{2}{D} \min\{1,\frac{\rho^2}{D^2}\}$, $\frac{\kl{\pv{\rtvy}{t}}{\pv{\vy}{t,0}}}{C' \log^5 (m)} \leq D^2, \forall t$, and $\rho > 0$.

Since our regret-upper-bounds are of the form $C' \log^5(m) \left( \eta + \frac{\kl{\pv{\rtvx}{t}}{\pv{\vx}{t,0}}}{C' \log^5(m) \eta} \right)$, they are non-smooth near zero and not strongly convex if $\pv{\rtvx}{t} = \pv{\vx}{t,0}$. Therefore, we run $\EWOO$ over the regularized sequence $U_1, \ldots, U_T$, where
\begin{equation*}
    U_t(\eta) = C' \log^5(m) \left( \eta + \frac{\frac{\kl{\pv{\rtvx}{t}}{\pv{\vx}{t,0}}}{C' \log^5(m)} + D^2 \rho^2}{\eta} \right).
\end{equation*}

We are now ready to present our main result for convergence to CCE.

\begin{theorem}
    \label{theorem:cce}
Let $\tvx := (1 - \alpha) \vx + \alpha \frac{1}{d} \mathbbm{1}_d$. When all players employ $\opthedge$ with $\pv{\vx}{t,0}_k = \frac{1}{t-1} \sum_{s<t} \pv{\rtvx}{s}$, learning rate given by $\EWOO$ \citep{hazan2007logarithmic} with suitably chosen hyperparameters, and $\alpha = \frac{1}{\sqrt{Tm}}$, there exist universal constants $C, C' >  0$ so that the average regret of each player $k \in \range{n}$ over a sequence of $T$ repeated games is bounded by
\begin{equation*}
\begin{aligned}
    \frac{1}{T} \sum_{t=1}^T \reg_k^{(t,m)} &\leq 2\sqrt{\frac{m}{T}}\\ &+ \sqrt{C' \log^5 (m) \log(d \sqrt{mT})} \left( \min \left\{ \frac{\sqrt{\log(d \sqrt{mT})}}{\eta^* \sqrt{T C' \log^5 (m)}}, \frac{1}{T^{1/4}} \right\} + \frac{1 + \log(T+1)}{2 \sqrt{T}} \right)\\ &+ \min_{0 < \eta \leq \Bar{\eta}} \left\{\eta C' \log^5 (m) + \frac{V_k^2}{\eta} + \frac{8d \sqrt{m}(1 + \log T)}{\eta \sqrt{T}} \right\},
\end{aligned}
\end{equation*}
where $\Bar{\eta} = \frac{1}{C n \log^4(m)}$ and $V_k^2 \defeq \frac{1}{T} \sum_{t=1}^T \kl{\pv{\rvx}{t}_k}{\bvx_k}$ for $\bvx_k \defeq \frac{1}{T} \sum_{t=1}^T \rvx_k^{(t)}$.%, and the $\alpha m$ term is the regret incurred from initializing away from the boundary of the strategy set.
\end{theorem}
\begin{proof}
By Theorem \ref{theorem:Dask},
\begin{align*}
    \frac{1}{T} \sum_{t=1}^T \reg_k^{(t,m)} &\leq 2\alpha m + \frac{1}{T} \sum_{t=1}^T \left(\pv{\eta}{t} C' \log^5 (m) + \frac{1}{\pv{\eta}{t}} \kl{\pv{\rtvx}{t}_k}{\pv{\vx}{t,0}_k} \right) \\
    &= 2\alpha m + \Delta_U + \frac{1}{T} \min_{0 < \eta \leq \Bar{\eta}} \left\{ \sum_{t=1}^T \left( \eta C' \log^5 (m) + \frac{1}{\eta} \kl{\pv{\rtvx}{t}_k}{\pv{\vx}{t,0}_k} \right) \right\}, \\
\end{align*}
where $\Delta_U := \frac{1}{T} \sum_{t=1}^T \pv{\eta}{t} C' \log^5 (m) + \frac{1}{\pv{\eta}{t}} \kl{\pv{\rtvx}{t}_k}{\pv{\vx}{t,0}_k} - \frac{1}{T} \min_{0 < \eta \leq \Bar{\eta}} \sum_{t=1}^T \eta C' \log^5 (m) $ $+ \frac{1}{\eta} \kl{\pv{\rtvx}{t}_k}{\pv{\vx}{t,0}_k}$.

\begin{equation*}
\begin{aligned}
    \frac{1}{T} \sum_{t=1}^T \reg_k^{(t,m)} &\leq 2\alpha m + \Delta_U + \frac{1}{T} \min_{0 < \eta \leq \Bar{\eta}} \left\{ \sum_{t=1}^T \left( \eta C' \log^5 (m) \right) + \min_{\vx_k \in \Delta^d}  \sum_{t=1}^T \left( \frac{1}{\eta} \kl{\pv{\rtvx}{t}_k}{\vx_k} \right)\right.\\
    &+ \left. \sum_{t=1}^T \left(\frac{1}{\eta} \kl{\pv{\rtvx}{t}_k}{\pv{\vx}{t,0}_k} \right) - \min_{\vx_k \in \Delta^d} \sum_{t=1}^T \left(\frac{1}{\eta} \kl{\pv{\rtvx}{t}_k}{\vx_k} \right) \right\}\\
    &\leq 2\alpha m + \Delta_U + \frac{1}{T} \min_{0 < \eta \leq \Bar{\eta}} \left\{ \sum_{t=1}^T \left(\eta C' \log^5 (m) \right) \right.\\ &+ \left. \min_{\vx_k \in \Delta^d} \sum_{t=1}^T \left(\frac{1}{\eta} \kl{\pv{\rtvx}{t}_k}{\vx_k}\right) + \frac{8d(1 + \log T)}{\eta \alpha} \right\}\\
    &= 2\alpha m + \Delta_U + \frac{1}{T} \min_{0 < \eta \leq \Bar{\eta}} \left\{ \sum_{t=1}^T \left(\eta C' \log^5 (m) \right) \right.\\ &+ \left. \sum_{t=1}^T \left(\frac{1}{\eta} \kl{\pv{\rtvx}{t}_k}{\Tilde{\bvx}_k} \right) + \frac{8d(1 + \log T)}{\eta \alpha} \right\}\\
    &\leq 2\alpha m + \Delta_U + \min_{0 < \eta \leq \Bar{\eta}} \left\{\eta C' \log^5 (m) \right.\\ &+ \left.\frac{1}{\eta T} \sum_{t=1}^T \left(\kl{\pv{\rvx}{t}_k}{\bvx_k}\right) + \frac{8d(1 + \log T)}{\eta \alpha T} \right\},\\
\end{aligned}
\end{equation*}
where the second inequality follows from~\citet[Lemma A.1]{balcan2022meta} with $S = \frac{d}{\alpha}$ and $K = 1$, the inequality follows from the fact that $\FTL$ over a sequence of Bregman divergences reduces to the average~\citep{Banerjee05:Clustering}, and the third inequality follows from the joint convexity of KL divergence.

Next, we bound $\Delta_U$ using the regret guarantees of $\EWOO$. In particular, we apply~\citet[Corollary C.2]{Khodak19:Adaptive}:

\begin{corollary}[\citet{Khodak19:Adaptive}]\label{cor:ewoo}
    Let $\{\pv{U}{t}: \mathbb{R}_+ \rightarrow \mathbb{R}\}_{t \geq 1}$ be a sequence of functions of form $\pv{U}{t}(\eta) = \left(\frac{(\pv{B}{t})^2}{\eta} + \eta \right) \pv{\gamma}{t}$ for any positive scalars $\pv{\gamma}{1}, \ldots, \pv{\gamma}{T} \in \mathbb{R}_+$ and adversarially chosen $B_t \in [0, D]$. Then the $\epsilon-\EWOO$ algorithm, for which $\epsilon > 0$ uses the actions of $\EWOO$ run on the functions $\Tilde{U}_t(\eta) = \left(\frac{(\pv{B}{t})^2 + \epsilon^2}{\eta} + \eta \right) \pv{\gamma}{t}$ over the domain $[\epsilon, \sqrt{D^2 + \epsilon^2}]$ to determine $\pv{\eta}{t}$ achieves regret
    \begin{equation*}
        \min \left\{\frac{\epsilon^2}{\eta^*}, \epsilon\right\} \sum_{t=1}^T \pv{\gamma}{t} + \frac{D \gamma_{max}}{2} \max \left\{\frac{D^2}{\epsilon^2}, 1\right\}(1 + \log(T+1))
    \end{equation*}
    for all $\eta^* > 0$.
\end{corollary}
Applying~\Cref{cor:ewoo} with $\epsilon = \rho D$, $\pv{\gamma}{t} = C' \log^5 (m), \; \forall t \in \range{T}$, $D = \frac{\sqrt{\log(d / \alpha)}}{\sqrt{C' \log^5 (m)}}$, and $\rho = \frac{1}{T^{1/4}}$ we see that
\begin{equation*}
\begin{aligned}
    \Delta_U &\leq DC' \log^5 (m) \left( \min \{ \frac{D}{\eta^* \sqrt{T}}, \frac{1}{T^{1/4}}\} + \frac{1 + \log(T+1)}{2 \sqrt{T}} \right)\\
    &=  \sqrt{C' \log^5 (m) \log(d / \alpha)} \left( \min \left\{ \frac{\sqrt{\log(d / \alpha)}}{\eta^* \sqrt{T C' \log^5 (m)}}, \frac{1}{T^{1/4}} \right\} + \frac{1 + \log(T+1)}{2 \sqrt{T}} \right).
\end{aligned}
\end{equation*}

Plugging in our upper bound for $\Delta_U$ into our task average regret-upper-bound, we obtain
\begin{equation*}
\begin{aligned}
    \frac{1}{T} \sum_{t=1}^T \reg_k^{(t,m)} &\leq 2\alpha m\\ &+ \sqrt{C' \log^5 (m) \log(d / \alpha)} \left( \min \left\{ \frac{\sqrt{\log(d / \alpha)}}{\eta^* \sqrt{T C' \log^5 (m)}}, \frac{1}{T^{1/4}} \right\} + \frac{1 + \log(T+1)}{2 \sqrt{T}} \right)\\ &+ \min_{0 < \eta \leq \Bar{\eta}} \left\{ \eta C' \log^5 (m) + \frac{1}{\eta T} \sum_{t=1}^T \left(\kl{\pv{\rvx}{t}_k}{\bvx_k}\right) + \frac{8d(1 + \log T)}{\eta \alpha T} \right\}.
\end{aligned}
\end{equation*}

By setting $\alpha = \frac{1}{\sqrt{mT}}$,

\begin{equation*}
\begin{aligned}
    \frac{1}{T} \sum_{t=1}^T \reg_k^{(t,m)} &\leq 2\sqrt{\frac{m}{T}}\\ &+ \sqrt{C' \log^5 (m) \log(d \sqrt{mT})} \left( \min \left\{ \frac{\sqrt{\log(d \sqrt{mT})}}{\eta^* \sqrt{T C' \log^5 (m)}}, \frac{1}{T^{1/4}} \right\} + \frac{1 + \log(T+1)}{2 \sqrt{T}} \right)\\ &+ \min_{0 < \eta \leq \Bar{\eta}} \left\{\eta C' \log^5 (m) + \frac{V_k^2}{\eta} + \frac{8d \sqrt{m}(1 + \log T)}{\eta \sqrt{T}} \right\},
\end{aligned}
\end{equation*}
where $V_k^2 \defeq \frac{1}{T} \sum_{t=1}^T \kl{\pv{\rvx}{t}_k}{\bvx_k}$.
\end{proof}

\iffalse
\begin{corollary}
    When all players employ $\opthedge$ with $\pv{\vx}{t,0}_k = \frac{1}{t-1} \sum_{s<t} \pv{\rtvx}{s}$ and $\alpha = \frac{1}{\sqrt{Tm}}$, the empirical frequency of play is a 
\end{corollary}
\fi 

It is worth noting that one can obviate meta-learning the learning rate using an uncoupled version of \emph{clairvoyant hedge}~\citep{Piliouras21:Optimal}, albeit with the need to appropriately average the iterates.

\subsection{Correlated Equilibria}
\label{appendix:CE}

Finally, we conclude this section with a refined meta-learning guarantee for converging to \emph{correlated equilibria}, an important strengthening of coarse correlated equilibria. As in the previous subsection, we focus on normal-form games. As it turns out~\citep{blum2007external}, converging to the set of correlated equilibria requires minimizing a stronger notion of regret, referred to as \emph{swap regret}:

\begin{equation}
    \label{eq:swapreg}
    \swreg_k^{(m)} \defeq \max_{\phi^{\star} \in \Phi} \left\{ \sum_{i=1}^m \langle \phi^{\star}(\vx_k^{(i)}) - \vx_k^{(i)}, \vu_k^{(t)} \rangle \right\},
\end{equation}
where $\Phi$ includes all possible linear functions $ \phi : \Delta(\calA_k) \to \Delta(\calA_k)$. Swap regret~\eqref{eq:swapreg} is clearly a stronger notion of hindsight rationality, since $\Phi$ includes all constants transformations. To minimize swap regret, \citet{blum2007external} gave a general reduction to the problem of minimizing external regret, which we briefly describe for the sake of completeness. 

\paragraph{The algorithm of Blum and Mansour} Each player $k$ maintains a separate regret minimizer for each available action $a \in \calA_k$. The next strategy is computed by obtaining the next strategy of each of those individual regret minimizers, and then determining any stationary distribution $\vx_k^{(i)}$ of the induced Markov chain---wherein each state corresponds to an action and the transition probabilities are given by the startegies of the individual regret minimizers. Finally, upon receiving as feedback a utility vector $\vu^{(i)}_k \in \R^{d_k}$, it suffices to forward to the regret minimizer associated with action $a$ the utility $ \vu_a^{(i)} \defeq \vx_k^{(i)}[a] \vu_k^{(i)}$, for all $a \in \calA_k$.

In our meta-learning setting, we will use as base learner $\OMD$ with the log-barrier as regularizer.  Recent work~\citep{anagnostides2022uncoupled} established that this leads to an RVU bound for swap regret. This will also bypass the need to meta-learn the learning rate, as we had to do in \Cref{theorem:cce}. We point out that the RVU bound for swap regret in~\citep{anagnostides2022uncoupled} can be readily expressed in terms of the initialization (\emph{e.g.}, see~\citep{Wei18:More}), leading to the following conclusion.

\begin{theorem}[\citep{anagnostides2022uncoupled}]
    \label{theorem:swaprvu}
    Suppose that each player $k \in \range{n}$ employs the no-swap-regret algorithm of~\citet{blum2007external} instantiated with $\OMD$ under a logarithmic regularizer $\calR_k : \vx_k \mapsto - \sum_{a \in \calA_k} \log \vx_k[a]$. There exists a sufficiently small learning rate $\eta > 0$ such that
     \begin{equation}
        \label{eq:rvuswap}
         \sum_{k=1}^n \swreg^{(m)}_k \leq \sum_{k=1}^n \sum_{a \in \calA_k} \reg_a^{(m)}(\rvx_a) \leq \frac{1}{\eta} \sum_{k=1}^n \sum_{a \in \calA_k} \brg{k}{\rvx_a}{\vx_a^{(0)}}.
     \end{equation}
\end{theorem}

In~\eqref{eq:rvuswap} we denote by $\reg_a^{(m)}$ the \emph{external} regret up to time $m \in \N$ of the regret minimizer associated with action $a \in \calA_k$ of player $k \in \range{n}$. We are now ready to show our refinement for swap regret in the meta-learning setting.

\begin{theorem}
    \label{theorem:ce}
    Let $\tvx_k := (1 - \alpha) \vx_k + \alpha \frac{1}{d} \mathbbm{1}_{d_k}$, and suppose that each player $k \in \range{n}$ employs the no-swap-regret algorithm of~\citet{blum2007external} instantiated with $\OMD$ under a logarithmic regularizer $\calR_k : \vx_k \mapsto - \sum_{a \in \calA_k} \log \vx_k[a]$ with a sufficiently small learning rate $\eta > 0$ and $\pv{\vx}{t,0}_a = \frac{1}{t-1} \sum_{s<t} \pv{\rtvx_a}{s}$, for all $a \in \calA_k$ and $t \geq 2$. Then, for $\alpha \defeq (m T)^{-1/3}$, 
    \begin{equation*}
        \frac{1}{T} \sum_{t=1}^T \swreg_k^{(t,m)} \leq \frac{1}{\eta} \sum_{k=1}^n \sum_{a \in \calA_k} \tilde{\hinsim}_{k,a}^2  + \frac{m^{2/3}}{T^{1/3}} \left( 2n + \sum_{k=1}^n \frac{8d_k^3 (1 + \log T)}{\eta} \right),
    \end{equation*}
    where $\tilde{\hinsim}_{k,a}^2 \defeq \frac{1}{T} \sum_{t=1}^T \brg{k}{\tilde{\rvx}^{(t)}_a}{\bar{\vx}_a}$ for $\bar{\vx}_a = \frac{1}{T} \sum_{t=1}^T \tilde{\rvx}_a^{(t)}$ is the task similarity for the individual regret minimizer of player $k \in \range{n}$ associated with action $a \in \calA_k$.
\end{theorem}

\begin{proof}
    By \Cref{theorem:swaprvu},
    \begin{align}
        \frac{1}{T} \sum_{t=1}^T \sum_{k=1}^n \swreg_k^{(t,m)} &\leq 2 \alpha m n + \frac{1}{\eta T} \sum_{t=1}^T \sum_{k=1}^n \sum_{a \in \calA_k} \brg{k}{\tilde{\rvx}^{(t)}_a}{\vx_a^{(t,0)}} \label{align:shiftreg} \\
        &\leq  2 \alpha m n + \frac{1}{\eta} \sum_{k=1}^n \sum_{a \in \calA_k} \frac{1}{T} \sum_{t=1}^T \brg{k}{\tilde{\rvx}^{(t)}_a}{\bar{\vx}_a} \notag \\
        &\qquad\qquad\qquad\qquad\qquad\qquad + \frac{1}{T} \sum_{k=1}^n \frac{8d_k^3(1 + \log T)}{\eta \alpha^2},\label{align:A1}
    \end{align}
    where \eqref{align:shiftreg} follows from \Cref{theorem:swaprvu} and the fact that for any player $k \in \range{n}$, 
    \begin{align*}
    \sum_{a \in \calA_k} \reg_a^{(t,m)}(\rvx^{(t)}_a) &\leq \sum_{a \in \calA_k} \reg_a^{(t,m)}(\tilde{\rvx}^{(t)}_a) + \alpha \sum_{i=1}^m \sum_{a \in \calA_k} \|\vu_a^{(t,i)}\|_\infty \|\rvx^{(t)}_a - \tilde{\rvx}^{(t)}_a\|_1 \\ 
    &\leq \sum_{a \in \calA_k} \reg_a^{(t,m)}(\tilde{\rvx}^{(t)}_a) + 2 \alpha \sum_{i=1}^m \sum_{a \in \calA_k} \vx_k^{(t,i)}[a] \\
    &\leq \sum_{a \in \calA_k} \reg_a^{(t,m)}(\tilde{\rvx}^{(t)}_a) + 2 \alpha m,
    \end{align*}
    since $\vu_a^{(t,i)} = \vx_k^{(t,i)}[a] \vu_k^{(t,i)}$ and $\|\vu_k^{(t,i)}\|_\infty \leq 1$ (by assumption), and \eqref{align:A1} uses~\citep[Lemma A.1]{balcan2022meta} with $K = 1$ and $S = \frac{d_k^2}{\alpha^2}$. Finally, the statement follows by taking $\alpha \defeq (m T)^{-1/3} $ and observing that $\sum_{k'=1}^n \swreg_{k'}^{(t,m)} \geq \swreg^{(t,m)}_k$, for any player $k \in \range{n}$ and task $t \in \range{T}$, since swap regret is always nonnegative.\footnote{Better rates may be possible by meta-learning the optimal boundary offset $\alpha$ as in~\citet[Theorem 5]{balcan2022meta}, \emph{i.e.}, by running multiplicative weights over a discretized grid of boundary offset values.}
\end{proof}
\section{Proofs from \texorpdfstring{\Cref{sec:stack}}{Section 3.3}: Meta-Learning in Stackelberg Games}
\label{appendix:stack}
\begin{theorem}\label{thm:stack-full}
Given a sequence of $T$ repeated Stackelberg security games with $d$ targets, $k$ attacker types, and within-game time-horizon $m$, running \texttt{MWU} over the set of extreme points $\cE$ as defined in \citeauthor{balcan2015commitment} with initialization $\pv{\vy}{t,0} = \frac{1}{t-1} \sum_{s<t} \pv{\rvy}{t}$ and learning rate given by $\beta$-EWOO \citep{hazan2007logarithmic} with suitably chosen hyperparameters achieves expected task-averaged Stackelberg regret
\begin{equation*}
\begin{aligned}
    \frac{1}{T} \sum_{t=1}^T \E [\pv{\sreg}{t,m}] &\leq 4 \sqrt{\frac{m}{T}}\\ &+ \sqrt{m \log(|\cE| \sqrt{mT})} \left( \min \left\{ \frac{\sqrt{\log(|\cE| \sqrt{mT})}}{\eta^* \sqrt{T m}}, \frac{1}{T^{1/4}} \right\} + \frac{1 + \log(T+1)}{2 \sqrt{T}} \right)\\ &+ \min_{0 < \eta \leq \Bar{\eta}} \left\{ \eta m + \frac{H(\bvy)}{\eta} + \frac{8|\cE|\sqrt{m}(\log T + 1)}{\eta \sqrt{T}} \right\}, \\ 
\end{aligned}
\end{equation*}
where the sequence of attackers in each task can be adversarially chosen.
\end{theorem}

\begin{proof}
    First, we have
    \begin{equation*}
    \begin{aligned}
        \pv{\sreg}{t,m} &= \sum_{i=1}^m \langle \pv{\rvx}{t}, \pv{\vu}{t}(b_{\pv{f}{t,i}}(\pv{\rvx}{t})) \rangle - \langle \pv{\vx}{t,i}, \pv{\vu}{t}(b_{\pv{f}{t,i}}(\pv{\vx}{t,i})) \rangle\\
        &\leq \max_{\pv{\vx}{t} \in \cE} \sum_{i=1}^m \langle \pv{\vx}{t}, \pv{\vu}{t}(b_{\pv{f}{t,i}}(\pv{\vx}{t})) \rangle - \langle \pv{\vx}{t,i}, \pv{\vu}{t}(b_{\pv{f}{t,i}}(\pv{\vx}{t,i})) \rangle + 2\gamma m, \\
    \end{aligned}
    \end{equation*}
    where the inequality follows from Lemma 4.3 of \citet{balcan2015commitment}. Let $\pv{\vy}{t,i} \in \Delta^{|\cE|}$ denote the distribution over actions the algorithm plays in game $t$ at time $i$ and $\pv{\vu}{t,i} = [\langle \vx, \pv{\vu}{t}(b_{\pv{f}{t,i}}(\vx)) \rangle]_{\vx \in \cE} \in [-1,1]^{|\cE|}$. Then,
    \begin{equation*}
    \begin{aligned}
        \E[\pv{\sreg}{t,m}] &\leq \sum_{i=1}^m \left(\langle \pv{\rvy}{t}, \pv{\vu}{t,i} \rangle - \langle \pv{\vy}{t,i}, \pv{\vu}{t,i} \rangle \right) + 2 \gamma m\\
        &= \sum_{i=1}^m \left(\langle \pv{\rvy}{t}, \pv{\vu}{t,i} \rangle - \langle \pv{\rtvy}{t}, \pv{\vu}{t,i}\rangle + \langle \pv{\rtvy}{t}, \pv{\vu}{t,i}\rangle - \langle \pv{\vy}{t,i}, \pv{\vu}{t,i} \rangle \right) + 2 \gamma m\\
        &\leq 2 \gamma m + 2\alpha m + \pv{\eta}{t} m + \frac{1}{\pv{\eta}{t}}\kl{\pv{\rtvy}{t}}{\pv{\vy}{t,0}},
    \end{aligned}
    \end{equation*}
    where $\pv{\rtvy}{t} := (1 - \alpha) \pv{\rvy}{t} + \frac{\alpha}{|\cE|} \mathbbm{1}_{|\cE|}$.
    As was the case in~\Cref{theorem:cce}, it is necessary to initialize $\alpha$-away from the boundary of the strategy space and meta-learn the learning rate using $\EWOO$ \citep{hazan2007logarithmic}. We refer the reader to the proof of~\Cref{theorem:cce} and references therein for more details about $\EWOO$.
    \begin{equation*}
    \begin{aligned}
        \frac{1}{T} \sum_{t=1}^T \E[\pv{\sreg}{t,m}] &\leq 2\alpha m + 2 \gamma m + \frac{1}{T} \sum_{t=1}^T \left(\pv{\eta}{t} m + \frac{\kl{\pv{\rtvy}{t}}{\pv{\vy}{t,0}}}{\pv{\eta}{t}} \right)\\
        &= 2\alpha m + 2 \gamma m + \frac{1}{T} \min_{\eta > 0} \left\{ \sum_{t=1}^T \left( \eta m + \frac{\kl{\pv{\rtvy}{t}}{ \pv{\vy}{t,0}}}{\eta} \right) \right\}\\ &
        + \frac{1}{T} \sum_{t=1}^T \left(\pv{\eta}{t} m + \frac{\kl{\pv{\rtvy}{t}}{\pv{\vy}{t,0}}}{\pv{\eta}{t}} \right)\\ - &\frac{1}{T} \min_{\eta > 0} \left\{ \sum_{t=1}^T \left(\eta m + \frac{\kl{\pv{\rtvy}{t}}{\pv{\vy}{t,0}}}{\eta} \right) \right\}\\
        &= 2\alpha m + 2 \gamma m + \Delta_U + \frac{1}{T} \min_{\eta > 0} \left\{\sum_{t=1}^T \left( \eta m + \frac{\kl{\pv{\rtvy}{t}}{\pv{\vy}{t,0}}}{\eta} \right) \right\}, \\
    \end{aligned}
    \end{equation*}    
    where $\Delta_U := \frac{1}{T} \sum_{t=1}^T \pv{\eta}{t} m + \frac{\kl{\pv{\rtvy}{t}}{\pv{\vy}{t,0}}}{\pv{\eta}{t}} - \frac{1}{T} \min_{0 < \eta \leq \Bar{\eta}} \sum_{t=1}^T \eta m + \frac{\kl{\pv{\rtvy}{t}}{\pv{\vy}{t,0}}}{\eta}$.
    
    \begin{equation*}
    \begin{aligned}
        \frac{1}{T} \sum_{t=1}^T \E[\pv{\sreg}{t,m}] &\leq 2\alpha m + 2 \gamma m + \Delta_U + \min_{\eta > 0} \left\{ \eta m + \frac{1}{T} \min_{\vy \in \Delta^{|\cE|}} \sum_{t=1}^T \frac{\kl{\pv{\rtvy}{t}}{\vy}}{\eta} \right.\\ 
        &+ \left. \frac{1}{T} \sum_{t=1}^T \frac{\kl{\pv{\rtvy}{t}}{\pv{\vy}{t,0}}}{\eta} - \frac{1}{T} \min_{\vy \in \Delta^{|\cE|}} \sum_{t=1}^T \frac{\kl{\pv{\rtvy}{t}}{ \vy}}{\eta} \right\}\\
        &\leq 2\alpha m + 2 \gamma m + \Delta_U\\ &+ \min_{\eta > 0} \left\{ \eta m  +  \frac{1}{T} \min_{\vy \in \Delta^{|\cE|}} \sum_{t=1}^T \left(\frac{\kl{\pv{\rtvy}{t}}{\vy}}{\eta}\right) + \frac{8|\cE|(\log T + 1)}{\eta \alpha T} \right\}\\ 
        &= 2\alpha m + 2 \gamma m + \Delta_U\\ &+ \min_{\eta > 0} \left\{ \eta m + \frac{1}{T} \sum_{t=1}^T \left(\frac{\kl{\pv{\rtvy}{t}}{ \Tilde{\Bar{\vy}}}}{\eta}\right) + \frac{8|\cE|(\log T + 1)}{\eta \alpha T} \right\}\\ 
        &\leq 2\alpha m + 2 \gamma m + \Delta_U\\ &+ \min_{\eta > 0} \left\{ \eta m + \frac{1}{T} \sum_{t=1}^T \left(\frac{\kl{\pv{\rvy}{t}}{ \bvy}}{\eta}\right) + \frac{8|\cE|(\log T + 1)}{\eta \alpha T} \right\}\\ 
        &= 2\alpha m + 2 \gamma m + \Delta_U + \min_{\eta > 0} \left\{ \eta m + \frac{H(\bvy)}{\eta} + \frac{8|\cE|(\log T + 1)}{\eta \alpha T} \right\}, \\ 
    \end{aligned}
    \end{equation*}
    where the second inequality follows from Lemma A.1 of \citet{balcan2022meta} with $S = \frac{|\cE|}{\alpha}$ and $K = 1$ and the first equality follows from the fact that $\FTL$ over a sequence of Bregman divergences reduces to the average \citep{Banerjee05:Clustering}.

Applying~\Cref{cor:ewoo} with $\epsilon = \rho D$, $\pv{\gamma}{t} = m, \; \forall t \in \range{T}$, $D = \frac{\sqrt{\log(|\cE| / \alpha)}}{\sqrt{m}}$, and $\rho = \frac{1}{T^{1/4}}$ we see that
\begin{equation*}
\begin{aligned}
    \Delta_U &\leq D m \left( \min \{ \frac{D}{\eta^* \sqrt{T}}, \frac{1}{T^{1/4}}\} + \frac{1 + \log(T+1)}{2 \sqrt{T}} \right)\\
    &=  \sqrt{m \log(|\cE| / \alpha)} \left( \min \left\{ \frac{\sqrt{\log(|\cE| / \alpha)}}{\eta^* \sqrt{T m}}, \frac{1}{T^{1/4}} \right\} + \frac{1 + \log(T+1)}{2 \sqrt{T}} \right)
\end{aligned}
\end{equation*}

Therefore, 
\begin{equation*}
\begin{aligned}
\frac{1}{T} \sum_{t=1}^T \E[\pv{\sreg}{t,m}] &\leq 2\alpha m + 2 \gamma m\\ &+ \sqrt{m \log(|\cE| / \alpha)} \left( \min \left\{ \frac{\sqrt{\log(|\cE| / \alpha)}}{\eta^* \sqrt{T m}}, \frac{1}{T^{1/4}} \right\} + \frac{1 + \log(T+1)}{2 \sqrt{T}} \right)\\ &+ \min_{0 < \eta \leq \Bar{\eta}} \left\{ \eta m + \frac{H(\bvy)}{\eta} + \frac{8|\cE|(\log(T) + 1)}{\eta \alpha T} \right\}. \\ 
\end{aligned}
\end{equation*}
Setting $\alpha = \gamma = \frac{1}{\sqrt{mT}}$ completes the proof.
\end{proof}
\section{Further Experimental Results}\label{appendix:experiments}

In this section, we present some additional experimental results we omitted earlier from~\Cref{sec:exp}. In particular,~\Cref{fig:more-eta,fig:more-eta-2} illustrate the task-averaged Nash equilibrium gap of $\OGD$ under different values for the learning rate.%, while \Cref{fig:zoom} highlights the Nash equilibrium gap for each task---\emph{not averaged}---for the first $25$ tasks of each endgame.

\begin{figure}[ht]
    \centering
    \includegraphics[width=\textwidth]{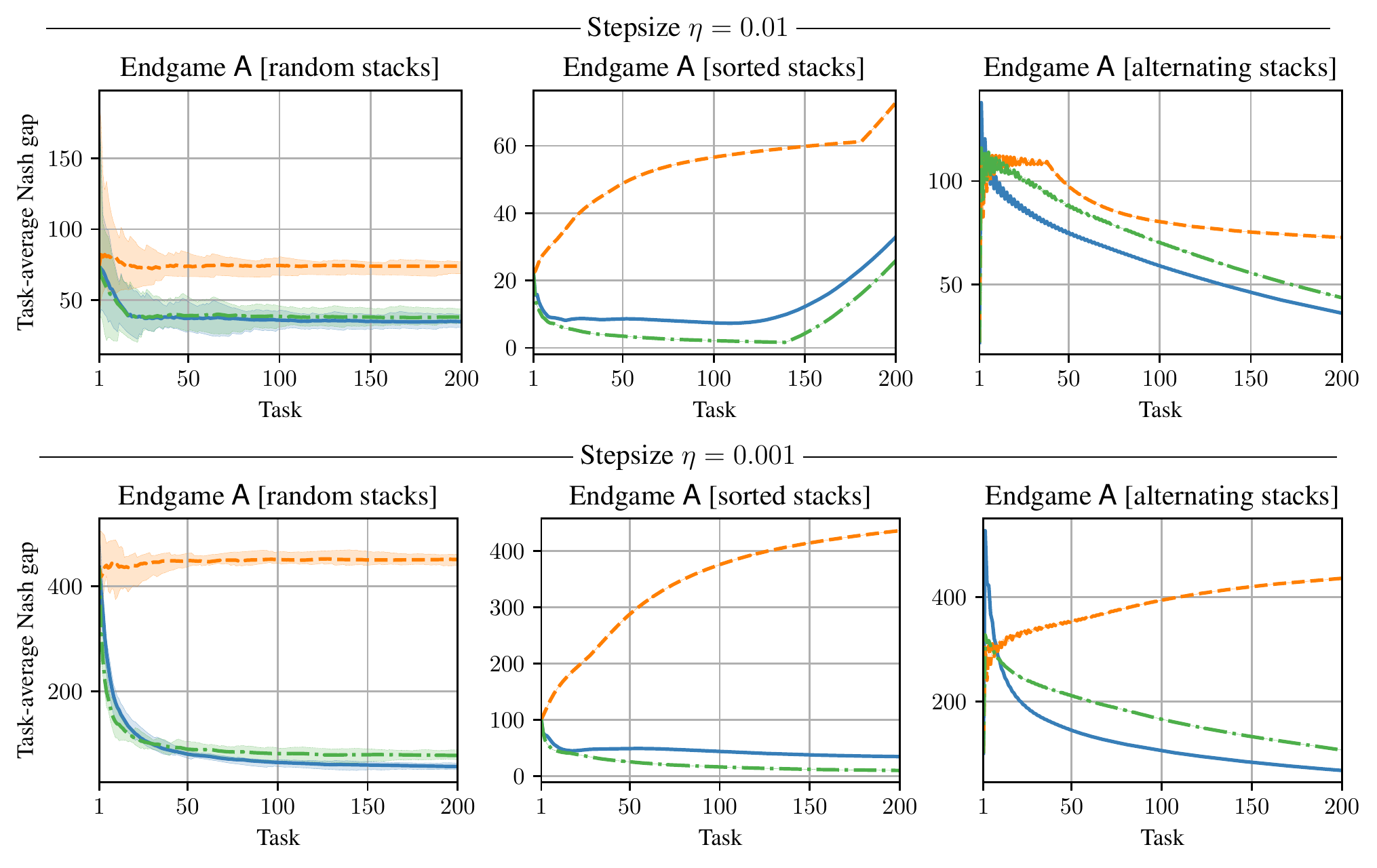}\\
    \makelegend
    \caption{Vanilla $\OGD$ versus our meta-learning versions of $\OGD$ for different values of the learning rate $\eta$ in Endgame A. Results for $\eta = 0.1$ are not included, as it was too large of a learning rate for any of the methods to converge in this endgame.}\label{fig:more-eta}
\end{figure}

\begin{figure}[ht]
    \centering
    \includegraphics[width=\textwidth]{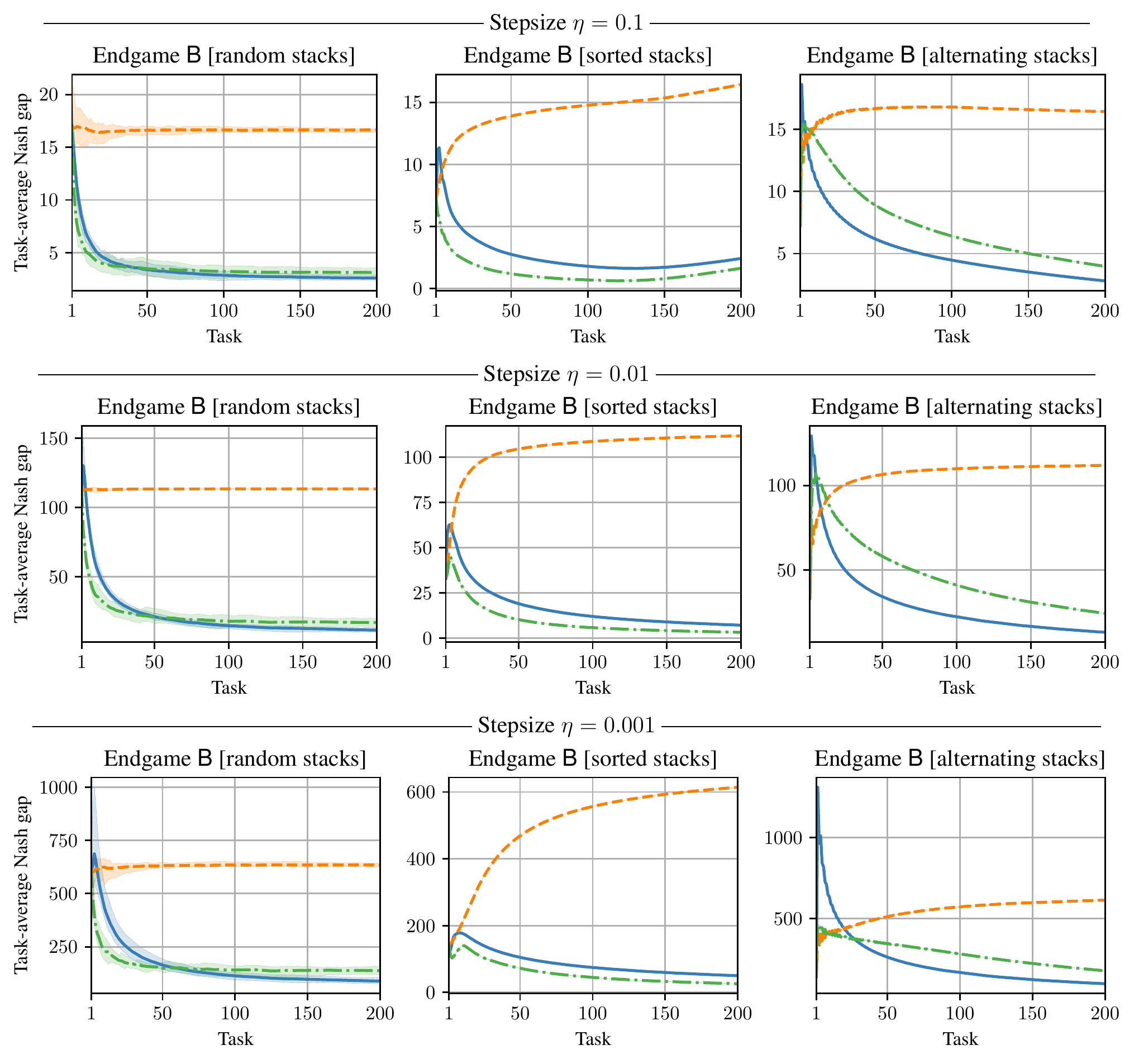}\\
    \makelegend
    \caption{Vanilla $\OGD$ versus our meta-learning versions of $\OGD$ for different values of the learning rate $\eta$ in Endgame B.}\label{fig:more-eta-2}
\end{figure}

% \begin{figure}[ht]
%     \label{fig:zoom}
%     \centering
%     \includegraphics[width=\textwidth]{plots/river_endgame_plots_nonaveraged_all.pdf}\\
%     \makelegend
%     \caption{Vanilla $\OGD$ versus our meta-learning version of $\OGD$ for different values of the learning rate $\eta$. Unlike \Cref{fig:more-eta}, here we plot the Nash equilibrium gap for each task without performing further averaging.}
% \end{figure}

\end{document}